\documentclass[twocolumn]{article}
\usepackage{physics}
\usepackage{amsfonts}
\usepackage{amsthm}
\usepackage{geometry}
\usepackage{tikz}
\usepackage{subcaption}
\usepackage{blkarray}

\newtheorem{theorem}{Theorem}

\newtheorem{corollary}[theorem]{Corollary}
\newtheorem{definition}[theorem]{Definition}
\newtheorem{proposition}[theorem]{Proposition}

\usepackage{authblk}

\usepackage{algorithm}
\usepackage[noend]{algpseudocode}

\usepackage[pdftitle=
{Optimizing Fermionic Encodings for both Hamiltonian and Hardware}]{hyperref}
\hypersetup{
    colorlinks=true
}

\usepackage{xcolor}

\newcommand\todoalt[1]{\textcolor{blue}{#1}}

\title{Optimizing Fermionic Encodings for both Hamiltonian and Hardware}
\author[1,2]{Riley W. Chien}

\author[1]{Joel Klassen}
\affil[1]{Phasecraft Ltd.}
\affil[2]{Dartmouth College}
\begin{document}

\twocolumn[
\begin{@twocolumnfalse}

\maketitle

\begin{abstract}

In this work we present a method for generating a fermionic encoding tailored to a set of target fermionic operators and to a target hardware connectivity. Our method uses brute force search, over the space of all encodings which map from Majorana monomials to Pauli operators, to find an encoding which optimizes a target cost function. In contrast to earlier works in this direction, our method searches over an extremely broad class of encodings which subsumes all known second quantized encodings that constitute algebra homomorphisms. In order to search over this class, we give a clear mathematical explanation of how precisely it is characterized, and how to translate this characterization into constructive search criteria. A benefit of searching over this class is that our method is able to supply fairly general optimality guarantees on solutions. A second benefit is that our method is, in principal, capable of finding more efficient representations of fermionic systems when the set of fermionic operators under consideration are faithfully represented by a smaller quotient algebra. Given the high algorithmic cost of performing the search, we adapt our method to handle translationally invariant systems that can be described by a small unit cell that is less costly. We demonstrate our method on various pairings of target fermionic operators and hardware connectivities. We additionally show how our method can be extended to find error detecting fermionic encodings in this class.

\end{abstract}

  \end{@twocolumnfalse}
  ]
\vspace{2cm}
\section{Introduction}
An important application of quantum computing is the modelling of systems where quantum physics dominates. Many such systems consist of fermions -- which are one of the two fundamental types of particle. Any simulation of fermions on a quantum computer requires a specification of how the fermions are represented on the memory of the device -- an encoding of fermions into qubits. These encodings must replicate the non-local anti-commutation relations of the fermionic operator algebra on the qubit system. This is a non-trivial task, with many possible solutions of which there is no a priori correct choice.

Many encodings have been invented (\cite{jordanwigner,bravyi2002fermionic,ball2005fermions,verstraete2005mapping,steudtner2018fermion,chen2018exact,setia2019superfast,jiang2019majorana,derby2021compact,derby2021compactalt,chen2022equivalence}), each with their own upsides and downsides which depend on the details of the model and of the hardware being used. Historically, encodings have not been designed for a particular model or piece of computing hardware. Instead they are often designed for a generic advantageous trait, such as low operator weight representations of certain interactions\footnote{Here by operator weight we mean the number of qubits on which operator has support}, efficient use of qubits, ease of state preparation, the ability to detect or correct errors, or merely simplicity. When one wishes to implement these encodings on specific hardware, one generally needs to think carefully about how the encoding can be made to ``fit'' into the device, given the particular algorithm -- generally with the aim of minimizing circuit depth and quantum memory use. This can be fairly non-trivial and it suggests that instead of using an out of the box encoding, it may be more fruitful to automatically generate a new encoding tailored to the specifics of the hardware and of the model.

Some work has been done in this direction. Ref. \cite{chiew2021optimal} explores finding optimal arrangements of the Jordan-Wigner transform onto a specific hardware geometry. This has the benefit of minimizing the operator weights of the terms in a given Hamiltonian, while not increasing the qubit overhead. However it may be desirable to make such a trade-off, in which case different methods are required in order to search over the space of fermionic encodings that employ ancillary qubits to reduce the operator weight. In Ref. \cite{bringewatt2022} custom codes for the purposes of parallelization are considered. Here the optimization is over the class of encodings described in \cite{chien2020custom}. However there are many potential encodings which may not fall within this class. 

In this work we present a method for generating  fermionic encodings tailored to the given hardware connectivity of a device, and to a given fermionic Hamiltonian. This method takes as input a set of fermionic operators, constituting the terms in the Hamiltonian, a graph specifying the possible two-qubit interactions of the hardware, and a maximum cost. The method finds the fermionic encoding with the least possible cost, or otherwise determines that no encoding exists with cost less than the specified maximum. Unlike earlier methods, ours searches in an exhaustive fashion over an extremely broad family of encodings, namely all encodings which would map Majorana monomials to Pauli operators. To our knowledge this family subsumes nearly all existing second quantized fermionic encodings -- with the notable exception of Ref. \cite{kirby2021second}\footnote{This exception differs from the other encodings in that it constitutes a representation of the unitary group of Majorana monomials, instead of an algebra homomorphism on the group algebra of the Majorana monomials, and leverages block encodings and linear combination of unitaries methods to simulate a unitary generated by the required element of the algebra.}.

The cost model we consider is as follows. Given an encoding, every term in the Hamiltonian is mapped to a qubit operator, which has support on a subset of qubits. In order to simulate the term, these qubits must be made to mutually interact. This requires interactions across at least all edges of a minimal Steiner tree (on the hardware graph) that contains these qubits. Thus an effective proxy for the circuit depth of simulating this term is the number of edges in this Steiner tree, and the cost of a given encoding is the maximum of this metric over all terms in the Hamiltonian.

The method is brute force and employs a branch-and-bound algorithm to search the space of possible encodings. The size of the search space grows very rapidly with the number of terms in the Hamiltonian and the size of the hardware graph. We introduce a number of strategies to reduce the size of this search space and exit branches early, which makes it possible to solve small problem sizes. For large problem sizes, it is likely that this method rapidly becomes too costly. 

However there is a specific use case where the size of the input specification remains small and where our method is likely to be useful. This is where both the system being modelled and the hardware graph each possess translational symmetry, and may be described concisely by a unit cell. Translational symmetry is ubiquitous in physics, but the domain where we expect this will find most use is in material science and condensed matter, where the majority of systems of interest possess translational symmetry. Furthermore, we have already observed that as superconducting quantum computers scale up, qubits are being arranged in a tiled pattern. We expect this trend to continue, as it lends itself well to large scale manufacturing. Since translational invariance is so important to this method, we incorporate it into all aspects of the work presented here. We employ the polynomial formalism popularized by Haah \cite{haah2013lattice,haah2016algebraic} to describe translationally invariant families of fermionic and qubit operators.

All existing encodings that employ additional qubits to reduce operator weights do so by way of projecting into a code space via the stabilizer formalism. The encodings produced by our method use the same strategy. An extra benefit of this strategy is that some encodings may also detect or even correct errors \cite{setia2019superfast,jiang2019majorana}, due to their high code distance. We show how our method can be made to restrict its search to fermionic encodings with a minimum code distance, and thus find low cost encodings with good error mitigating features. However high code distance is fundamentally at odds with the cost that we are trying to minimize in our method, and so for practical reasons we primarily concern ourselves with error detecting codes. 

We have employed our method to find optimal encodings for fermionic models on various geometries, mapped to a number of different translationally invariant hardware graphs, including many existing superconducting architectures. In some cases these optimal encodings are new, and in other cases they were already known. Importantly, even in the case where an encoding was known, our method certifies that the encoding chosen is the best choice for the given pairing (under our cost model).

\section{Preliminaries} 

\subsection{Mathematical Formalism of Fermionic Encodings}\label{sec:formalism}

We now review the mathematical formalism underpinning the kinds of fermionic encodings that our method is able to find. Nearly all known fermionic encodings fall within this formalism, with \cite{kirby2021second} being the only exception we are aware of. 
Employing the Majorana operators 
\begin{align}
    \gamma_j &:= a_j + a_j^{\dag}\\
    \overline{\gamma}_j &:= i(a_j^{\dag} - a_j),
\end{align}
where $a_j^\dag$ and $a_j$ are the standard fermionic creation and annihilation operators, we may consider the group of Majorana monomials
\begin{align}\label{majorana_group}
\nonumber \mathcal{M} := &\left\lbrace M(b):=\prod_{j=0}^{m-1} \gamma_j^{b_j} \overline\gamma_j^{b_{j+m}} \;\middle \vert \; b \in \{0,1\}^{2m}  \right\rbrace\times\\
 & \{\pm 1, \pm i \} 
\end{align}

Generally physicists are interested in an algebra of observables, i.e. the group algebra $\mathbb{C}[\mathcal{M}]$ on $\mathcal{M}$, or otherwise the group algebra $\mathbb{C}[F]$ on some subgroup $F \leq \mathcal{M}$ of $\mathcal{M}$. 
For example fermionic parity superselection restricts the algebra of observables of natural systems to be the group algebra on the even monomials 
\begin{equation} \mathcal{M}^+ :=\mathcal{M} \textrm{ with } |b|=0,
\end{equation}
where $ \vert b \vert := \sum_i b_i \textrm{ mod 2} $.
One may also be interested in restricting to smaller subgroups, for example if one is interested in conservation of a symmetry.

To construct an encoding of these observables into qubits, one must find an algebra isomorphism between the given group algebra and an algebra of observables on the qubit system. 
Since these are group algebras it suffices to find a faithful group representation that acts trivially on elements of the complex field (i.e. maps $-1$ to $-1$ and $i$ to $i$), which can then be extended to the group algebra by linearity of the representation.

We are often interested in engineering our encoding to satisfy specific features on some privileged subset $\mathcal{F} \subseteq F$ of the group. 
So we want a way to specify the form of the elements in $\mathcal{F}$, and then to build a faithful group representation of the rest of the group $F$ which is consistent with that form. The procedure for doing this is as follows.
For the sake of simplicity $F$ is assumed to be generated by $\mathcal{F}$.

Given a subset $\mathcal{F} \subseteq \mathcal{M}$ of Majorana monomials, the strategy is to find a mapping $\sigma: \mathcal{F} \rightarrow \mathcal{P}$ from $\mathcal{F}$ to the group of Pauli operators on $n$ qubits :
\begin{align}\label{pauli_group}
\nonumber \mathcal{P} := & \left\lbrace P(b):=\prod_{j=0}^{n-1} X_j^{b_j}Z_j^{b_{j+n}} \;\middle \vert \; b \in \{0,1\}^{2n}  \right\rbrace \times\\
&\{\pm 1, \pm i \}
\end{align}
such that $\sigma$ preserves the group commutation relations
\begin{equation}\label{commutation}
\forall \; f_i, f_j \in \mathcal{F}: \;[f_i,f_j]:=f_i^{-1}f_j^{-1}f_i f_j=\pm1
\end{equation}
and inverse/hermiticity relations
\begin{equation} \label{inverse}
 \forall f \in \mathcal{F} :\;  f^{-1}=f^\dagger = \pm f 
\end{equation}
 among elements of $\mathcal{F}$.
Note however, that $\sigma$ need not, and indeed typically does not, preserve any other product relations. Next we show that once a mapping of this form has been found, it uniquely specifies a group representation of the subgroup $F=\langle \mathcal{F} \rangle$ generated by $\mathcal{F}$, subject to one caveat.
 
Consider the finitely presented group
\begin{align}
\Gamma=&\langle \mathcal{F}\times \{\pm 1, \pm i\} \vert (\ref{commutation}), (\ref{inverse}) \rangle \\
=& \left\lbrace \Gamma(b):=\prod_{j=0}^{\vert \mathcal{F} \vert -1} f_j^{b_j} \;\middle \vert \; b \in \{0,1\}^{\vert \mathcal{F} \vert}  \right\rbrace \times \\
&\{\pm 1, \pm i \}
\end{align}
consisting of the free group on the set $\mathcal{F}\times \{\pm 1, \pm i\}$, subject to relations (\ref{commutation}) and (\ref{inverse}), but not any of the other product relations of $\mathcal{M}$. Most intuitively, one should think of the elements of $\Gamma$ as having forgotten that they were once products of Majorana operators, remembering only commutation/anticommutation relations and their inverse. The mapping $\sigma$ admits a natural extension from the set $\mathcal{F}$ to the group $\Gamma$ via the group structure of $\mathcal{P}$, and by construction is a group representation of $\Gamma$. 

There is a natural group homomorphism $\tau: \; \Gamma \rightarrow F$ which is defined by its action on the generating elements: $\tau(f_i \in \Gamma ):= f_i \in F$ and then extended to the full group $\Gamma$ via existing product relations in $F$, i.e. $\tau(f_i f_j):=\tau(f_i)\tau(f_j)$.
Furthermore, we have that $F \simeq \Gamma / \ker(\tau)$. 

In what follows we will make use of the notion of a ``descended'' representation. Given a group $A$, a representation $\phi$ and a subgroup $B < A$, if $B$ is in the kernel of $\phi$ then $\phi$ can be descended to a representation $\phi'$ of the quotient group $A/B$:
$$\phi'(aB):= \phi(a)$$ which is consistently defined since $\phi(B)=1$. For ease of exposition, we will drop the distinction between $\phi$ and its descended representation $\phi'$, which can be inferred implicitly from the particular group it acts upon. So for example we may say $\phi$ is a representation of $A/B$.

In order to construct a group representation of $F$, we must identify the abelian subgroup:
\begin{equation}
S:=\sigma(\ker(\tau)) \subseteq \mathcal{P}
\end{equation}
which, as long as $-1\not \in S$, is a stabilizer group with stabilizer code space $V$.
This allows us to define the projection of $\sigma$ onto the subspace $V$
\begin{equation}\label{projected_sigma}
 \sigma_V := \Pi_V \circ \sigma
\end{equation} 
which remains a representation of $\Gamma$, since $\Pi_V$ commutes with the group $\sigma(\Gamma)$. Furthermore, since $\sigma_V(\ker(\tau))=1$, it is also a descended representation of $\Gamma/\ker(\tau) \simeq F$.

Intuitively, the group $\ker{\tau}$ describes which products of elements in $\mathcal{F}$ should be yielding the identity once they remember that are were constructed from Majorana operators. By projecting into the code space $V$ we ensure that this product structure is preserved.

Finally, we need to ensure our representation is faithful -- i.e. there is a one-to-one correspondence between qubit operations on the code space and logical fermionic operations. It can be shown (see Appendix \ref{sigmav_faithful}) that the representation $\sigma_V$ is a faithful representation of $F$ if and only if $\ker(\sigma) \subseteq \ker(\tau)$. If this is not the case, then there is a group of observables
\begin{equation}
G:=\tau(\ker(\sigma)) \subseteq F 
\end{equation} whose value becomes fixed to $+1$ when projecting into $V$. We refer to such a group $G$ as a superselection group. In this case, $\sigma_V$ is instead a faithful representation of the quotient group $F/G$  (see Appendix \ref{sigmav_faithful}).

In summary, if we can map a set of fermionic observables $\mathcal{F}$ to Paulis which satisfy the group commutation relations (\ref{commutation}) and inverse relations (\ref{inverse}) among the elements of $\mathcal{F}$, then we can identify a stabilizer group $S \subseteq \mathcal{P}$ and a superselection group $G \subseteq F:=\langle \mathcal{F}\rangle$ such that we can extend this mapping $\sigma$ to an algebra isomorphism on the group algebra $\mathbb{C}[F/G]$ by projecting into the stabilizer code space $V$ of $S$, fixing the superselection group $G$ to be $+1$. The one caveat is that the stabilizer group $S$ may not contain $-1$. 

\subsubsection{Binary Representation}\label{sec:binary}

A nice feature of both the Pauli group $\mathcal{P}$ and the Majorana monomial group $\mathcal{M}$ is that up to a phase their elements may be represented by a binary vector (as made explicit in definitions (\ref{majorana_group}) and (\ref{pauli_group})). Furthermore, the product relations among elements in these groups correspond directly to addition of the corresponding binary vectors -- again up to a phase. 

Thus it is useful to specify the projective portion of the maps $\tau$ and $\sigma$ as matrices over a binary field, and postpone the specification of the phase e.g.:
\begin{equation}
\hat{\tau} = 
\begin{blockarray}{cccccc}
&f_1 & f_2 & f_3 & \cdots & f_{|\mathcal{F}|} \\
\begin{block}{c(ccccc)}
  \gamma_0 & 0 & 1 & 1 & \cdots & 0 \\
 \gamma_1  & 1 & 0 & 0 & \cdots & 0 \\
 \vdots  & \vdots & \vdots & \vdots & \ddots & \vdots  \\
 \gamma_{m-1}  & 1 & 1 & 1 & \cdots & 1\\
\bar{\gamma}_0  & 0 & 1 & 0 & \cdots & 0 \\
\bar{\gamma}_1  & 1 & 1 & 1 & \cdots & 1\\
 \vdots  & \vdots & \vdots & \vdots & \ddots & \vdots \\
\bar{\gamma}_{m-1}  & 1 & 1 & 1 & \cdots & 1 \\
\end{block}
\end{blockarray} 
\end{equation}

\begin{equation}
\hat{\sigma} = 
\begin{blockarray}{cccccc}
&f_1 & f_2 & f_3 & \cdots & f_{|\mathcal{F}|} \\
\begin{block}{c(ccccc)}
  X_0 & 1 & 0 & 0 & \cdots & 0 \\
 X_1  & 1 & 0 & 1 & \cdots & 0 \\
 \vdots  & \vdots & \vdots & \vdots & \ddots & \vdots  \\
 X_{n-1}  & 0 & 0 & 1 & \cdots & 0\\
Z_0  & 1 & 0 & 1 & \cdots & 1 \\
Z_1  & 0 & 1 & 0 & \cdots & 0\\
 \vdots  & \vdots & \vdots & \vdots & \ddots & \vdots \\
Z_{n-1}  & 0 & 0 & 0 & \cdots & 1 \\
\end{block}
\end{blockarray} 
\end{equation}
Where the columns $\hat{\tau}_i$ and $\hat{\sigma}_i$ are given by
$$\tau(f_i) = \delta_\tau(f_i)M( \hat{\tau}_i )  \;,\; \sigma(f_i)=  \delta_\sigma(f_i) P( \hat{\sigma}_i ) $$
where $\delta_\tau$ and $\delta_\sigma$ are as yet unspecified phases.

Writing the maps in this way allows us to easily represent the group commutation relations as a matrix equation, and also to easily identify the stabilizer group $S$ and superselection group $G$.

For two Pauli operators $P(a), P(b)$ specified by $a,b \in \mathbb{F}_2^{2n}$, their group commutation relation is given by,
\begin{equation} 
    [P(a),P(b)] = (-1)^{\omega_q(a,b)}
\end{equation}
where $\omega_q$ is a binary symplectic form 
\begin{equation} \label{eq:pauli_symplectic}
    \omega_q(a,b) = a^\dag \Lambda_q b,
    \quad \Lambda_q = 
    \begin{pmatrix}
    0 & 1\\
    1 & 0
    \end{pmatrix} \otimes I_{n\times n}.
\end{equation}

In the case of Majoranas, if we restrict $\mathcal{F}$ to be a subset of the even Majorana monomials $\mathcal{M}^+$, then for two Majorana operators $M(a), M(b)$ specified by $a,b \in \mathbb{F}_2^{2m}$, their group commutation relation is given by
\begin{equation}
   [M(a),M(b)] = (-1)^{\omega_f(a,b)}
\end{equation}   
Using the quadratic form $\omega_f(a,b) = a^\dag b$ works for even parity Majorana operators, however it fails for odd parity operators. To correct this, the products of the Hamming weights mod 2 should be added.

In order to have quadratic form that behaves properly also for odd operators, we will use a matrix in the quadratic form
\begin{equation}\label{eq:fermi_symplectic}
   \omega_f(a,b) = a^\dag \Lambda_f b, \quad \Lambda_f = I + C_1
\end{equation}
where $I$ is the $2n\times 2n$ identity matrix and $C_1$ is the constant matrix with $1$ in every entry. It can be easily verified that $a^{\dag}C_1 b = |a||b|$, so we then have 
\begin{equation}
    \omega_f(a,b) = a^{\dag}b + |a||b|.
\end{equation}

Thus, we may say that $\sigma$ preserves the group commutation relations (\ref{commutation}) if and only if the matrix equation
\begin{equation}\label{matrix_comm_cond}
\hat{\tau}^\dag \Lambda_f \hat{\tau} = \hat{\sigma}^\dag \Lambda_q \hat{\sigma}
\end{equation}
is satisfied. In almost all practical use cases $\mathcal{F} \subseteq M^+$.

Next, in order to ensure that the map $\sigma$ satisfies the inverse/hermiticity relations (\ref{inverse}), we must additionally specify the phases, $\delta_\tau$ and $\delta_\sigma$. The phase $\delta_\tau$ is assumed to be given explicitly by the choice of $f_i$, for example one may want the elements in $\mathcal{F}$ to be some subset of the edge and vertex operators:
\begin{align}
    V_j &= -i\gamma_j\overline{\gamma}_j\quad \text{(vertex)}\label{eq:vertex_op}\\
    E_{jk} &= -i \gamma_j\gamma_k\quad \text{(edge)}\label{eq:edge_op}.
\end{align}
which appear frequently in local fermionic interactions and may be used to generate $\mathcal{M}^+$. Here the phase is explicit. Since Pauli operators square to the identity, and $f_i^2 \in \pm 1$, the inverse condition $\sigma(f_i)^2 = f_i^2$ is satisfied when the choice of phase $\delta_\sigma$ satisfies
\begin{align}\label{delta_sigma_cond}
\delta_\sigma(f_i)^2 = f_i^2
\end{align}
This fixes whether $\delta_\sigma(f_i)$ is complex or real, but it leaves free a choice of sign, which may be useful when we consider the superselection group, $G$, and the stabilizer group, $S$.

Given $\hat{\tau}$ and $\delta_\tau$, and a choice of $\hat{\sigma}$  and $\delta_\sigma$ satisfying conditions (\ref{matrix_comm_cond}) and (\ref{delta_sigma_cond}), we may now compute the stabilizer group $S = \sigma(\ker(\tau))$ and the superselection group $G=\tau(\ker(\sigma))$. It is not difficult to show (see Appendix \ref{proof_ker_sigma_tau}) that
\begin{equation}\label{ker_sigma}
\ker(\sigma) = \{\sigma(\Gamma(b))^* \Gamma(b) \vert b \in \ker(\hat{\sigma}) \}
\end{equation}
and similarly
\begin{equation}\label{ker_tau}
\ker(\tau) = \{\tau(\Gamma(b))^* \Gamma(b) \vert b \in \ker(\hat{\tau}) \}
\end{equation}
where we note that $\sigma(\Gamma(b))\in \{\pm1, \pm i \}$ when $b \in \ker(\hat{\sigma})$ and $\tau(\Gamma(b)) \in \{\pm1, \pm i \}$ when $b \in \ker(\hat{\tau})$. In other words the kernels of the maps $\sigma$ and $\tau$ are in one-to-one correspondence with the kernels of the matrices $\hat{\sigma}$ and $\hat{\tau}$ respectively, up to a uniquely specified and easily computable phase. The kernel of a binary matrix can be efficiently computed by Gaussian elimination on an augmented matrix. In the case where the matrix elements are in a polynomial ring over a binary field as will be required when we consider translationally invariant systems, the kernel may also be computed using known Grobner basis methods (see Appendix \ref{sec:kernel}). Thus it is possible to compute the stabilizer group $S$ and the superselection group $G$ in either case. 

Having computed $S$, it may turn out that a particular choice of sign convention for $\delta_\sigma$ implies that $-1 \in S$. In this case one can always modify the choice of sign for $\delta_\sigma$ such that $S$ no longer contains $-1$. The procedure for doing this is described in Appendix \ref{sec:stab_group_minus}. %

In addition to resolving any problems with $-1\in S$, the choice of sign convention in $\delta_\sigma$ also has an impact on the sign of the elements in $G$. Thus one has some freedom to choose the particular signs of the elements in $G$. For example if the fermionic parity observable is in $G$, then one may be able to choose if the code space represents even or odd  fermion parity by a judicious choice of sign convention in $\delta_\sigma$.

\subsubsection{Translationally Invariant Systems}

In many cases it may be useful to find an encoding that tessellates space. So it is helpful to extend the binary formalism discussed in Section \ref{sec:binary} to handle a concise representation of translationally invariant fermionic operators, Pauli operators, and mappings between them. The formalism that follows was first introduced by Haah \cite{haah2013lattice} and it was also used in the context of mappings between Majorana and Pauli operators in \cite{tantivasadakarn2020jordan}.

In this case we employ the Laurent polynomial formalism wherein we represent Pauli operators, Majorana monomials, and elements of the free group $\Gamma$, and the mappings $\sigma$ and $\tau$ by vectors and matrices over multi-variable Laurent polynomials, the spaces of which we denote by 
\begin{multline}
\mathbb{F}_2[x_1,..,x_D]^{n \times m}:= \\ \left\lbrace \sum_{k \in \mathbb{Z}^D} a_{k} \prod_{i=1}^D x_i^{k_i} \;,\; a_k \in \mathbb{F}_2^{n \times m} \right \rbrace.
\end{multline} 
Here the number of variables is equal to the dimensionality $D$ of the system, and the degree $k_i$ need not be positive. In order to make the algebra nicer it is also useful to employ the convention that $a^\dagger =  \sum_{k \in \mathbb{Z}^D} a_{k}^{\mathsf{T}} \prod_{i=1}^D x_i^{-k_i}$.  Finally vector and matrix multiplication is defined by inheriting the product relations of the polynomials:
$$ ab =\sum_{p,q \in \mathbb{Z}^D} a_{p}b_q \prod_{i=1}^D x_i^{p_i+q_i}.$$ We introduce the notation $ T_k := \prod_{i=1}^D x_i^{k_i}$, to be understood as a translation in the lattice by vector $k$.

Intuitively the elements $a_k$ denote the presence or absence of particular terms in a monomial, in the same way as in the binary formalism, at lattice position $k$. Thus we may extend the definitions of the functions $P$, $\Gamma$ and $M$ so that they are maps from the space $\mathbb{F}_2(x_1,..,x_D)^n$ into the tessellated versions of the Pauli operators $\mathcal{P}$, the free group $\Gamma$, and the Majorana monomials $\mathcal{M}$ (respectively) in the following fashion:
\begin{align}
P(a)&:= \prod_{k \in \mathbb{Z}^D} \prod_{j=0}^{n-1} X_{j,k}^{a_k[j]} Z_{j,k}^{a_k[j+n]}, \\
\Gamma(a)&:= \prod_{k \in \mathbb{Z}^D} \prod_{j=0}^{J-1} f_{j,k}^{a_k[j]},  \\
M(a)&:= \prod_{k \in \mathbb{Z}^D} \prod_{j=0}^{m-1} \gamma_{j,k}^{a_k[j]}\bar{\gamma}_{j,k}^{a_k[j+m]}.
\end{align}
Here $a= \sum_{k \in \mathbb{Z}^D} a_{k} T_k$, the multi index $(j,k)$ indicates the $j$th qubit, mode or privileged element $f_j$ acting on a unit cell of the tessellation translated from the origin by $k$, and $a_k[j]$ is the $j$th entry of $a_k$. Here each unit cell contains $n$ qubits, $J$ privileged operators  or $m$ modes (respectively). The products are always ordered in accordance with any preferred ordering on the vectors $k$. 

The symplectic forms given by equations (\ref{eq:pauli_symplectic}) and (\ref{eq:fermi_symplectic}) may be extended to $\mathbb{F}_2(x_1,..,x_D)^n$  in the following fashion
\begin{equation}
\Lambda_Q :=  \Lambda_q 
\end{equation}
\begin{equation}
\Lambda_F :=I+ \sum_{k \in \mathbb{Z}^D} C_1 T_k
\end{equation}
recalling that $C_1$ is the constant matrix with $1$ in every entry. The expression $a^\dagger \Lambda_{Q/F} b$ is an element of $\mathbb{F}_2(x_1,..,x_D)$, and the quadratic form between $a$ and $b$ defining the appropriate group commutation relations is given by the zeroth term in the expression:
 \begin{equation}
 w_{Q/F}(a,b):= (a^\dagger \Lambda_{Q/F} b)_{\vec{0}}.
 \end{equation}
The other terms contain information about the group commutation relations among various relative translations of $a$ and $b$:
$$ (a^\dagger \Lambda_{Q/F} b)_{k}= (a^\dagger \Lambda_{Q/F} b T_k)_{0}= w_{Q/F}(a,b T_k).$$
Although $\Lambda_F$ constitutes an unbounded expression, when we consider even Majorana operators it suffices to truncate the expression, since two terms acting on disjoint sets of modes will always commute with one another. 

Intuitively it should be clear that the formalism described in Section \ref{sec:binary} should extend to the translationally invariant case described here. More concretely, and analogously to Section \ref{sec:binary}, we would like that given a matrix $\hat{\tau} \in \mathbb{F}_2(x_1,..,x_D)^{2m \times J}$ which describes the projective mapping from the group $\Gamma$ to $\mathcal{M}$, a fermionic encoding is completely specified, up to signs, by a matrix $\hat{\sigma}$ over $\mathbb{F}_2(x_1,..,x_D)^{2n \times J}$ describing the projective mapping from $\Gamma$ to $\mathcal{P}$ satisfying the relation:
\begin{equation}\label{matrix_comm_cond_TI}
\hat{\tau}^\dag \Lambda_F \hat{\tau} = \hat{\sigma}^\dag \Lambda_Q \hat{\sigma}
\end{equation}
 Instead of proving this in detail here, we simply note that a solution $\hat{\sigma}$ of this kind uniquely specifies an encoding of the kind  discussed in Section \ref{sec:binary} for any region $R\subset \mathbb{Z}^D $ of a finite size, by tesselating $\hat{\sigma}$ and $\hat{\tau}$ over $R$:  
 $$\hat{\sigma}_R := \left( \hat{\sigma}T_{r_1}, \hat{\sigma}T_{r_2}, ..., \hat{\sigma} T_{r_{|R|}}\right) \;,\; r_i \in R$$  
-- similarly for $\hat{\tau}_R$ -- and 
 mapping every column $a$ of $\hat{\sigma}_R$ or $\hat{\tau}_R$ into a binary vector via:
\begin{equation}
\sum_{k \in \mathbb{Z}^D} a_{k} \prod_{i=1}^D x_i^{k_i} \rightarrow \bigoplus_{k \in R} a_{k} 
\end{equation} 
where here a canonical ordering on $k$ is chosen.

\subsubsection{Example: fermions on a square lattice}

We will illustrate our application of the formalism above to describing the encoding of a model of spinless fermions on a square lattice. There is a single complex fermionic mode in each unit cell. We will encode this system onto a qubit system with two qubits per unit cell.

The algebra of fermionic operators can be generated by a single vertex operator (\ref{eq:vertex_op}) acting on a reference cell as well as edge operators (\ref{eq:edge_op}) acting between the mode in the reference cell and the neighboring cells in the $x$- and $y$-directions as well as their translations. We collect these three operators acting on the two species of Majoranas per unit cell into the $2\times 3$ matrix $\hat{\tau}$,
\begin{equation}\small
    \hat{\tau} = \begin{pmatrix}
    V_1 & E_{1,x} & E_{1,y}
    \end{pmatrix}_f
    = \begin{pmatrix}
    1 & 1+x & 1+y \\
    \hline
    1 & 0 & 0
    \end{pmatrix}.
\end{equation}

The commutation relations amongst the generating operators is encoded in the entries in the matrix

\begin{multline}
    \hat{\tau}^{\dag}\Lambda_F\hat{\tau} = \\
      \setlength{\arraycolsep}{1.7pt}
  \renewcommand{\arraystretch}{0.62}
    \begin{pmatrix}
    0 & 1+x & 1+y\\
    1+x^{\text{-}1} & x+ x^{\text{-}1} & 1+y+x^{\text{-}1}+yx^{\text{-}1}\\
    1+y^{\text{-}1} & 1+y^{\text{-}1} + x + y^{\text{-}1}x & y+y^{\text{-}1}
    \end{pmatrix}
\end{multline}
We can reproduce a known encoding here for fermions on a square lattice with two qubits per unit cell. Here $V_1 \rightarrow Z_0 Z_1$  The encoded generating operators can be chosen as

\begin{equation}
    \hat{\sigma} = \begin{pmatrix}
    \tilde{V}_1 & \tilde{E}_{1,x} & \tilde{E}_{1,y}
    \end{pmatrix}_q
    = \begin{pmatrix}
    0 & x & y \\
    0 & 1 & 1 \\
    \hline
    1 & 1 & 1+y\\
    1 & 1 & 0
    \end{pmatrix}.
\end{equation}

It is straightforward to verify that this matrix $\hat{\sigma}$ satisfies the encoding equation
\begin{equation}\label{eq:TIcondition}
\hat{\sigma}^{\dag}\Lambda_Q \hat{\sigma} = \hat{\tau}^\dag \Lambda_F \hat{\tau}.
\end{equation}
Therefore the Pauli operators specified by the columns of $\hat{\sigma}$ provide a valid encoding of the Majorana operators given as the columns of $\hat{\tau}$ for any sized square lattice.

\subsection{Symmetries and Fermionic Subalgebras}

Many physical systems of interest have symmetries that restrict the form their Hamiltonians may take. In some cases it is possible to capture these symmetries by restricting to a subalgebra of the full fermionic algebra. This is most obvious in the case of parity supersymmetry, which restricts the algebra to even products of Majoranas. Such symmetries loosen the constraints on the encodings, since it may be the case that the encoding need not represent the full algebra, for example the even Majorana algebra may be completely generated by edge and vertex operators, defined as:
\begin{equation}
E_{ij} := -i \gamma_{i} \gamma_{j} \;,\; V_i := -i \gamma_{i} \bar{\gamma}_{i}
\end{equation}
Thus single Majorana operators need not admit a representation (see \cite{derby2021compact} for an example of this). This further motivates designing encodings tailored to representing the specific terms in the Hamiltonian. If the terms make manifest the symmetries implicit in the Hamiltonian then by searching for encodings which represent only those terms, it may be possible to leverage these symmetries. We would like to briefly comment on two more notable symmetries that can be made manifest in the Hamiltonian terms.

In many systems in condensed matter, such as the Fermi-Hubbard model, the total fermion number is a conserved quantity. Expressed in terms of creation/annihilation operators, each operator must consist of an equal number of creation and annihilation operators. Edge operators on the other hand, when expressed in terms of creation/annihilation operators, are decomposed into four quadratic terms with $0$, $1$, and $2$ creation operators. The kinetic term in fermion Hamiltonians is the hopping term which we can express in terms of edge and vertex operators,
\begin{equation}
    a_j^{\dag}a_k + a_k^{\dag}a_j \propto V_j E_{jk} + E_{jk}V_k.
\end{equation}

In this case, we know that whenever edge operators appear in our Hamiltonian, they will always appear as multiplied by a vertex operator as above. We can then explicitly look for low cost representations of both of the terms in the above sum. We may then expect to find more efficient encodings which would not admit a representation of individual edge or vertex operators. However the careful reader will note that this would preclude the inclusion of any number operators, $n_i = (I-V_i)/2$, in the algebra, so this would rule out including external fields to the Hamiltonian or the possibility of measuring single mode fermion density. We would note however that pairs of vertex operators would necessarily be included in the representation, and so the algebra may include two point density-density interactions. This is important since the Fermi-Hubbard model includes a spin-spin interaction. In fact the Fermi-Hubbard model has the additional spin conservation symmetry. In this case there are no hopping terms between the two different spin sectors. This further reduces the size of the algebra.

We include here the privileged fermionic operators $\hat{\tau}$ associated with the spinful Fermi-Hubbard model on the square lattice, with two modes per unit cell, one for spin up and one for spin down:
\begin{equation}
    \hat{\tau} = \begin{pmatrix}
    1 & x & 1 & y & 1 & 0 & 0 & 0 & 0 \\ 
    1 & 0 & 0 & 0 & 0 & x & 1 & y & 1 \\ 
    \hline
    1 & 1 & x & 1 & y & 0 & 0 & 0 & 0 \\ 
    1 & 0 & 0 & 0 & 0 & 1 & x & 1 & y 
    \end{pmatrix}.
\end{equation}
We note that a single edge operator within a spin sector would correspond to a column vector non-zero in only one entry -- which is not in the column space of this matrix. Unfortunately, because of the number of privileged terms, we were unable to find example codes in a reasonable time frame for this set of operators.

\subsection{Existence of an encoding for a given fermionic unit cell}

Given a translationally-invariant fermionic system that one would like to encode and a choice of unit cell, it is useful to ask when we can be guaranteed that an encoding exists if we allow for arbitrarily high operator weight. When not considering strictly translationally-invariant systems, one can always be sure that some mapping exists as long as the number of qubits is equal to the number of number of fermionic modes being encoded. This commonly achieved by simply associating the fermionic modes to sites on a chain and using the Jordan-Wigner transformation to map the fermionic system to a nonlocally-interacting spin system.

In the two dimensional translationally-invariant setting when we are hoping to find a mapping from local even-parity fermionic operators to local qubit operators, we are not free to use Jordan-Wigner as this mapping breaks locality. One can however, always find an encoding that preserves locality and translation-invariance by choosing a hardware unit cell containing one more qubit than there are fermionic modes per cell, assuming the hardware graph is connected.

This can straightforwardly be achieved by associating the $n$ modes within each unit cell with vertices of a path graph, (an $n$-site open chain). Within each unit cell, select one of the vertices, say the $0$th vertex, and add edges between the $0$th vertices in each cell and the one above it in the $y$ direction. Now, add an edge between $0$th vertices in each cell and the $n-1$-th vertex in its neighboring cell in the $x$ direction. The resulting graph is a square lattice with horizontal edges decorated by $n-1$ additional vertices. Using the generalized superfast encoding due to Setia, et al.\cite{setia2019superfast} applied to the constructed graph will provide an encoding of the $n$ modes in each fermionic system unit cell onto $n+1$ qubits in the hardware cells. The resulting operators are local with respect to the lattice of unit cells and are of weight depending on the number of modes within each unit cell but not the total system size. Thus, as long as there are $n+1$ qubits in each unit cell and the translationally-invariant hardware graph is connected, there will exist an encoding of the desired fermionic system.

\section{Description of the Method}

Having established the requisite mathematical language to be able to state the problem, our method for solving it follows fairly straightforwardly. Given the $2m \times J$ matrix $\hat{\tau}$ describing the fermionic operators one wishes to represent, and a function giving the cost of an individual Pauli operator, we need only populate a $2n \times J$ matrix $\hat{\sigma}$ in such a way that equation (\ref{eq:TIcondition}) is satisfied and that the greatest cost of any column (Pauli operator) is minimized. Our strategy for doing this is a brute force branch-and-bound search algorithm, supplemented by some optimization tricks. Of course it is likely that this is not the most efficient way to do this -- ours is merely a proof of principle. We proceed with a high level concrete description of the algorithm, and afterwards will discuss some of the details of its implementation.

Let $A_j$ be the $j$th column of a matrix $A$. Let $[A]_j$ be the principal minor of $A$ consisting of all rows and columns less than or equal to $j$. Let $\mathcal{P}_n$ be an iterator over the $n$ qubit Paulis, expressed as column vectors in $\mathbb{F}_2[x_1,... x_D]^{2n}$. Let {\it cost}() be a cost function evaluating the cost of an individual Pauli. The branch-and-bound method is described in Algorithm \ref{branchandbound}.

\begin{algorithm}[h]
\caption{Branch and Bound Search Algorithm for Finding Fermionic Encodings}\label{branchandbound}
\begin{algorithmic}[1]
\State $\hat{\sigma}=0$, best cost = $\infty$, best candidate=none
\Procedure{Branch And Bound}{$j$}
	\If{$j>J$} 
		\State best candidate = $\hat{\sigma}$	
		\State best cost = max cost over columns in $\hat{\sigma}$
		\State return
	\EndIf
	
	\For{pauli in $\mathcal{P}_n$}
 		\State $\hat{\sigma}_j$ = pauli
		\State valid = $[\hat{\tau}]_j^\dag \Lambda_F [\hat{\tau}]_j == [\hat{\sigma}]_j^\dag \Lambda_Q [\hat{\sigma}]_j$
		\State bounded= {\it cost}(pauli) $\leq$ best cost
 		\If{valid and bounded}
 			\State Branch And Bound($j+1$)
 		\EndIf
 	\EndFor	
\EndProcedure

\State Branch And Bound(0)
\State return best candidate

\end{algorithmic}
\end{algorithm}
In words, the algorithm recursively populates the matrix $\hat{\sigma}$ with Pauli operators that satisfy the (anti-)commutation conditions of all previously chosen Pauli operators, and stores the best completed $\hat{\sigma}$ as it goes. It ignores any Pauli operators that are guaranteed to yield a worse cost. Critically the total cost of an encoding is taken to be the max cost of any of the chosen Pauli operators. If this were not the case, then the bounding condition would fail.

\subsection{Optimizations}
Clearly this method is not efficient either in the number of privileged elements $J$, or in the number of qubits $n$. As $n$ grows, the number of Pauli operators grows exponentially, and as $J$ grows the depth of the search tree grows, exponentially compounding the size of the search space. The hope is that by keeping the size of $J$ and $n$ small by considering translationally invariant unit cells, and by making judicious optimizations of the algorithm, we may at least be able to search for encodings in practically useful scenarios.

The primary point of optimization is in how the iterator over Paulis $\mathcal{P}_n$ is ordered and pruned. Ideally, we would like for $\mathcal{P}_n$ to be ordered by cost. In this case, the algorithm can terminate iterating over $\mathcal{P}_n$ when the cost of a Pauli exceeds the best cost, and the best solution is the first solution found. In the small instances we have considered, we have found that it is possible to cache a pre-ordered list of Pauli operators, which dramatically improves the run-time of the algorithm. However, as the number of qubits grows this begins to become infeasible. It is possible to dynamically generate Paulis ordered by weight (number of qubits on which the Pauli has support), and for the cost function we consider the Pauli weight lower-bounds the cost. Thus, in the case where we can not pre-order the Paulis by cost we can at least terminate iterating over $\mathcal{P}_n$ when the Pauli weight exceeds the best cost.

If the number of Paulis is small enough to pre-order by cost, then it is also possible to pre-compute look-up tables for all operators which anti-commute and all operators which commute with a given Pauli operator. With these look-up tables one may prune $\mathcal{P}_n$ to only contain Paulis which satisfy the requisite (anti-)commutation relations -- specified by the prior choices of Paulis -- by taking intersections of these sets. In conjunction with an ordering of the cost this means the search algorithm need only check if the intersections are empty, or otherwise take the first element in the list. Given that the (anti-)commutation relations are highly structured, we conjecture that look-up tables may not need to be precomputed, and instead it may be possible to dynamically generate a pruned $\mathcal{P}_n$ containing only valid Paulis for the given circumstance.

In the case of translationally invariant systems, we may also prune $\mathcal{P}_n$ so that it only includes Paulis with support on the central unit cell (i.e. the column contains at least one entry with a $1$ in it). Furthermore we can bound the range of the support of the Paulis to for example only extend to nearest or next nearest neighbouring cells. 

Further minor pruning of $\mathcal{P}_n$ includes the obvious step of setting a max Pauli weight and the more complicated step of avoiding iterating over solutions that are equivalent up to single qubit Clifford operations.

\subsection{Cost Function}

We would like to tailor the encoding we find to a particular use case. This is generally enforced by the cost function. The algorithm as described above makes use of one essential feature, which is that the total cost is given by the max cost over the chosen Paulis. Further optimizations make use of the fact that the cost function is lower bounded by the weight of the Paulis, but this is not essential. Otherwise we are free to set the cost function as we wish. 

For the purposes of demonstration, we consider a concrete cost function motivated by the following consideration. Although Pauli weight of operators is natural, and a convenient choice of cost for the purposes of search optimization, it is generally divorced from the details of any hardware. Generally we wish to lower the Pauli weight of our logical fermionic operators in order to more efficiently execute unitary operations generated by those operators. This is central to dynamical simulations using product formulas for time-evolution as well as in variational quantum algorithms for ground state problems with, for example, the Hamiltonian variational ansatz \cite{cerezo2021variational,wecker2015progress}.

In order to perform a unitary generated by a multi-qubit Pauli operator, the standard circuit is to apply single qubit Clifford gates that, upon conjugation, map the Pauli operator to a product of $Z$s with the same support. Then a series of CNOT gates is applied to collect parities onto a single qubit to which a single qubit rotation is then applied. The parity collection and single-qubit Clifford gates are then undone. The depth of the described circuit block is determined by the number of two-qubit gates which must be performed. 

The qubits in some quantum computing platforms, for example superconducting circuits, are subject to connectivity constraints. There is a fixed graph with qubits at vertices and edges connecting qubits that can be jointly acted upon by two-qubit gates. 

Suppose an operator acted on two qubits that were not adjacent on the graph. The circuit that performs a rotation generated by such an operator must act on all the qubits in a path terminating on the associated vertices. As such, the cost we associate to the operator is not the weight, but rather the number of edges in the shortest path connecting the two vertices. We would like to generalize this notion to higher weight Pauli operators. Specifically, for a Pauli operator acting on a collection of qubits associated to vertices of a graph, we define the cost of the operator to be the number of edges in the Steiner tree for the marked vertices on the hardware graph. Recall that the Steiner tree is the minimal connected tree subgraph containing all the marked vertices.

Note that this cost is a more restrictive notion than weight. In particular, if an operator in the Pauli group on $n$ qubits, $p \in \mathcal{P}_n$, has a weight of $\text{weight}(p)$, then the cost of the operator, $\text{cost}(p)\geq \text{weight}(p)-1$. Therefore, 
\begin{equation}
    \{p \in \mathcal{P}_n | \text{cost}(p) \leq w-1\} \subseteq \{p\in \mathcal{P}_n  | \text{weight}(p) \leq w\}.
\end{equation}

In principle, further refinements on the cost function could be included to reflect important details in the hardware. For example, the edges of the hardware connectivity graph could be weighted with the quality of the 2-qubit interaction on the device. However this may not play nicely with the translationally invariant formalism we are employing here.

\subsection{Optimality of the Solution}

The method described here has the added advantage that any encoding found is certified to be optimal with respect to the chosen cost by the nature of the brute force search -- provided the algorithm is allowed to terminate, which may not always be practical.  However, this optimality is subject to the following caveats when we consider translationally invariant systems.

Firstly, the method takes as input a unit cell for both the fermionic lattice system as well as the hardware layout. Importantly, the choice of unit cell to provide as input is generically not unique. For example, on a square lattice one could choose a single-mode unit cell with $x$- and $y$-direction translations along the rows and columns of the lattice. However, one could just as well have chosen a two-mode unit cell or indeed any number of modes in each unit cell. 
Furthermore, once a unit cell is chosen, one also has to specify how the unit cell tiles the plane. In particular, the choice of lattice translations is not in general unique. For example, suppose we fix a two-qubit unit cell for the hardware lattice with the qubits in adjacent columns and the same row. We can choose the unit $x$-translation to map each column to the column two to the right and the unit $y$-translation to map each row to the row above. We could also choose the unit $x$-translation to be one spacing to the right and one down and the unit $y$-translation to be one spacing to the right and one up.

Both of these types of choices can have an effect on the outcome of an encoding search. Thus there may be encodings associated with a different choice of unit cell that have a better cost than the one found. It may even be the case that for one choice of unit cell a solution does not exist for the specified search domain whereas for another choice of unit cell on the same graph a solution does exist. See Section \ref{sec:new_codes} for an example of this.

The second caveat is that in principle one may be able to find better encodings if one loosens the restriction on the neighbourhood within which Pauli operators are allowed to act, or if one loosens the constraint that Pauli operators need act on the central unit cell. However this seems unlikely in general.

\section{Results from Applying the Method}
\label{sec:Results}

We used the encoding generator to find encodings of each fermionic system listed in Table \ref{tab:fermionic_systems}. The operators we specified to be found were edge and vertex operators in all cases except one. On the square lattice we also search for operators of the form $E_{ij}V_j$ and $V_j E_{ij}$. From these operators, one can immediately construct particle number preserving hopping operators. 
The device layouts we consider include the square, hexagonal, and triangular lattices, as well as layouts currently used in superconducting devices, namely the heavy hexagon layout employed by IBM and the truncated square tiling employed by Rigetti. We also considered a number of other lattices including the Kagome lattice with 2 different choices of unit cell. The greatest vertex degree amongst all the lattices we considered was 6 for the rhombile lattice. Finally, we searched for encodings on a small number of instances of layouts with non-planar connectivity. The largest unit cell in terms of number of qubits was $5$. All of the lattices and their associated unit cells are shown in Table \ref{tab:hardware graphs}.

\begin{table*}
    \centering
    \begin{tabular}{c|c|c|c}
    Lattice (modes/cell) & unit cell & operators & $\hat{\tau}$ \\
    \hline
    Square (1) & \scalebox{0.3}{\begin{tikzpicture}
    \foreach \x in {0,...,1}{
      \foreach \y in {0,...,1}{
        
        \draw (0.5,-0.8) -- (0.5,3.8);
        \draw (2.5,-0.8) -- (2.5,3.8);
        \draw (-0.8,0.5) -- (3.8,0.5);
        \draw (-0.8,2.5) -- (3.8,2.5);
        
        \node[draw,circle,inner sep=2pt,fill] at (2*\x+0.5,2*\y+0.5) {};
        
        \draw [gray,dashed] (-0.5,-0.8) -- (-0.5,3.8);
        \draw [gray,dashed] (1.5,-0.8) -- (1.5,3.8);
        \draw [gray,dashed] (-0.8,1.5) -- (3.8,1.5);
        \draw [gray,dashed] (-0.8,-0.5) -- (3.8,-0.5);
        \draw [gray,dashed] (-0.8,3.5) -- (3.8,3.5);
        \draw [gray,dashed] (3.5,-0.8) -- (3.5,3.8);
      }
    }
\end{tikzpicture}} & ($E$ and $V$) & $\begin{pmatrix}
    1+y & 1+x  & 1 \\
    \hline
    0    & 0   & 1 
\end{pmatrix}$ 
    \\
    \hline
    Square (Num. preserving) (1) & \scalebox{0.3}{\begin{tikzpicture}
    \foreach \x in {0,...,1}{
      \foreach \y in {0,...,1}{
        
        \draw (0.5,-0.8) -- (0.5,3.8);
        \draw (2.5,-0.8) -- (2.5,3.8);
        \draw (-0.8,0.5) -- (3.8,0.5);
        \draw (-0.8,2.5) -- (3.8,2.5);
        
        \node[draw,circle,inner sep=2pt,fill] at (2*\x+0.5,2*\y+0.5) {};
        
        \draw [gray,dashed] (-0.5,-0.8) -- (-0.5,3.8);
        \draw [gray,dashed] (1.5,-0.8) -- (1.5,3.8);
        \draw [gray,dashed] (-0.8,1.5) -- (3.8,1.5);
        \draw [gray,dashed] (-0.8,-0.5) -- (3.8,-0.5);
        \draw [gray,dashed] (-0.8,3.5) -- (3.8,3.5);
        \draw [gray,dashed] (3.5,-0.8) -- (3.5,3.8);
      }
    }
\end{tikzpicture}} & (hopping and $V$) & $\begin{pmatrix}
    y & 1 & x & 1 & 1 \\
    \hline
    1 & y & 1 & x & 1 
\end{pmatrix}$
    \\
    \hline
    Triang (1) & \scalebox{0.3}{\begin{tikzpicture}
    \foreach \x in {0,...,1}{
      \foreach \y in {0,...,1}{
        
        \draw (0.5,-0.8) -- (0.5,3.8);
        \draw (2.5,-0.8) -- (2.5,3.8);
        \draw (-0.8,0.5) -- (3.8,0.5);
        \draw (-0.8,2.5) -- (3.8,2.5);
        \draw (-0.8,3.2) -- (-0.2,3.8);
        \draw (-0.8,1.2) -- (1.8,3.8);
        \draw (-0.8,-0.8) -- (3.8,3.8);
        \draw (1.2,-0.8) -- (3.8,1.8);
        \draw (3.2,-0.8) -- (3.8,-0.2);
        
        \node[draw,circle,inner sep=2pt,fill] at (2*\x+0.5,2*\y+0.5) {};
        
        \draw [gray,dashed] (-0.5,-0.8) -- (-0.5,3.8);
        \draw [gray,dashed] (1.5,-0.8) -- (1.5,3.8);
        \draw [gray,dashed] (-0.8,1.5) -- (3.8,1.5);
        \draw [gray,dashed] (-0.8,-0.5) -- (3.8,-0.5);
        \draw [gray,dashed] (-0.8,3.5) -- (3.8,3.5);
        \draw [gray,dashed] (3.5,-0.8) -- (3.5,3.8);
      }
    }
\end{tikzpicture}} & ($E$ and $V$) &
    $\begin{pmatrix}
    1+xy & 1+y & 1+x & 1 \\
    \hline
    0 & 0 & 0 & 1
\end{pmatrix}$ 
    \\
    \hline
    Spinful sq (2) & \scalebox{0.3}{\begin{tikzpicture}
    \foreach \x in {0,...,1}{
      \foreach \y in {0,...,1}{
        
        \draw (-0.8,0) -- (3.8,0);
        \draw (-0.8,1) -- (3.8,1);
        \draw (-0.8,2) -- (3.8,2);
        \draw (-0.8,3) -- (3.8,3);
        \draw (0,-0.8) -- (0,3.8);
        \draw (1,-0.8) -- (1,3.8);
        \draw (2,-0.8) -- (2,3.8);
        \draw (3,-0.8) -- (3,3.8);
        \draw (0,1) -- (1,0);
        \draw (0,3) -- (1,2);
        \draw (2,1) -- (3,0);
        \draw (2,3) -- (3,2);
        \node[draw,circle,inner sep=2pt,fill] at (2*\x+1,2*\y) {};
        \node[draw,circle,inner sep=2pt,fill] at (2*\x,2*\y+1) {};
        
        \node[align=left] at (0.3,1+0.3) {\Large\textcolor{orange}{1}};
        \node[align=left] at (1+0.3,0.3) {\Large\textcolor{orange}{2}};
        
        \draw [gray,dashed] (-0.5,-0.8) -- (-0.5,3.8);
        \draw [gray,dashed] (1.5,-0.8) -- (1.5,3.8);
        \draw [gray,dashed] (-0.8,1.5) -- (3.8,1.5);
        \draw [gray,dashed] (-0.8,-0.5) -- (3.8,-0.5);
        \draw [gray,dashed] (-0.8,3.5) -- (3.8,3.5);
        \draw [gray,dashed] (3.5,-0.8) -- (3.5,3.8);
      }
    }
\end{tikzpicture}} & ($E$ and $V$) &
    $\begin{pmatrix}
    0 & 0 & 0 & 1+y & 1+x & 1 \\
    1+y & 1+x & 1 & 0 & 0 & 0 \\
    \hline
    0 & 0 & 0 & 0 & 0 & 1 \\
    0 & 0 & 1 & 0 & 0 & 0 
\end{pmatrix}$
    \\
    \hline
    Hex (2) & \scalebox{0.3}{\begin{tikzpicture}
    \foreach \x in {0,...,1}{
      \foreach \y in {0,...,1}{
        
        \draw (2.2,3.8) -- (3.8,2.2);
        
        \draw (-0.8,2.8) -- (2.8,-0.8);
        
        \draw (0.2,3.8) -- (3.8,0.2);
        
        \draw (-0.8,0.8) -- (0.8,-0.8);
        
        \draw (2*\x,2*\y) -- (2*\x+1,2*\y+1);
        
        \node[draw,circle,inner sep=2pt,fill] at (2*\x,2*\y) {};
        \node[draw,circle,inner sep=2pt,fill] at (2*\x+1,2*\y+1) {};
        
        \node[align=left] at (-0.3,-0.3) {\Large\textcolor{orange}{1}};
        \node[align=left] at (1+0.3,1+0.3) {\Large\textcolor{orange}{2}};
        
        \draw [gray,dashed] (-0.5,-0.8) -- (-0.5,3.8);
        \draw [gray,dashed] (1.5,-0.8) -- (1.5,3.8);
        \draw [gray,dashed] (-0.8,1.5) -- (3.8,1.5);
        \draw [gray,dashed] (-0.8,-0.5) -- (3.8,-0.5);
        \draw [gray,dashed] (-0.8,3.5) -- (3.8,3.5);
        \draw [gray,dashed] (3.5,-0.8) -- (3.5,3.8);
      }
    }
\end{tikzpicture}} & ($E$ and $V$) & $\begin{pmatrix}
    y & x & 1 & 0 & 1  \\
    1 & 1 & 1 & 1 & 0  \\
    \hline
    0 & 0 & 0 & 0 & 1  \\
    0 & 0 & 0 & 1 & 0  
\end{pmatrix}$
    \\
    \hline
    Tilted sq (2) & \scalebox{0.3}{\begin{tikzpicture}
    \foreach \x in {0,...,1}{
      \foreach \y in {0,...,1}{
        
        \draw (-0.8,0.2) -- (2.8,3.8);
        \draw (0.2,-0.8) -- (3.8,2.8);
        \draw (1.8,-0.8) -- (-0.8,1.8);
        \draw (3.8,-0.8) -- (-0.8,3.8);
        \draw (3.8,-0.8) -- (-0.8,3.8);
        \draw (3.8,1.2) -- (1.2,3.8);
        \draw (-0.8,2.2) -- (0.8,3.8);
        \draw (2.2,-0.8) -- (3.8,0.8);
        \node[draw,circle,inner sep=2pt,fill] at (2*\x+1,2*\y) {};
        \node[draw,circle,inner sep=2pt,fill] at (2*\x,2*\y+1) {};
        
        \node[align=left] at (0.3,1+0.3) {\Large\textcolor{orange}{1}};
        \node[align=left] at (1+0.3,0.3) {\Large\textcolor{orange}{2}};
        
        \draw [gray,dashed] (-0.5,-0.8) -- (-0.5,3.8);
        \draw [gray,dashed] (1.5,-0.8) -- (1.5,3.8);
        \draw [gray,dashed] (-0.8,1.5) -- (3.8,1.5);
        \draw [gray,dashed] (-0.8,-0.5) -- (3.8,-0.5);
        \draw [gray,dashed] (-0.8,3.5) -- (3.8,3.5);
        \draw [gray,dashed] (3.5,-0.8) -- (3.5,3.8);
      }
    }
\end{tikzpicture}}  & ($E$ and $V$) & $\begin{pmatrix}
    1 & x & 1 & 1 & 0 & 1 \\
    x^{-1}y & 1 & y & 1 & 1 & 0 \\
    \hline
    0 & 0 & 0 & 0 & 0 & 1 \\
    0 & 0 & 0 & 0 & 1 & 0 
\end{pmatrix}$
    \\
    \hline
    Kagome (3) & \scalebox{0.3}{\begin{tikzpicture}
    \foreach \x in {0,...,1}{
      \foreach \y in {0,...,1}{
        
        \draw (0,-0.8) -- (0,3.8);
        \draw (2,-0.8) -- (2,3.8);
        \draw (-0.8,1) -- (3.8,1);
        \draw (-0.8,3) -- (3.8,3);
        \draw (-0.8,-0.8) -- (3.8,3.8);
        \draw (-0.8,1.2) -- (1.8,3.8);
        \draw (1.2,-0.8) -- (3.8,1.8);
        
        \node[draw,circle,inner sep=2pt,fill] at (2*\x,2*\y) {};
        \node[draw,circle,inner sep=2pt,fill] at (2*\x,2*\y+1) {};
        \node[draw,circle,inner sep=2pt,fill] at (2*\x+1,2*\y+1) {};
        
        \node[align=left] at (0+0.3,0) {\Large\textcolor{orange}{1}};
        \node[align=left] at (1,1-0.3) {\Large\textcolor{orange}{2}};
        \node[align=left] at (0+0.3,1+0.3) {\Large\textcolor{orange}{3}};
        
        \draw [gray,dashed] (-0.5,-0.8) -- (-0.5,3.8);
        \draw [gray,dashed] (1.5,-0.8) -- (1.5,3.8);
        \draw [gray,dashed] (-0.8,1.5) -- (3.8,1.5);
        \draw [gray,dashed] (-0.8,-0.5) -- (3.8,-0.5);
        \draw [gray,dashed] (-0.8,3.5) -- (3.8,3.5);
        \draw [gray,dashed] (3.5,-0.8) -- (3.5,3.8);
      }
    }
\end{tikzpicture}} & ($E$ and $V$) &
    $\begin{pmatrix}
    0 & 0 & xy & 1 & y & 1 & 0 & 0 & 1 \\
    x & 1 & 0 & 0 & 1 & 1 & 0 & 1 & 0 \\
    1 & 1 & 1 & 1 & 0 & 0 & 1 & 0 & 0 \\
    \hline
    0 & 0 & 0 & 0 & 0 & 0 & 0 & 0 & 1 \\
    0 & 0 & 0 & 0 & 0 & 0 & 0 & 1 & 0 \\
    0 & 0 & 0 & 0 & 0 & 0 & 1 & 0 & 0 
\end{pmatrix}$
    \\
    \hline
    Kagome alt. cell (3) & \scalebox{0.3}{\begin{tikzpicture}
    \foreach \x in {0,...,1}{
      \foreach \y in {0,...,1}{
        
        \draw (1,-0.8) -- (1,3.8);
        \draw (3,-0.8) -- (3,3.8);
        \draw (-0.8,0) -- (3.8,0);
        \draw (-0.8,2) -- (3.8,2);
        \draw (-0.8,2.8) -- (2.8,-0.8);
        \draw (3.8,0.2) -- (0.2,3.8);
        \draw (0.8,-0.8) -- (-0.8,0.8);
        \draw (3.8,2.2) -- (2.2,3.8);
        
        \node[draw,circle,inner sep=2pt,fill] at (2*\x,2*\y) {};
        \node[draw,circle,inner sep=2pt,fill] at (2*\x+1,2*\y) {};
        \node[draw,circle,inner sep=2pt,fill] at (2*\x+1,2*\y+1) {};
        
        \node[align=left] at (-0.1,0.4) {\Large\textcolor{orange}{1}};
        \node[align=left] at (1-0.3,0.3) {\Large\textcolor{orange}{2}};
        \node[align=left] at (1-0.4,1) {\Large\textcolor{orange}{3}};
        
        \draw [gray,dashed] (-0.5,-0.8) -- (-0.5,3.8);
        \draw [gray,dashed] (1.5,-0.8) -- (1.5,3.8);
        \draw [gray,dashed] (-0.8,1.5) -- (3.8,1.5);
        \draw [gray,dashed] (-0.8,-0.5) -- (3.8,-0.5);
        \draw [gray,dashed] (-0.8,3.5) -- (3.8,3.5);
        \draw [gray,dashed] (3.5,-0.8) -- (3.5,3.8);
      }
    }
\end{tikzpicture}} & ($E$ and $V$) &
    $\begin{pmatrix}
    1 & x & 0 & 0 & x & 1 & 0 & 0 & 1 \\
    0 & 0 & 1 & 1 & 1 & 1 & 0 & 1 & 0 \\
    y^{-1} & 1 & y^{-1} & 1 & 0 & 0 & 1 & 0 & 0 \\
    \hline
    0 & 0 & 0 & 0 & 0 & 0 & 0 & 0 & 1 \\
    0 & 0 & 0 & 0 & 0 & 0 & 0 & 1 & 0 \\
    0 & 0 & 0 & 0 & 0 & 0 & 1 & 0 & 0 
\end{pmatrix}$
    \\
    \end{tabular}
    \caption{Fermionic systems to which our algorithm was applied.}
    \label{tab:fermionic_systems}
\end{table*}

\begin{table}[]
    \centering
    \begin{tabular}{c|c}
        Lattice (qubits/cell) & Unit cell \\
        \hline
        
        Square (2) & \scalebox{0.3}{\begin{tikzpicture}
    \foreach \x in {0,...,1}{
      \foreach \y in {0,...,1}{
        
        \draw (0,-0.8) -- (0,3.8);
        \draw (1,-0.8) -- (1,3.8);
        \draw (2,-0.8) -- (2,3.8);
        \draw (3,-0.8) -- (3,3.8);
        \draw (-0.8,0.5) -- (3.8,0.5);
        \draw (-0.8,2.5) -- (3.8,2.5);
        
        \node[draw,circle,inner sep=2pt,fill] at (2*\x,2*\y+0.5) {};
        \node[draw,circle,inner sep=2pt,fill] at (2*\x+1,2*\y+0.5) {};
        
        \node[align=left] at (0.3,0.5+0.3) {\Large\textcolor{orange}{1}};
        \node[align=left] at (1+0.3,0.5+0.3) {\Large\textcolor{orange}{2}};

        \draw [gray,dashed] (-0.5,-0.8) -- (-0.5,3.8);
        \draw [gray,dashed] (1.5,-0.8) -- (1.5,3.8);
        \draw [gray,dashed] (-0.8,1.5) -- (3.8,1.5);
        \draw [gray,dashed] (-0.8,-0.5) -- (3.8,-0.5);
        \draw [gray,dashed] (-0.8,3.5) -- (3.8,3.5);
        \draw [gray,dashed] (3.5,-0.8) -- (3.5,3.8);
      }
    }
\end{tikzpicture}}\\
        \hline
        Tilted square (2) & \scalebox{0.3}{\begin{tikzpicture}
    \foreach \x in {0,...,1}{
      \foreach \y in {0,...,1}{
        
        \draw (-0.8,0.2) -- (2.8,3.8);
        \draw (0.2,-0.8) -- (3.8,2.8);
        \draw (1.8,-0.8) -- (-0.8,1.8);
        \draw (3.8,-0.8) -- (-0.8,3.8);
        \draw (3.8,-0.8) -- (-0.8,3.8);
        \draw (3.8,1.2) -- (1.2,3.8);
        \draw (-0.8,2.2) -- (0.8,3.8);
        \draw (2.2,-0.8) -- (3.8,0.8);
        \node[draw,circle,inner sep=2pt,fill] at (2*\x+1,2*\y) {};
        \node[draw,circle,inner sep=2pt,fill] at (2*\x,2*\y+1) {};
        
        \node[align=left] at (0.3,1+0.3) {\Large\textcolor{orange}{1}};
        \node[align=left] at (1+0.3,0.3) {\Large\textcolor{orange}{2}};
        
        \draw [gray,dashed] (-0.5,-0.8) -- (-0.5,3.8);
        \draw [gray,dashed] (1.5,-0.8) -- (1.5,3.8);
        \draw [gray,dashed] (-0.8,1.5) -- (3.8,1.5);
        \draw [gray,dashed] (-0.8,-0.5) -- (3.8,-0.5);
        \draw [gray,dashed] (-0.8,3.5) -- (3.8,3.5);
        \draw [gray,dashed] (3.5,-0.8) -- (3.5,3.8);
      }
    }
\end{tikzpicture}}\\
        \hline
        Sq bilayer (2) & \scalebox{0.3}{\begin{tikzpicture}
    \foreach \x in {0,...,1}{
      \foreach \y in {0,...,1}{
        
        \draw (-0.8,0) -- (3.8,0);
        \draw (-0.8,1) -- (3.8,1);
        \draw (-0.8,2) -- (3.8,2);
        \draw (-0.8,3) -- (3.8,3);
        \draw (0,-0.8) -- (0,3.8);
        \draw (1,-0.8) -- (1,3.8);
        \draw (2,-0.8) -- (2,3.8);
        \draw (3,-0.8) -- (3,3.8);
        \draw (0,1) -- (1,0);
        \draw (0,3) -- (1,2);
        \draw (2,1) -- (3,0);
        \draw (2,3) -- (3,2);
        \node[draw,circle,inner sep=2pt,fill] at (2*\x+1,2*\y) {};
        \node[draw,circle,inner sep=2pt,fill] at (2*\x,2*\y+1) {};
        
        \node[align=left] at (0.3,1+0.3) {\Large\textcolor{orange}{1}};
        \node[align=left] at (1+0.3,0.3) {\Large\textcolor{orange}{2}};
        
        \draw [gray,dashed] (-0.5,-0.8) -- (-0.5,3.8);
        \draw [gray,dashed] (1.5,-0.8) -- (1.5,3.8);
        \draw [gray,dashed] (-0.8,1.5) -- (3.8,1.5);
        \draw [gray,dashed] (-0.8,-0.5) -- (3.8,-0.5);
        \draw [gray,dashed] (-0.8,3.5) -- (3.8,3.5);
        \draw [gray,dashed] (3.5,-0.8) -- (3.5,3.8);
      }
    }
\end{tikzpicture}}\\
        \hline
        Hex (2) & \scalebox{0.3}{\begin{tikzpicture}
    \foreach \x in {0,...,1}{
      \foreach \y in {0,...,1}{
        
        \draw (2.2,3.8) -- (3.8,2.2);
        
        \draw (-0.8,2.8) -- (2.8,-0.8);
        
        \draw (0.2,3.8) -- (3.8,0.2);
        
        \draw (-0.8,0.8) -- (0.8,-0.8);
        
        \draw (2*\x,2*\y) -- (2*\x+1,2*\y+1);
        
        \node[draw,circle,inner sep=2pt,fill] at (2*\x,2*\y) {};
        \node[draw,circle,inner sep=2pt,fill] at (2*\x+1,2*\y+1) {};
        
        \node[align=left] at (-0.3,-0.3) {\Large\textcolor{orange}{1}};
        \node[align=left] at (1+0.3,1+0.3) {\Large\textcolor{orange}{2}};
        
        \draw [gray,dashed] (-0.5,-0.8) -- (-0.5,3.8);
        \draw [gray,dashed] (1.5,-0.8) -- (1.5,3.8);
        \draw [gray,dashed] (-0.8,1.5) -- (3.8,1.5);
        \draw [gray,dashed] (-0.8,-0.5) -- (3.8,-0.5);
        \draw [gray,dashed] (-0.8,3.5) -- (3.8,3.5);
        \draw [gray,dashed] (3.5,-0.8) -- (3.5,3.8);
      }
    }
\end{tikzpicture}}\\
        \hline
        Triang (2) & \scalebox{0.3}{\begin{tikzpicture}
    \foreach \x in {0,...,1}{
      \foreach \y in {0,...,1}{
        
        \draw (0,-0.8) -- (0,3.8);
        \draw (1,-0.8) -- (1,3.8);
        \draw (2,-0.8) -- (2,3.8);
        \draw (3,-0.8) -- (3,3.8);
        \draw (-0.8,0.5) -- (3.8,0.5);
        \draw (-0.8,2.5) -- (3.8,2.5);
        \draw (-0.65,-0.8) -- (1.65,3.8);
        \draw (-0.8,0.85) -- (.65,3.8);
        \draw (0.35,-0.8) -- (2.65,3.8);
        \draw (1.35,-0.8) -- (3.65,3.8);
        \draw (2.35,-0.8) -- (3.65,1.85);
        \node[draw,circle,inner sep=2pt,fill] at (2*\x,2*\y+0.5) {};
        \node[draw,circle,inner sep=2pt,fill] at (2*\x+1,2*\y+0.5) {};
        
        \node[align=left] at (0.3,0.5+0.3) {\Large\textcolor{orange}{1}};
        \node[align=left] at (1+0.3,0.5+0.3) {\Large\textcolor{orange}{2}};
        
        \draw [gray,dashed] (-0.5,-0.8) -- (-0.5,3.8);
        \draw [gray,dashed] (1.5,-0.8) -- (1.5,3.8);
        \draw [gray,dashed] (-0.8,1.5) -- (3.8,1.5);
        \draw [gray,dashed] (-0.8,-0.5) -- (3.8,-0.5);
        \draw [gray,dashed] (-0.8,3.5) -- (3.8,3.5);
        \draw [gray,dashed] (3.5,-0.8) -- (3.5,3.8);
      }
    }
\end{tikzpicture}}\\
        \hline
        Lieb lattice (3) & \scalebox{0.3}{\begin{tikzpicture}
    \foreach \x in {0,...,1}{
      \foreach \y in {0,...,1}{
        
        \draw (0,-0.8) -- (0,3.8);
        \draw (2,-0.8) -- (2,3.8);
        \draw (-0.8,0) -- (3.8,0);
        \draw (-0.8,2) -- (3.8,2);
        \node[draw,circle,inner sep=2pt,fill] at (2*\x,2*\y) {};
        \node[draw,circle,inner sep=2pt,fill] at (2*\x+1,2*\y) {};
        \node[draw,circle,inner sep=2pt,fill] at (2*\x,2*\y+1) {};
        
        \node[align=left] at (0+0.3,1+0.3) {\Large\textcolor{orange}{1}};
        \node[align=left] at (0+0.3,0+0.3) {\Large\textcolor{orange}{2}};
        \node[align=left] at (1+0.3,0+0.3) {\Large\textcolor{orange}{3}};
        
        \draw [gray,dashed] (-0.5,-0.8) -- (-0.5,3.8);
        \draw [gray,dashed] (1.5,-0.8) -- (1.5,3.8);
        \draw [gray,dashed] (-0.8,1.5) -- (3.8,1.5);
        \draw [gray,dashed] (-0.8,-0.5) -- (3.8,-0.5);
        \draw [gray,dashed] (-0.8,3.5) -- (3.8,3.5);
        \draw [gray,dashed] (3.5,-0.8) -- (3.5,3.8);
      }
    }
\end{tikzpicture}}\\
        \hline
        Rhombile (3) & \scalebox{0.3}{\begin{tikzpicture}
    
    \foreach \x in {0,...,1}{
      \foreach \y in {0,...,1}{
        
        \draw (2*\x,2*\y+0.125*1.414213) -- (2*\x+0.5,2*\y+0.5*1.414213+0.125*1.414213);
        \draw (2*\x+1,2*\y+0.125*1.414213) -- (2*\x+0.5,2*\y+0.5*1.414213+0.125*1.414213);
        \draw (2*\x+0.5,0.5*1.414213+0.125*1.414213) -- (2*\x,2+0.125*1.414213);
        \draw (1,2*\y+0.125*1.414213) -- (2+0.5,2*\y+0.5*1.414213+0.125*1.414213);
        \draw (0,2+0.125*1.414213) -- (2+0.5,0.5*1.414213+0.125*1.414213);
        \draw (1,2+0.125*1.414213) -- (2+0.5,0.5*1.414213+0.125*1.414213);
        
        \draw (-0.8,2) -- (0.5,0.5*1.414213+0.125*1.414213);
        \draw (-0.8,1.54) -- (0.5,0.5*1.414213+0.125*1.414213);
        
        \draw (-0.6,3.8) -- (0.5,2+0.5*1.414213+0.125*1.414213);
        \draw (-0.8,3.54) -- (0.5,2+0.5*1.414213+0.125*1.414213);
        
        \draw (1.4,3.8) -- (2+0.5,2+0.5*1.414213+0.125*1.414213);
        \draw (0.8,3.8) -- (2+0.5,2+0.5*1.414213+0.125*1.414213);
        
        \draw (2,2+0.125*1.414213) -- (3.8,1.2);
        \draw (2+1,2+0.125*1.414213) -- (3.8,1.45);
        
        \draw (0,0.125*1.414213) -- (1.8,-0.8);
        \draw (1,0.125*1.414213) -- (2,-0.8);
        
        \draw (2,0.125*1.414213) -- (3.8,-0.8);
        \draw (2+1,0.125*1.414213) -- (3.8,-0.6);
        
        \draw (2*\x,0.125*1.414213) -- (2*\x+0.4,-0.8);
        
        \draw (2*\x+0.5,2+0.5*1.414213+0.125*1.414213) -- (2*\x+0.2,3.8);
        
        \draw (0.5,2*\y+0.5*1.414213+0.125*1.414213) -- (-0.8,2*\y +0.4);
        
        \draw (2+1,2*\y+0.125*1.414213) -- (3.8,2*\y +0.6);
        
        \draw (2.8,3.8) -- (3.8,3.3);
        
        \draw (-0.8,-0.1) -- (0.1,-0.8);
        
        \node[draw,circle,inner sep=2pt,fill] at (2*\x,2*\y+0.125*1.414213) {}; 
        \node[draw,circle,inner sep=2pt,fill] at (2*\x+1,2*\y+0.125*1.414213) {}; 
        \node[draw,circle,inner sep=2pt,fill] at (2*\x+0.5,2*\y+0.5*1.414213+0.125*1.414213) {}; 
        
        \node[align=left] at (-0.3,-0.3+0.125*1.414213) {\Large\textcolor{orange}{1}};
        \node[align=left] at (1.2,-0.4+0.125*1.414213) {\Large\textcolor{orange}{2}};
        \node[align=left] at (0.4+0.5,0.5*1.414213+0.125*1.414213) {\Large\textcolor{orange}{3}};

        \draw [gray,dashed] (-0.5,-0.8) -- (-0.5,3.8);
        \draw [gray,dashed] (1.5,-0.8) -- (1.5,3.8);
        \draw [gray,dashed] (-0.8,1.5) -- (3.8,1.5);
        \draw [gray,dashed] (-0.8,-0.5) -- (3.8,-0.5);
        \draw [gray,dashed] (-0.8,3.5) -- (3.8,3.5);
        \draw [gray,dashed] (3.5,-0.8) -- (3.5,3.8);
      }
    }
\end{tikzpicture}}\\
        \hline
        Truncated sq (4) &\scalebox{0.3}{\begin{tikzpicture}
    \foreach \x in {0,...,1}{
      \foreach \y in {0,...,1}{
        
        \draw (0.5,-0.8) -- (0.5,0);
        \draw (2.5,-0.8) -- (2.5,0);
        \draw (0.5,3) -- (0.5,3.8);
        \draw (2.5,3) -- (2.5,3.8);
        \draw (-0.8,0.5) -- (0,0.5);
        \draw (-0.8,2.5) -- (0,2.5);
        \draw (3,0.5) -- (3.8,0.5);
        \draw (3,2.5) -- (3.8,2.5);
        \draw (0.5,1) -- (0.5,2);
        \draw (2.5,1) -- (2.5,2);
        \draw (1,0.5) -- (2,0.5);
        \draw (1,2.5) -- (2,2.5);
        
        \draw (0.5,0) -- (0,0.5);
        \draw (0.5,0) -- (1,0.5);
        \draw (0.5,1) -- (1,0.5);
        \draw (0.5,1) -- (0,0.5);
        
        \draw (0.5,2) -- (0,2.5);
        \draw (0.5,2) -- (1,2.5);
        \draw (0.5,3) -- (1,2.5);
        \draw (0.5,3) -- (0,2.5);
        
        \draw (2.5,0) -- (2,0.5);
        \draw (2.5,0) -- (3,0.5);
        \draw (2.5,1) -- (3,0.5);
        \draw (2.5,1) -- (2,0.5);
        
        \draw (2.5,2) -- (2,2.5);
        \draw (2.5,2) -- (3,2.5);
        \draw (2.5,3) -- (3,2.5);
        \draw (2.5,3) -- (2,2.5);
        
        \node[draw,circle,inner sep=2pt,fill] at (0.5+2*\x,0+2*\y) {};
        \node[draw,circle,inner sep=2pt,fill] at (0+2*\x,0.50+2*\y) {};
        \node[draw,circle,inner sep=2pt,fill] at (1+2*\x,0.50+2*\y) {};
        \node[draw,circle,inner sep=2pt,fill] at (0.5+2*\x,1+2*\y) {};
        
        \node[align=left] at (0.5+0.3,0-0.3) {\Large\textcolor{orange}{1}};
        \node[align=left] at (1+0.3,0.5-0.3) {\Large\textcolor{orange}{2}};
        \node[align=left] at (0-0.3,0.5+0.3) {\Large\textcolor{orange}{3}};
        \node[align=left] at (0.5-0.3,1+0.3) {\Large\textcolor{orange}{4}};
        
        \draw [gray,dashed] (-0.5,-0.8) -- (-0.5,3.8);
        \draw [gray,dashed] (1.5,-0.8) -- (1.5,3.8);
        \draw [gray,dashed] (-0.8,1.5) -- (3.8,1.5);
        \draw [gray,dashed] (-0.8,-0.5) -- (3.8,-0.5);
        \draw [gray,dashed] (-0.8,3.5) -- (3.8,3.5);
        \draw [gray,dashed] (3.5,-0.8) -- (3.5,3.8);
      }
    }
\end{tikzpicture}} \\
        \hline
        Hex bilayer (4) & \scalebox{0.3}{\begin{tikzpicture}
    \foreach \x in {0,...,1}{
      \foreach \y in {0,...,1}{
        
        \draw (-0.8,0.2+.3333*1) -- (2.8-.3333,3.8);
        \draw (0.2-.3333,-0.8) -- (3.8,2.8+.3333);
        \draw (-0.8,2.2+.3333) -- (0.8-.3333,3.8);
        \draw (2.2-.3333,-0.8) -- (3.8,0.8+.3333);
        
        \draw (0,1-.3333) -- (1,0-.3333);
        \draw (0,3-.3333) -- (1,2-.3333);
        \draw (2,1-.3333) -- (3,0-.3333);
        \draw (2,3-.3333) -- (3,2-.3333);
        
        \draw (-0.8,0.2-0.3333) -- (2.8+.3333,3.8);
        \draw (0.2+.3333,-0.8) -- (3.8,2.8-.3333);
        \draw (-0.8,2.2-.3333) -- (0.8+.3333,3.8);
        \draw (2.2+.3333,-0.8) -- (3.8,0.8-.3333);
        
        \draw (0,1+0.3333) -- (1,0+0.3333);
        \draw (0,3+0.3333) -- (1,2+0.3333);
        \draw (2,1+0.3333) -- (3,0+0.3333);
        \draw (2,3+0.3333) -- (3,2+0.3333);
        
        \draw (2*\x+1,2*\y+0.33) -- (2*\x+1,2*\y-0.33);
        \draw (2*\x,2*\y+1.33) -- (2*\x,2*\y+0.67);
        
        \node[draw,circle,inner sep=2pt,fill] at (2*\x+1,2*\y+0.33) {};
        \node[draw,circle,inner sep=2pt,fill] at (2*\x,2*\y+1.33) {};
        \node[draw,circle,inner sep=2pt,fill] at (2*\x+1,2*\y-0.33) {};
        \node[draw,circle,inner sep=2pt,fill] at (2*\x,2*\y+0.67) {};
        
        \node[align=left] at (-0.3,1-0.3) {\Large\textcolor{orange}{1}};
        \node[align=left] at (1-0.3,0-0.3) {\Large\textcolor{orange}{2}};
        \node[align=left] at (-0.3,1.33+0.3) {\Large\textcolor{orange}{3}};
        \node[align=left] at (1,.33+0.4) {\Large\textcolor{orange}{4}};
        
        \draw [gray,dashed] (-0.5,-0.8) -- (-0.5,3.8);
        \draw [gray,dashed] (1.5,-0.8) -- (1.5,3.8);
        \draw [gray,dashed] (-0.8,1.5) -- (3.8,1.5);
        \draw [gray,dashed] (-0.8,-0.5) -- (3.8,-0.5);
        \draw [gray,dashed] (-0.8,3.5) -- (3.8,3.5);
        \draw [gray,dashed] (3.5,-0.8) -- (3.5,3.8);
      }
    }
\end{tikzpicture}}\\
        \hline
        Sq bilayer (4) & \scalebox{0.3}{\begin{tikzpicture}
    \foreach \x in {0,...,1}{
      \foreach \y in {0,...,1}{
        
        \draw (2*\x+0,-0.8) -- (2*\x+0,3.8);
        \draw (2*\x+.33,-0.8) -- (2*\x+.33,3.8);
        \draw (2*\x+.67,-0.8) -- (2*\x+.67,3.8);
        \draw (2*\x+1,-0.8) -- (2*\x+1,3.8);
        \draw (-0.8,2*\y+0) -- (3.8,2*\y+0);
        \draw (-0.8,2*\y+1) -- (3.8,2*\y+1);
        
        \draw (2*\x+.33,2*\y) -- (2*\x,2*\y+1);
        \draw (2*\x+1,2*\y) -- (2*\x+0.67,2*\y+1);
        
        \node[draw,circle,inner sep=2pt,fill] at (2*\x+.33,2*\y) {};
        \node[draw,circle,inner sep=2pt,fill] at (2*\x+1,2*\y) {};
        \node[draw,circle,inner sep=2pt,fill] at (2*\x,2*\y+1) {};
        \node[draw,circle,inner sep=2pt,fill] at (2*\x+0.67,2*\y+1) {};
        
        \node[align=left] at (0.33-0.25,0-0.3) {\Large\textcolor{orange}{1}};
        \node[align=left] at (1+0.3,0-0.3) {\Large\textcolor{orange}{2}};
        \node[align=left] at (0-0.3,1+0.3) {\Large\textcolor{orange}{3}};
        \node[align=left] at (1.2,1+0.3) {\Large\textcolor{orange}{4}};
        
        \draw [gray,dashed] (-0.5,-0.8) -- (-0.5,3.8);
        \draw [gray,dashed] (1.5,-0.8) -- (1.5,3.8);
        \draw [gray,dashed] (-0.8,1.5) -- (3.8,1.5);
        \draw [gray,dashed] (-0.8,-0.5) -- (3.8,-0.5);
        \draw [gray,dashed] (-0.8,3.5) -- (3.8,3.5);
        \draw [gray,dashed] (3.5,-0.8) -- (3.5,3.8);
      }
    }
\end{tikzpicture}}\\
        \hline
        Snub sq (4) & \scalebox{0.3}{\begin{tikzpicture}
    \foreach \x in {0,...,1}{
      \foreach \y in {0,...,1}{
        
        \draw (1,0) -- (2,1);
        \draw (1,2) -- (2,3);
        \draw (-0.8,0.2) -- (0,1);
        \draw (-0.8,2.2) -- (0,3);
        \draw (3,0) -- (3.8,0.8);
        \draw (3,2) -- (3.8,2.8);
        
        \draw (1,1) -- (0,2);
        \draw (3,1) -- (2,2);
        \draw (0.8,-0.8) -- (0,0);
        \draw (2.8,-0.8) -- (2,0);
        \draw (1,3) -- (0.2,3.8);
        \draw (3,3) -- (2.2,3.8);
        
        \draw (-0.8,2*\y) -- (3.8,2*\y);
        \draw (-0.8,2*\y+1) -- (3.8,2*\y+1);
        \draw (2*\x,-0.8) -- (2*\x,3.8);
        \draw (2*\x+1,-0.8) -- (2*\x+1,3.8);
        
        \node[draw,circle,inner sep=2pt,fill] at (2*\x,2*\y) {};
        \node[draw,circle,inner sep=2pt,fill] at (2*\x,2*\y+1) {};
        \node[draw,circle,inner sep=2pt,fill] at (2*\x+1,2*\y) {};
        \node[draw,circle,inner sep=2pt,fill] at (2*\x+1,2*\y+1) {};
        
        \node[align=left] at (0-0.3,0.3) {\Large\textcolor{orange}{1}};
        \node[align=left] at (1-0.3,0+0.3) {\Large\textcolor{orange}{2}};
        \node[align=left] at (0-0.3,1+0.3) {\Large\textcolor{orange}{3}};
        \node[align=left] at (1+0.3,1+0.3) {\Large\textcolor{orange}{4}};
        
        \draw [gray,dashed] (-0.5,-0.8) -- (-0.5,3.8);
        \draw [gray,dashed] (1.5,-0.8) -- (1.5,3.8);
        \draw [gray,dashed] (-0.8,1.5) -- (3.8,1.5);
        \draw [gray,dashed] (-0.8,-0.5) -- (3.8,-0.5);
        \draw [gray,dashed] (-0.8,3.5) -- (3.8,3.5);
        \draw [gray,dashed] (3.5,-0.8) -- (3.5,3.8);
    
      }
    }
\end{tikzpicture}}\\
        \hline
        Heavy hex (5) & \scalebox{0.3}{\begin{tikzpicture}
    \foreach \x in {0,...,1}{
      \foreach \y in {0,...,1}{
        
        
        \draw (2*\x+0.3,2*\y-0.3) -- (2*\x+1.3,2*\y+0.7);
        
        \draw (0.8,-0.8) -- (-0.8,0.8);
        \draw (2.8,-0.8) -- (-0.8,2.8);
        \draw (3.8,0.2) -- (0.2,3.8);
        \draw (3.8,2.2) -- (2.2,3.8);
        
        \node[draw,circle,inner sep=2pt,fill] at (2*\x+0.8,2*\y+0.2) {};
        \node[draw,circle,inner sep=2pt,fill] at (2*\x+0.3,2*\y-0.3) {};
        \node[draw,circle,inner sep=2pt,fill] at (2*\x+1.3,2*\y+0.7) {};
        \node[draw,circle,inner sep=2pt,fill] at (2*\x+0.8,2*\y+1.2) {};
        \node[draw,circle,inner sep=2pt,fill] at (2*\x-0.2,2*\y+0.2) {};
        
        \node[align=left] at (-0.2,0.6) {\Large\textcolor{orange}{1}};
        \node[align=left] at (0,-0.4) {\Large\textcolor{orange}{2}};
        \node[align=left] at (1.2,-0.1) {\Large\textcolor{orange}{3}};
        \node[align=left] at (1.7,1) {\Large\textcolor{orange}{4}};
        \node[align=left] at (0.5,1) {\Large\textcolor{orange}{5}};

        \draw [gray,dashed] (-0.5,-0.8) -- (-0.5,3.8);
        \draw [gray,dashed] (1.5,-0.8) -- (1.5,3.8);
        \draw [gray,dashed] (-0.8,1.5) -- (3.8,1.5);
        \draw [gray,dashed] (-0.8,-0.5) -- (3.8,-0.5);
        \draw [gray,dashed] (-0.8,3.5) -- (3.8,3.5);
        \draw [gray,dashed] (3.5,-0.8) -- (3.5,3.8);
      }
    }
\end{tikzpicture}}\\
    \end{tabular}
    \caption{Hardware graphs and unit cells}
    \label{tab:hardware graphs}
\end{table}

\subsection{Newly identified encodings}
\label{sec:new_codes}

Given the constraints of maximum weight 3 and maximum graph cost 2, our searches over the fermionic systems in (Table \ref{tab:fermionic_systems}) and hardware graphs (Table \ref{tab:hardware graphs}) returned 25 encodings which are collected in Appendix \ref{sec:codes}. We have inspected each encoding returned and tried to discern whether the encoding could be reasonably considered to belong to an existing family of encodings, up to conjugation by single-qubit Clifford unitaries. To the best of our knowledge, 14 of these are new. 

Our algorithm returned encodings of a Kagome lattice with 3 fermionic modes per unit cell with edge and vertex operators onto two choices of 4 qubit unit cells for the truncated square lattice. These encodings had qubit/mode ratios of $4/3$ while still having all operators of weight 3 or below. This is particularly surprising given the relatively low connectivity of the truncated square lattice - each qubit is coupled to only 3 others.

While all encoded operators we allowed in our search had graph costs of 2 or less, the lowest average graph cost our search returned was an encoding of fermions on the hexagonal lattice onto the snub square qubit lattice at $1.2$. All of the operators we requested were encoded onto the lattice acting on nearest-neighbor pairs of qubits except for one 3-qubit edge operator. It is perhaps not surprising that the hexagonal lattice, which requires comparatively few edge operators per unit cell is efficiently encoded onto a high-connectivity lattice such as the snub square.

It is also noteworthy that a number of encodings that were found did not make use of all available qubits. In some cases, such as the encodings of the square lattice of fermions onto the bilayer geometries, the search was able to find a sufficient encoding onto just one of the two layers. In other cases, such as the encoding of the hexagonal lattice onto the heavy hex lattice, the search found that it was best to use just three qubits per unit cell lying on edges of the hexagon in only one orientation.

It is also interesting to observe the effect of the choice of unit cell on the ability of the search to find an encoding within the specified limits on operator weight and cost. We considered two different unit cells for Kagome lattices of fermionic modes. In both cases, previously unknown encodings were returned by the search, however for one of the two choices, an encoding onto the heavy hex lattice was returned,  again for a chosen unit cell.

\subsection{Optimality for previously known encodings}

While 14 of the encodings our search found were new, 9 of them were instances of another encoding family. 

First, for the square lattice with one mode per unit cell and specified operators consisting of horizontal and vertical operators and single vertex operator, the following graphs containing 2 qubits per unit cell returned an instance of a GSE encoding: sq (2), sq bilayer (2), tilted sq (2), hex (2), and triang (2). Our algorithm also constructed encodings on qubit lattices with 4 modes per unit cell for this system that did not use all the available qubits. In particular for this system a GSE-style encoding was generated for the 4 qubit per cell hex bilayer (4) and sq bilayer (4). In these two instances, two of the four qubits per cell were unused.

Recall that this implies that of the class of encodings with encoded operators invariant under translation under one unit cell, the GSE provides an optimal encoding given the cost metric discussed previously.

Other previously known mappings that were returned include a mapping of the tilted square lattice onto the Lieb lattice (also known as the decorated square). This mapping was equivalent up to single-qubit Cliffords to the Derby-Klassen Compact encoding.

One interesting encoding generated by our algorithm was of the triang (1) onto the tilted sq (2). This mapping lies outside the family of those generated by naively applying the GSE to the triangular lattice, as this would require a 3 qubits per fermionic mode. Upon closer inspection, one finds that the diagonal edge operators in the $xy$-direction are actually a product of edge operators in the $y$-direction and the $x$-direction however translated by one unit cell.

\section{Error-detecting codes}

An $[[n,k,d]]$ stabilizer code is characterized by the number of physical qubits $n$, the number of logical qubits encoded $k$, and the code distance $d$. The number of physical qubits minus the number of logical qubits is equal to the number of independent generators for the stabilizer group of the code, $|S| = 2^{n-k}$. The code distance is the minimum number of qubits that supports an operator which commutes with all elements of the stabilizer group and which is not itself in the stabilizer group. 

As mentioned above, local fermionic encodings in greater than one dimension are also a type of stabilizer code. They generically use more physical qubits than they encode fermionic modes and simulation takes place within a subspace defined as the common $+1$ eigenspace of an abelian group of Pauli operators. Typically, and as is the case here, a primary goal in designing a fermionic encodings is to identify a mapping which allows for the shortest circuit depth possible. This requires operators which are of low weight, ideally constant in the system size. This is fundamentally opposed to the features of codes sought out in quantum error correction where a high code distance is desirable to ensure logical qubits are well protected from errors. 

Nevertheless, some encodings which are optimal in the sense described in the previous sections, happen to have non-trivial code distances. An encoding with a code distance of $d$ (by definition) has the property that any Pauli operator of weight up to $d-1$ anticommutes with at least one stabilizer generator. It is not enough to check that all of the encoded edge operators, vertex operators, and their products have weight at least $d$, as it is possible that some Pauli operators commute with all the stabilizer group elements but do not correspond to even parity fermion operators.

The generators of the stabilizer group can be identified by computing generators of the kernel of $\tau$, either by Gaussian elimination in the case of a binary matrix, or computing syzygys in the case of a polynomial matrix (see Appendix \ref{sec:kernel}), and then applying the map $\sigma$ to find the associated Pauli operators. In the case where the kernel of $\sigma$ and the kernel of $\tau$ overlap, this generating set may be overcomplete, and include the identity.

Once a collection of stabilizer generators, $ \{s_0, s_1,...\}$, have been found that, along with their translations, generate the entire stabilizer group, we can collect the operators into a new object,
\begin{equation}
    S = \begin{pmatrix}
    s_0 & s_1 & ...
    \end{pmatrix}.
\end{equation}
As logical operators commute with all elements of the stabilizer group, they will return $0$ for their commutation value with each of $ \{s_0, s_1,...\}$. That is, they are elements of $\ker S^{\dag}\Lambda_Q$.

It suffices to check that each single-qubit operator acting on the reference cell has nonzero commutation relations with at least one of $ \{s_0, s_1,...\}$. If it can be verified that for each of the $3n$ such Pauli operators, $p \notin \ker S^{\dag}\Lambda_Q$, then the encoding provides a distance $2$ stabilizer code. Further, $S^{\dag}\Lambda_Q p$ is a column vector describing the syndrome pattern for the single-qubit error. The $i$th entry of the vector gives the translations of $s_i$ which anticommute with the operator given by $p$.

It is possible that the optimal encoding for a given hardware graph consists of operators that do not act on all of the qubits in each cell. See for example the encoding of the square lattice fermion system onto the 4-qubit/cell hex bilayer in Appendix \ref{sec:codes}. In this case, single-qubit operators on the unused qubits will commute with each element of the stabilizer group but will have a trivial action on the logical information. If all single qubit errors on all the qubits used by the encoding are detected, then the encoding consists of an error detecting code tensored with idle qubits. Despite the fact that the condition $p\notin \ker S^{\dag}\Lambda$ for all single-qubit $p$ is not strictly obeyed, the encoding can still be regarded as an error detecting code. It is also straightforward to check whether there are unused qubits.

The algorithm presented here allows for a simple filtering of encodings for error detecting codes. We will only focus on the simplest case of distance 2 codes. Firstly, by definition, in a distance 2 code, all logical operators are weight 2 or higher. Therefore, if our search is targeting error detecting codes, we will limit the Pauli operators we consider as encoded fermion operators to have minimum weight 2. When an encoding is found that satisfies the requisite (anti)-commutation conditions, we check that it is an error detecting code given the definition above. If it is, we accept it and if not, we discard it and continue the search.

For the purposes of this work we have opted to eschew writing the algorithm for computing syzygy's and instead compute the stabilizers of the encoding populated out to the 8 cells surrounding the unit cell. Although this does not in general guarantee all elements of the stabilizer group are computed, for the simple instances we consider we feel confident that this method suffices. Furthermore, in any case where we claim that error detection is achieved, any further stabilizers found would only strengthen this result.

\section{Discussion}
We have presented a method for searching over an extremely broad family of fermionic encodings in order to find encodings tailored both to particular fermionic algebras and to particular hardware. The class of fermionic encodings we search over is all second quantized encodings that constitute algebra homomorphisms and that map from Majorana monomials to Pauli operators. An essential element of this method is a clear formulation of the search criteria in terms of binary matrix equations and a rigorous argument for why solutions to such equations always yield valid encodings, and how all encodings within this class fall within this framework. To this search criteria we have applied fairly well known brute force search methods to optimize for our particular cost model. We expect these search methods are likely ripe for further optimizations. The cost model we consider is meant merely to demonstrate the power of this method, and also to satisfy our curiosity about how compactly we can represent operators on various systems. The search method may also be applied to many alternative cost models not explored here. However it is important to emphasize that in order to leverage many of the optimizations we have employed to improve runtime, one needs to be careful about how the cost of an individual Pauli assignment to a particular privileged fermionic operator relates to the total cost of the encoding.   

We have explained how our methods can be applied to translationally invariant systems with only minor modifications. This has the significant benefit of reducing the algorithmic cost of finding encodings for tiled fermionic and qubit systems, which is likely to be the most common use case. We would like to emphasize however that the encodings found using these methods need not require that the Hamiltonian that is simulated be translationally invariant, just that the privileged elements of the algebra used to generate the Hamiltonian be translationally invariant. For example coefficients of the Hamiltonian may vary however one wishes.

One feature of our method which we would like to strongly emphasize is its ability to leverage symmetries in the target algebra. As we have discussed, it may be the case that a particular system has symmetries which implies the algebra needed to represent the Hamiltonian is smaller than the whole fermionic algebra or the even fermionic algebra. Our method is able to automatically incorporate this fact, unlike many existing strategies which will often search for new encodings by applying variations on an established encoding that does not take advantage of these symmetries. However our method can only leverage those symmetries which appear in the group structure of the monomials, i.e. those symmetries which mean that certain monomials can not be constructed from products of the privileged set under consideration. Some symmetries do not manifest in this way, and so our method does not generically leverage all possible symmetries.

We have applied our method to the standard fermionic algebra of edge and vertex operators on various planar lattices as well as a square bilayer lattice, and we have also applied the method to the Fermi-Hubbard algebra on the square lattice. For these algebras we have found optimal encodings on planar and non-planar hardware layouts, inspired by current and prospective designs for superconducting chips. These are given in Appendix \ref{sec:codes}. Our method has been able identify previously unknown encodings, as well as rediscover known encodings, certifying their optimality according to our cost function. 

Since the encodings we consider are stabilizer codes, we have included a method of checking their ability to detect errors, and have certified the error detecting capabilities of many of the encodings found. We foresee that the methods described here could be generalized to further tailor fermionic encodings to the particular noise profile of the hardware.

Beyond merely finding codes for particular algebras and hardware layouts, our method could also be used as a subroutine in the design of new hardware layouts tailored to a particular simulation task, or for searching for simulation tasks particularly well suited to a given hardware layout. Both of these use cases seem particularly relevant in the search for practical applications of quantum computers. 

The runtime of our method can be extremely prohibitive, and it remains to be seen whether further improvements -- whether they be in the design of the search algorithm or in the details of the implementation -- could bring these runtimes down far enough to make these techniques more useful. The systems we have looked at push the boundary of what we were able to compute in a reasonable time frame, and we don't expect many more instances of this size remain that would be worth looking at. If one could improve performance sufficiently to even marginally increase the number of possible privileged operators and the number of qubits per unit cell, the space of possible problem instances could open up dramatically and these methods could become more useful. We feel optimistic that this is possible.

\textit{Acknowledgements} The authors thank Charles Derby, Toby Cubitt and James Whitfield for helpful conversations and feedback. This work was supported by Innovate UK [grant numbers 44167, 76963].


\appendix
\onecolumn

\clearpage

\section{Summary of generated encodings}\label{sec:codes}

\subsection{How to read the following tables}

For a specified fermionic unit cell with $m$ modes and and collection of $J$ operators specified as a $2m\times J$ array of polynomials as well as qubit layout grid with $n$ qubits per unit cell, the output is a $2n\times J$ array of polynomials. As an example, we illustrate the encoding of a fermionic system on a square lattice with specified operators $E_{0y}\propto\gamma_0\gamma_y$, $E_{0x}\propto\gamma_0\gamma_x$, and $V_0\propto\gamma_0\overline{\gamma}_0$ into the qubit operators $Z_{0,1}Z_{0,2}X_{y,1}$, $Z_{0,1}X_{0,2}Y_{x,1}$, and $Z_{0,1}Y_{0,2}$ respectively.

\begin{center}
    \includegraphics[width=0.8\textwidth]{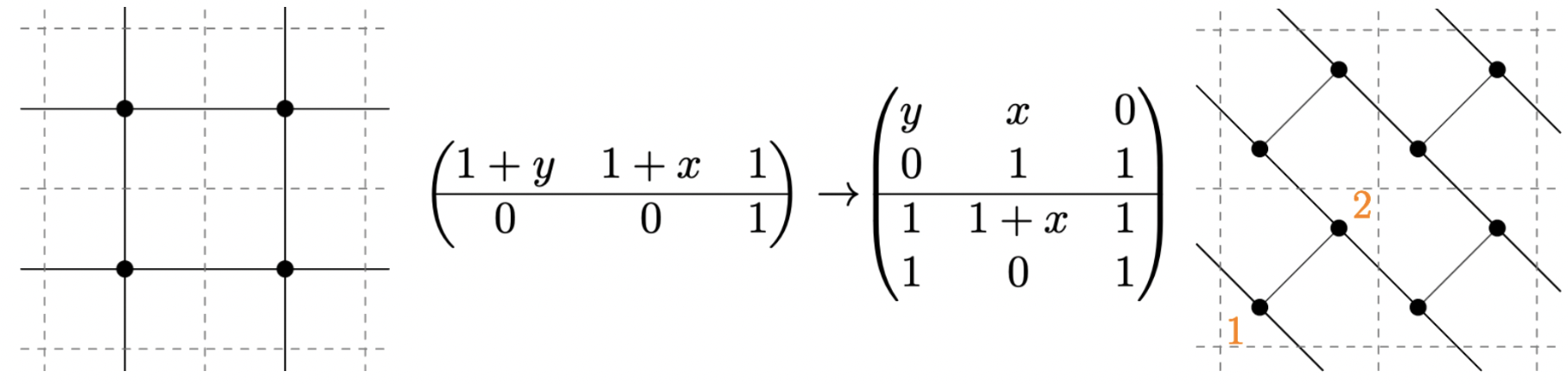}
\end{center}

To construct the encoded Pauli operator duals of the fermionic operators, match up the columns of the output polynomial array with the input polynomial array. Unit cells with more than one mode or qubit have the vertices numbered. Recall that on the fermionic operators side, a $1$ in the 1st entry corresponds to a Majorana $\gamma$ on the first mode in a reference unit cell, whereas a $x$ in the first entry gives a $\gamma$ in the first mode of the unit cell shifted by one in the $x$ direction. Nonzero entries in the $m+1$ through $2m$ entries of a column give the positions of $\overline{\gamma}$s. On the qubit side, the first $n$ entries in a column vector describe the positions of the $X$ operators, whereas the $n+1$ through $2n$ entries give the positions of $Z$s. We also include the stabilizer generators in the bottom-right.

\begin{center}
    \includegraphics[width=0.98\textwidth]{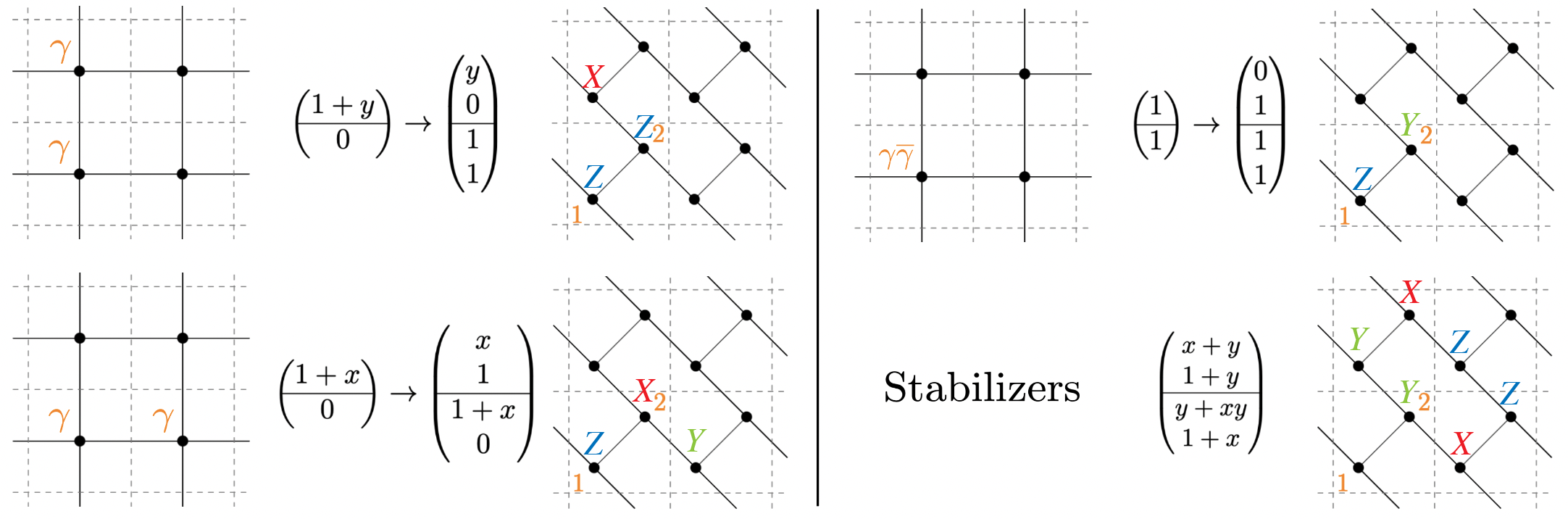}
\end{center}

\textit{Note:} In the following tables, Error detecting?: Yes* indicates that there are unused qubits in each unit cell but the encodings constitutes an error detecting code on the qubits which are used. We have also indicated whether each encoding is a member of an existing family or if it is first presented here.

\newgeometry{top=2cm, left=1cm, right=5mm, bottom=2cm} 

\begin{center}

\begin{tabular}{c|c|c|c}

    Fermionic unit cell & Fermionic operators &   \\
    \hline
    \begin{tabular}{@{}c@{}} 
    Square (1)\\ 
    \scalebox{0.7}{\begin{tikzpicture}
    \foreach \x in {0,...,1}{
      \foreach \y in {0,...,1}{
        
        \draw (0.5,-0.8) -- (0.5,3.8);
        \draw (2.5,-0.8) -- (2.5,3.8);
        \draw (-0.8,0.5) -- (3.8,0.5);
        \draw (-0.8,2.5) -- (3.8,2.5);
        
        \node[draw,circle,inner sep=2pt,fill] at (2*\x+0.5,2*\y+0.5) {};
        
        \draw [gray,dashed] (-0.5,-0.8) -- (-0.5,3.8);
        \draw [gray,dashed] (1.5,-0.8) -- (1.5,3.8);
        \draw [gray,dashed] (-0.8,1.5) -- (3.8,1.5);
        \draw [gray,dashed] (-0.8,-0.5) -- (3.8,-0.5);
        \draw [gray,dashed] (-0.8,3.5) -- (3.8,3.5);
        \draw [gray,dashed] (3.5,-0.8) -- (3.5,3.8);
      }
    }
\end{tikzpicture}}
    \end{tabular} & $\begin{pmatrix}
    1+y & 1+x  & 1 \\
    \hline
    0    & 0   & 1 
\end{pmatrix}$ & &
    \\

    \hline
    Qubit unit cell & Encoded operators & Stabilizers & Encoding properties\\
    \hline
    
    \begin{tabular}{@{}c@{}} 
    Square (2)\\ 
    \scalebox{0.7}{\begin{tikzpicture}
    \foreach \x in {0,...,1}{
      \foreach \y in {0,...,1}{
        
        \draw (0,-0.8) -- (0,3.8);
        \draw (1,-0.8) -- (1,3.8);
        \draw (2,-0.8) -- (2,3.8);
        \draw (3,-0.8) -- (3,3.8);
        \draw (-0.8,0.5) -- (3.8,0.5);
        \draw (-0.8,2.5) -- (3.8,2.5);
        
        \node[draw,circle,inner sep=2pt,fill] at (2*\x,2*\y+0.5) {};
        \node[draw,circle,inner sep=2pt,fill] at (2*\x+1,2*\y+0.5) {};
        
        \node[align=left] at (0.3,0.5+0.3) {\Large\textcolor{orange}{1}};
        \node[align=left] at (1+0.3,0.5+0.3) {\Large\textcolor{orange}{2}};

        \draw [gray,dashed] (-0.5,-0.8) -- (-0.5,3.8);
        \draw [gray,dashed] (1.5,-0.8) -- (1.5,3.8);
        \draw [gray,dashed] (-0.8,1.5) -- (3.8,1.5);
        \draw [gray,dashed] (-0.8,-0.5) -- (3.8,-0.5);
        \draw [gray,dashed] (-0.8,3.5) -- (3.8,3.5);
        \draw [gray,dashed] (3.5,-0.8) -- (3.5,3.8);
      }
    }
\end{tikzpicture}}
    \end{tabular} & $
\begin{pmatrix}
y & x & 0\\
0 & 1 & 1\\
\hline
1 & 1+x & 1\\
1 & 0 & 1
\end{pmatrix}
$

 & $\begin{pmatrix}
x+y \\
1+y \\
\hline
y+xy \\
1+x 
\end{pmatrix}$ &
    \begin{tabular}
    {@{}c@{}} 
    Qubits/mode: 2\\
    Max weight: 3\\
    Max graph cost: 2 \\
    Avg. graph cost: 1.67\\
    Error detecting?: Yes\\
    (GSE family \cite{setia2019superfast})
    \end{tabular}
    \\
    \hline
    \begin{tabular}{@{}c@{}} 
    Hex(2)\\ 
    \scalebox{0.7}{\begin{tikzpicture}
    \foreach \x in {0,...,1}{
      \foreach \y in {0,...,1}{
        
        \draw (2.2,3.8) -- (3.8,2.2);
        
        \draw (-0.8,2.8) -- (2.8,-0.8);
        
        \draw (0.2,3.8) -- (3.8,0.2);
        
        \draw (-0.8,0.8) -- (0.8,-0.8);
        
        \draw (2*\x,2*\y) -- (2*\x+1,2*\y+1);
        
        \node[draw,circle,inner sep=2pt,fill] at (2*\x,2*\y) {};
        \node[draw,circle,inner sep=2pt,fill] at (2*\x+1,2*\y+1) {};
        
        \node[align=left] at (-0.3,-0.3) {\Large\textcolor{orange}{1}};
        \node[align=left] at (1+0.3,1+0.3) {\Large\textcolor{orange}{2}};
        
        \draw [gray,dashed] (-0.5,-0.8) -- (-0.5,3.8);
        \draw [gray,dashed] (1.5,-0.8) -- (1.5,3.8);
        \draw [gray,dashed] (-0.8,1.5) -- (3.8,1.5);
        \draw [gray,dashed] (-0.8,-0.5) -- (3.8,-0.5);
        \draw [gray,dashed] (-0.8,3.5) -- (3.8,3.5);
        \draw [gray,dashed] (3.5,-0.8) -- (3.5,3.8);
      }
    }
\end{tikzpicture}}
    \end{tabular} & $
\begin{pmatrix}
y & x & 0\\
0 & 1 & 1\\
\hline
1 & 1+x & 1\\
1 & 0 & 1
\end{pmatrix}
$

 & $\begin{pmatrix}
x+y \\
1+y \\
\hline
y+xy \\
1+x \\
\end{pmatrix}$ &
    \begin{tabular}
    {@{}c@{}} 
    Qubits/mode: 2\\
    Max weight: 3\\
    Max graph cost: 2 \\
    Avg. graph cost: 1.67\\
    Error detecting?: Yes\\
    (GSE family \cite{setia2019superfast})
    \end{tabular}\\
    
    \hline
    \begin{tabular}{@{}c@{}} 
    Tilted sq (2)\\ 
    \scalebox{0.7}{\begin{tikzpicture}
    \foreach \x in {0,...,1}{
      \foreach \y in {0,...,1}{
        
        \draw (-0.8,0.2) -- (2.8,3.8);
        \draw (0.2,-0.8) -- (3.8,2.8);
        \draw (1.8,-0.8) -- (-0.8,1.8);
        \draw (3.8,-0.8) -- (-0.8,3.8);
        \draw (3.8,-0.8) -- (-0.8,3.8);
        \draw (3.8,1.2) -- (1.2,3.8);
        \draw (-0.8,2.2) -- (0.8,3.8);
        \draw (2.2,-0.8) -- (3.8,0.8);
        \node[draw,circle,inner sep=2pt,fill] at (2*\x+1,2*\y) {};
        \node[draw,circle,inner sep=2pt,fill] at (2*\x,2*\y+1) {};
        
        \node[align=left] at (0.3,1+0.3) {\Large\textcolor{orange}{1}};
        \node[align=left] at (1+0.3,0.3) {\Large\textcolor{orange}{2}};
        
        \draw [gray,dashed] (-0.5,-0.8) -- (-0.5,3.8);
        \draw [gray,dashed] (1.5,-0.8) -- (1.5,3.8);
        \draw [gray,dashed] (-0.8,1.5) -- (3.8,1.5);
        \draw [gray,dashed] (-0.8,-0.5) -- (3.8,-0.5);
        \draw [gray,dashed] (-0.8,3.5) -- (3.8,3.5);
        \draw [gray,dashed] (3.5,-0.8) -- (3.5,3.8);
      }
    }
\end{tikzpicture}}
    \end{tabular} & $
\begin{pmatrix}
0 & x & 1 \\
y & 1 & 0 \\
\hline
1 & 0 & 1\\
1 & 1+x & 1
\end{pmatrix}
$ & $\begin{pmatrix}
x+y \\
1+y \\
\hline
1+x \\
y+xy 
\end{pmatrix}$ &
    \begin{tabular}
    {@{}c@{}} 
    Qubits/mode: 2\\
    Max weight: 3\\
    Max graph cost: 2 \\
    Avg. graph cost: 1.67\\
    Error detecting?: Yes\\
    (GSE family \cite{setia2019superfast})
    \end{tabular}
    \\
    \hline
    \begin{tabular}{@{}c@{}} 
    Triang (2)\\ 
    \scalebox{0.7}{\begin{tikzpicture}
    \foreach \x in {0,...,1}{
      \foreach \y in {0,...,1}{
        
        \draw (0,-0.8) -- (0,3.8);
        \draw (1,-0.8) -- (1,3.8);
        \draw (2,-0.8) -- (2,3.8);
        \draw (3,-0.8) -- (3,3.8);
        \draw (-0.8,0.5) -- (3.8,0.5);
        \draw (-0.8,2.5) -- (3.8,2.5);
        \draw (-0.65,-0.8) -- (1.65,3.8);
        \draw (-0.8,0.85) -- (.65,3.8);
        \draw (0.35,-0.8) -- (2.65,3.8);
        \draw (1.35,-0.8) -- (3.65,3.8);
        \draw (2.35,-0.8) -- (3.65,1.85);
        \node[draw,circle,inner sep=2pt,fill] at (2*\x,2*\y+0.5) {};
        \node[draw,circle,inner sep=2pt,fill] at (2*\x+1,2*\y+0.5) {};
        
        \node[align=left] at (0.3,0.5+0.3) {\Large\textcolor{orange}{1}};
        \node[align=left] at (1+0.3,0.5+0.3) {\Large\textcolor{orange}{2}};
        
        \draw [gray,dashed] (-0.5,-0.8) -- (-0.5,3.8);
        \draw [gray,dashed] (1.5,-0.8) -- (1.5,3.8);
        \draw [gray,dashed] (-0.8,1.5) -- (3.8,1.5);
        \draw [gray,dashed] (-0.8,-0.5) -- (3.8,-0.5);
        \draw [gray,dashed] (-0.8,3.5) -- (3.8,3.5);
        \draw [gray,dashed] (3.5,-0.8) -- (3.5,3.8);
      }
    }
\end{tikzpicture}}
    \end{tabular} & $
\begin{pmatrix}
y & x & 0\\
0 & 1 & 1\\
\hline
1 & 1+x & 1\\
1 & 0 & 1
\end{pmatrix}
$ & $\begin{pmatrix}
x+y \\
1+y \\
\hline
y+xy \\
1+x 
\end{pmatrix}$ &
    \begin{tabular}
    {@{}c@{}} 
    Qubits/mode: 2\\
    Max weight: 3\\
    Max graph cost: 2 \\
    Avg. graph cost: 1.67\\
    Error detecting?: Yes\\
    (GSE family \cite{setia2019superfast})
    \end{tabular}\\

\end{tabular}
\clearpage
\begin{tabular}{c|c|c|c}
    Fermionic unit cell & Fermionic operators & \\
    \hline
    \begin{tabular}{@{}c@{}} 
    Square (1)\\ 
    \scalebox{0.7}{\begin{tikzpicture}
    \foreach \x in {0,...,1}{
      \foreach \y in {0,...,1}{
        
        \draw (0.5,-0.8) -- (0.5,3.8);
        \draw (2.5,-0.8) -- (2.5,3.8);
        \draw (-0.8,0.5) -- (3.8,0.5);
        \draw (-0.8,2.5) -- (3.8,2.5);
        
        \node[draw,circle,inner sep=2pt,fill] at (2*\x+0.5,2*\y+0.5) {};
        
        \draw [gray,dashed] (-0.5,-0.8) -- (-0.5,3.8);
        \draw [gray,dashed] (1.5,-0.8) -- (1.5,3.8);
        \draw [gray,dashed] (-0.8,1.5) -- (3.8,1.5);
        \draw [gray,dashed] (-0.8,-0.5) -- (3.8,-0.5);
        \draw [gray,dashed] (-0.8,3.5) -- (3.8,3.5);
        \draw [gray,dashed] (3.5,-0.8) -- (3.5,3.8);
      }
    }
\end{tikzpicture}}
    \end{tabular} & $\begin{pmatrix}
    1+y & 1+x  & 1 \\
    \hline
    0    & 0   & 1 
\end{pmatrix}$ & (cont.)
    \\
    \hline
    Qubit unit cell & Encoded operators & Stabilizers & Encoding properties\\
    \hline
    \begin{tabular}{@{}c@{}} 
    Sq bilayer (2)\\ 
    \scalebox{0.7}{\begin{tikzpicture}
    \foreach \x in {0,...,1}{
      \foreach \y in {0,...,1}{
        
        \draw (-0.8,0) -- (3.8,0);
        \draw (-0.8,1) -- (3.8,1);
        \draw (-0.8,2) -- (3.8,2);
        \draw (-0.8,3) -- (3.8,3);
        \draw (0,-0.8) -- (0,3.8);
        \draw (1,-0.8) -- (1,3.8);
        \draw (2,-0.8) -- (2,3.8);
        \draw (3,-0.8) -- (3,3.8);
        \draw (0,1) -- (1,0);
        \draw (0,3) -- (1,2);
        \draw (2,1) -- (3,0);
        \draw (2,3) -- (3,2);
        \node[draw,circle,inner sep=2pt,fill] at (2*\x+1,2*\y) {};
        \node[draw,circle,inner sep=2pt,fill] at (2*\x,2*\y+1) {};
        
        \node[align=left] at (0.3,1+0.3) {\Large\textcolor{orange}{1}};
        \node[align=left] at (1+0.3,0.3) {\Large\textcolor{orange}{2}};
        
        \draw [gray,dashed] (-0.5,-0.8) -- (-0.5,3.8);
        \draw [gray,dashed] (1.5,-0.8) -- (1.5,3.8);
        \draw [gray,dashed] (-0.8,1.5) -- (3.8,1.5);
        \draw [gray,dashed] (-0.8,-0.5) -- (3.8,-0.5);
        \draw [gray,dashed] (-0.8,3.5) -- (3.8,3.5);
        \draw [gray,dashed] (3.5,-0.8) -- (3.5,3.8);
      }
    }
\end{tikzpicture}}
    \end{tabular} & $
\begin{pmatrix}
y & x & 0\\
0 & 1 & 1\\
\hline
1 & 1+x & 1\\
1 & 0 & 1
\end{pmatrix}
$

 & $\begin{pmatrix}
x+y \\
1+y \\
\hline
y+xy \\
1+x 
\end{pmatrix}$ &
    \begin{tabular}
    {@{}c@{}} 
    Qubits/mode: 2\\
    Max weight: 3\\
    Max graph cost: 2 \\
    Avg. graph cost: 1.67\\
    Error detecting?: Yes\\
    (GSE family \cite{setia2019superfast})
    \end{tabular}
    \\
    
    \hline
    \begin{tabular}{@{}c@{}} 
    Rhombile (3)\\ 
    \scalebox{0.7}{\begin{tikzpicture}
    
    \foreach \x in {0,...,1}{
      \foreach \y in {0,...,1}{
        
        \draw (2*\x,2*\y+0.125*1.414213) -- (2*\x+0.5,2*\y+0.5*1.414213+0.125*1.414213);
        \draw (2*\x+1,2*\y+0.125*1.414213) -- (2*\x+0.5,2*\y+0.5*1.414213+0.125*1.414213);
        \draw (2*\x+0.5,0.5*1.414213+0.125*1.414213) -- (2*\x,2+0.125*1.414213);
        \draw (1,2*\y+0.125*1.414213) -- (2+0.5,2*\y+0.5*1.414213+0.125*1.414213);
        \draw (0,2+0.125*1.414213) -- (2+0.5,0.5*1.414213+0.125*1.414213);
        \draw (1,2+0.125*1.414213) -- (2+0.5,0.5*1.414213+0.125*1.414213);
        
        \draw (-0.8,2) -- (0.5,0.5*1.414213+0.125*1.414213);
        \draw (-0.8,1.54) -- (0.5,0.5*1.414213+0.125*1.414213);
        
        \draw (-0.6,3.8) -- (0.5,2+0.5*1.414213+0.125*1.414213);
        \draw (-0.8,3.54) -- (0.5,2+0.5*1.414213+0.125*1.414213);
        
        \draw (1.4,3.8) -- (2+0.5,2+0.5*1.414213+0.125*1.414213);
        \draw (0.8,3.8) -- (2+0.5,2+0.5*1.414213+0.125*1.414213);
        
        \draw (2,2+0.125*1.414213) -- (3.8,1.2);
        \draw (2+1,2+0.125*1.414213) -- (3.8,1.45);
        
        \draw (0,0.125*1.414213) -- (1.8,-0.8);
        \draw (1,0.125*1.414213) -- (2,-0.8);
        
        \draw (2,0.125*1.414213) -- (3.8,-0.8);
        \draw (2+1,0.125*1.414213) -- (3.8,-0.6);
        
        \draw (2*\x,0.125*1.414213) -- (2*\x+0.4,-0.8);
        
        \draw (2*\x+0.5,2+0.5*1.414213+0.125*1.414213) -- (2*\x+0.2,3.8);
        
        \draw (0.5,2*\y+0.5*1.414213+0.125*1.414213) -- (-0.8,2*\y +0.4);
        
        \draw (2+1,2*\y+0.125*1.414213) -- (3.8,2*\y +0.6);
        
        \draw (2.8,3.8) -- (3.8,3.3);
        
        \draw (-0.8,-0.1) -- (0.1,-0.8);
        
        \node[draw,circle,inner sep=2pt,fill] at (2*\x,2*\y+0.125*1.414213) {}; 
        \node[draw,circle,inner sep=2pt,fill] at (2*\x+1,2*\y+0.125*1.414213) {}; 
        \node[draw,circle,inner sep=2pt,fill] at (2*\x+0.5,2*\y+0.5*1.414213+0.125*1.414213) {}; 
        
        \node[align=left] at (-0.3,-0.3+0.125*1.414213) {\Large\textcolor{orange}{1}};
        \node[align=left] at (1.2,-0.4+0.125*1.414213) {\Large\textcolor{orange}{2}};
        \node[align=left] at (0.4+0.5,0.5*1.414213+0.125*1.414213) {\Large\textcolor{orange}{3}};

        \draw [gray,dashed] (-0.5,-0.8) -- (-0.5,3.8);
        \draw [gray,dashed] (1.5,-0.8) -- (1.5,3.8);
        \draw [gray,dashed] (-0.8,1.5) -- (3.8,1.5);
        \draw [gray,dashed] (-0.8,-0.5) -- (3.8,-0.5);
        \draw [gray,dashed] (-0.8,3.5) -- (3.8,3.5);
        \draw [gray,dashed] (3.5,-0.8) -- (3.5,3.8);
      }
    }
\end{tikzpicture}}
    \end{tabular} & $
\begin{pmatrix}
y & 1+x & 0\\
0 & 0 & 0\\
0 & xy^{-1} & 1\\
\hline
1 & 1 & 1\\
0 & 0 & 0\\
1 & 0 & 0
\end{pmatrix}
$ & $\begin{pmatrix}
1+x \\
0\\
x+xy^{-1} \\
\hline
x+y \\
0\\
1+x 
\end{pmatrix}$ &
    \begin{tabular}
    {@{}c@{}} 
    Qubits/mode: 2*\\
    Max weight: 3\\
    Max graph cost: 2 \\
    Avg. graph cost: 1.67\\
    Error detecting?: Yes*\\
    (New here)
    \end{tabular}\\

    \hline
    \begin{tabular}{@{}c@{}} 
    Hex bilayer (4)\\ 
    \scalebox{0.7}{\begin{tikzpicture}
    \foreach \x in {0,...,1}{
      \foreach \y in {0,...,1}{
        
        \draw (-0.8,0.2+.3333*1) -- (2.8-.3333,3.8);
        \draw (0.2-.3333,-0.8) -- (3.8,2.8+.3333);
        \draw (-0.8,2.2+.3333) -- (0.8-.3333,3.8);
        \draw (2.2-.3333,-0.8) -- (3.8,0.8+.3333);
        
        \draw (0,1-.3333) -- (1,0-.3333);
        \draw (0,3-.3333) -- (1,2-.3333);
        \draw (2,1-.3333) -- (3,0-.3333);
        \draw (2,3-.3333) -- (3,2-.3333);
        
        \draw (-0.8,0.2-0.3333) -- (2.8+.3333,3.8);
        \draw (0.2+.3333,-0.8) -- (3.8,2.8-.3333);
        \draw (-0.8,2.2-.3333) -- (0.8+.3333,3.8);
        \draw (2.2+.3333,-0.8) -- (3.8,0.8-.3333);
        
        \draw (0,1+0.3333) -- (1,0+0.3333);
        \draw (0,3+0.3333) -- (1,2+0.3333);
        \draw (2,1+0.3333) -- (3,0+0.3333);
        \draw (2,3+0.3333) -- (3,2+0.3333);
        
        \draw (2*\x+1,2*\y+0.33) -- (2*\x+1,2*\y-0.33);
        \draw (2*\x,2*\y+1.33) -- (2*\x,2*\y+0.67);
        
        \node[draw,circle,inner sep=2pt,fill] at (2*\x+1,2*\y+0.33) {};
        \node[draw,circle,inner sep=2pt,fill] at (2*\x,2*\y+1.33) {};
        \node[draw,circle,inner sep=2pt,fill] at (2*\x+1,2*\y-0.33) {};
        \node[draw,circle,inner sep=2pt,fill] at (2*\x,2*\y+0.67) {};
        
        \node[align=left] at (-0.3,1-0.3) {\Large\textcolor{orange}{1}};
        \node[align=left] at (1-0.3,0-0.3) {\Large\textcolor{orange}{2}};
        \node[align=left] at (-0.3,1.33+0.3) {\Large\textcolor{orange}{3}};
        \node[align=left] at (1,.33+0.4) {\Large\textcolor{orange}{4}};
        
        \draw [gray,dashed] (-0.5,-0.8) -- (-0.5,3.8);
        \draw [gray,dashed] (1.5,-0.8) -- (1.5,3.8);
        \draw [gray,dashed] (-0.8,1.5) -- (3.8,1.5);
        \draw [gray,dashed] (-0.8,-0.5) -- (3.8,-0.5);
        \draw [gray,dashed] (-0.8,3.5) -- (3.8,3.5);
        \draw [gray,dashed] (3.5,-0.8) -- (3.5,3.8);
      }
    }
\end{tikzpicture}}
    \end{tabular} & $
\begin{pmatrix}
y & x & 0 \\
0 & 1 & 1 \\
0 & 0 & 0 \\
0 & 0 & 0 \\
\hline
1 & 1+x & 1\\
1 & 0 & 1\\
0 & 0 & 0 \\
0 & 0 & 0 
\end{pmatrix}
$ & $\begin{pmatrix}
x+y \\
1+y \\
0\\
0\\
\hline
y+xy \\
1+x \\
0\\
0
\end{pmatrix}$ &
    \begin{tabular}
    {@{}c@{}} 
    Qubits/mode: 2*\\
    Max weight: 3\\
    Max graph cost: 2 \\
    Avg. graph cost: 1.67\\
    Error detecting?: Yes*\\
    (GSE family \cite{setia2019superfast})
    \end{tabular}\\
    
    \hline
    \begin{tabular}{@{}c@{}} 
    Sq bilayer (4)\\ 
    \scalebox{0.7}{\begin{tikzpicture}
    \foreach \x in {0,...,1}{
      \foreach \y in {0,...,1}{
        
        \draw (2*\x+0,-0.8) -- (2*\x+0,3.8);
        \draw (2*\x+.33,-0.8) -- (2*\x+.33,3.8);
        \draw (2*\x+.67,-0.8) -- (2*\x+.67,3.8);
        \draw (2*\x+1,-0.8) -- (2*\x+1,3.8);
        \draw (-0.8,2*\y+0) -- (3.8,2*\y+0);
        \draw (-0.8,2*\y+1) -- (3.8,2*\y+1);
        
        \draw (2*\x+.33,2*\y) -- (2*\x,2*\y+1);
        \draw (2*\x+1,2*\y) -- (2*\x+0.67,2*\y+1);
        
        \node[draw,circle,inner sep=2pt,fill] at (2*\x+.33,2*\y) {};
        \node[draw,circle,inner sep=2pt,fill] at (2*\x+1,2*\y) {};
        \node[draw,circle,inner sep=2pt,fill] at (2*\x,2*\y+1) {};
        \node[draw,circle,inner sep=2pt,fill] at (2*\x+0.67,2*\y+1) {};
        
        \node[align=left] at (0.33-0.25,0-0.3) {\Large\textcolor{orange}{1}};
        \node[align=left] at (1+0.3,0-0.3) {\Large\textcolor{orange}{2}};
        \node[align=left] at (0-0.3,1+0.3) {\Large\textcolor{orange}{3}};
        \node[align=left] at (1.2,1+0.3) {\Large\textcolor{orange}{4}};
        
        \draw [gray,dashed] (-0.5,-0.8) -- (-0.5,3.8);
        \draw [gray,dashed] (1.5,-0.8) -- (1.5,3.8);
        \draw [gray,dashed] (-0.8,1.5) -- (3.8,1.5);
        \draw [gray,dashed] (-0.8,-0.5) -- (3.8,-0.5);
        \draw [gray,dashed] (-0.8,3.5) -- (3.8,3.5);
        \draw [gray,dashed] (3.5,-0.8) -- (3.5,3.8);
      }
    }
\end{tikzpicture}}
    \end{tabular} & $
\begin{pmatrix}
y & x & 0\\
0 & 1 & 1\\
0 & 0 & 0\\
0 & 0 & 0\\
\hline
1 & 1+x & 1\\
1 & 0 & 1\\
0 & 0 & 0\\
0 & 0 & 0
\end{pmatrix}
$ & $\begin{pmatrix}
x+y \\
1+y \\
0\\
0\\
\hline
y+xy \\
1+x \\
0\\
0
\end{pmatrix}$ &
    \begin{tabular}
    {@{}c@{}} 
    Qubits/mode: 2*\\
    Max weight: 3\\
    Max graph cost: 2 \\
    Avg. graph cost: 1.67\\
    Error detecting?: Yes*\\
    (GSE family \cite{setia2019superfast})
    \end{tabular}\\

\end{tabular}

\end{center}

\clearpage
\begin{center}
\begin{tabular}{c|c|c|c}
    Fermionic unit cell & Fermionic operators & \\
    \hline
    \begin{tabular}{@{}c@{}} 
    Square (1)\\ 
    \scalebox{0.7}{\begin{tikzpicture}
    \foreach \x in {0,...,1}{
      \foreach \y in {0,...,1}{
        
        \draw (0.5,-0.8) -- (0.5,3.8);
        \draw (2.5,-0.8) -- (2.5,3.8);
        \draw (-0.8,0.5) -- (3.8,0.5);
        \draw (-0.8,2.5) -- (3.8,2.5);
        
        \node[draw,circle,inner sep=2pt,fill] at (2*\x+0.5,2*\y+0.5) {};
        
        \draw [gray,dashed] (-0.5,-0.8) -- (-0.5,3.8);
        \draw [gray,dashed] (1.5,-0.8) -- (1.5,3.8);
        \draw [gray,dashed] (-0.8,1.5) -- (3.8,1.5);
        \draw [gray,dashed] (-0.8,-0.5) -- (3.8,-0.5);
        \draw [gray,dashed] (-0.8,3.5) -- (3.8,3.5);
        \draw [gray,dashed] (3.5,-0.8) -- (3.5,3.8);
      }
    }
\end{tikzpicture}}
    \end{tabular} & $\begin{pmatrix}
    y & 1 & x & 1 & 1 \\
    \hline
    1 & y & 1 & x & 1 
\end{pmatrix}$ & (Num preserving)
    \\
    \hline
    Qubit unit cell & Encoded operators & Stabilizers & Encoding properties\\
    \hline

    \begin{tabular}{@{}c@{}} 
    Square (2)\\ 
    \scalebox{0.7}{\begin{tikzpicture}
    \foreach \x in {0,...,1}{
      \foreach \y in {0,...,1}{
        
        \draw (-0.8,0) -- (3.8,0);
        \draw (-0.8,1) -- (3.8,1);
        \draw (-0.8,2) -- (3.8,2);
        \draw (-0.8,3) -- (3.8,3);
        \draw (0,-0.8) -- (0,3.8);
        \draw (1,-0.8) -- (1,3.8);
        \draw (2,-0.8) -- (2,3.8);
        \draw (3,-0.8) -- (3,3.8);
        \draw (0,1) -- (1,0);
        \draw (0,3) -- (1,2);
        \draw (2,1) -- (3,0);
        \draw (2,3) -- (3,2);
        \node[draw,circle,inner sep=2pt,fill] at (2*\x+1,2*\y) {};
        \node[draw,circle,inner sep=2pt,fill] at (2*\x,2*\y+1) {};
        
        \node[align=left] at (0.3,1+0.3) {\Large\textcolor{orange}{1}};
        \node[align=left] at (1+0.3,0.3) {\Large\textcolor{orange}{2}};
        
        \draw [gray,dashed] (-0.5,-0.8) -- (-0.5,3.8);
        \draw [gray,dashed] (1.5,-0.8) -- (1.5,3.8);
        \draw [gray,dashed] (-0.8,1.5) -- (3.8,1.5);
        \draw [gray,dashed] (-0.8,-0.5) -- (3.8,-0.5);
        \draw [gray,dashed] (-0.8,3.5) -- (3.8,3.5);
        \draw [gray,dashed] (3.5,-0.8) -- (3.5,3.8);
      }
    }
\end{tikzpicture}}
    \end{tabular} & $
\begin{pmatrix}
0 & 1+y & x & 1 & 1 \\
0 & 0 & 1+x & 1+x & 0 \\
\hline
1+y & 0 & x & 1 & 1\\
1 & y & 0 & 1+x & 1
\end{pmatrix}
$

 & $\begin{pmatrix}
y+xy \\
1+x+y+xy\\
\hline
1+x \\
1+y
\end{pmatrix}$ &
    \begin{tabular}
    {@{}c@{}} 
    Qubits/mode: 2\\
    Max weight: 3\\
    Max graph cost: 2 \\
    Avg. graph cost: 1.8\\
    Error detecting?: Yes\\
    (GSE family \cite{setia2019superfast})
    \end{tabular}
    \\
    
\end{tabular}

\end{center}
\clearpage
\begin{center}
\begin{tabular}{c|c|c|c}
    Fermionic unit cell & Fermionic operators & \\
    \hline
    \begin{tabular}{@{}c@{}} 
    Triang (1)\\ 
    \scalebox{0.7}{\begin{tikzpicture}
    \foreach \x in {0,...,1}{
      \foreach \y in {0,...,1}{
        
        \draw (0.5,-0.8) -- (0.5,3.8);
        \draw (2.5,-0.8) -- (2.5,3.8);
        \draw (-0.8,0.5) -- (3.8,0.5);
        \draw (-0.8,2.5) -- (3.8,2.5);
        \draw (-0.8,3.2) -- (-0.2,3.8);
        \draw (-0.8,1.2) -- (1.8,3.8);
        \draw (-0.8,-0.8) -- (3.8,3.8);
        \draw (1.2,-0.8) -- (3.8,1.8);
        \draw (3.2,-0.8) -- (3.8,-0.2);
        
        \node[draw,circle,inner sep=2pt,fill] at (2*\x+0.5,2*\y+0.5) {};
        
        \draw [gray,dashed] (-0.5,-0.8) -- (-0.5,3.8);
        \draw [gray,dashed] (1.5,-0.8) -- (1.5,3.8);
        \draw [gray,dashed] (-0.8,1.5) -- (3.8,1.5);
        \draw [gray,dashed] (-0.8,-0.5) -- (3.8,-0.5);
        \draw [gray,dashed] (-0.8,3.5) -- (3.8,3.5);
        \draw [gray,dashed] (3.5,-0.8) -- (3.5,3.8);
      }
    }
\end{tikzpicture}}
    \end{tabular} & $\begin{pmatrix}
    1+xy & 1+y & 1+x & 1 \\
    \hline
    0 & 0 & 0 & 1
\end{pmatrix}$ &
    \\
    \hline
    Qubit unit cell & Encoded operators & Stabilizer & Encoding properties\\
    \hline

    \begin{tabular}{@{}c@{}} 
    Tilted sq (2)\\ 
    \scalebox{0.7}{\begin{tikzpicture}
    \foreach \x in {0,...,1}{
      \foreach \y in {0,...,1}{
        
        \draw (-0.8,0.2) -- (2.8,3.8);
        \draw (0.2,-0.8) -- (3.8,2.8);
        \draw (1.8,-0.8) -- (-0.8,1.8);
        \draw (3.8,-0.8) -- (-0.8,3.8);
        \draw (3.8,-0.8) -- (-0.8,3.8);
        \draw (3.8,1.2) -- (1.2,3.8);
        \draw (-0.8,2.2) -- (0.8,3.8);
        \draw (2.2,-0.8) -- (3.8,0.8);
        \node[draw,circle,inner sep=2pt,fill] at (2*\x+1,2*\y) {};
        \node[draw,circle,inner sep=2pt,fill] at (2*\x,2*\y+1) {};
        
        \node[align=left] at (0.3,1+0.3) {\Large\textcolor{orange}{1}};
        \node[align=left] at (1+0.3,0.3) {\Large\textcolor{orange}{2}};
        
        \draw [gray,dashed] (-0.5,-0.8) -- (-0.5,3.8);
        \draw [gray,dashed] (1.5,-0.8) -- (1.5,3.8);
        \draw [gray,dashed] (-0.8,1.5) -- (3.8,1.5);
        \draw [gray,dashed] (-0.8,-0.5) -- (3.8,-0.5);
        \draw [gray,dashed] (-0.8,3.5) -- (3.8,3.5);
        \draw [gray,dashed] (3.5,-0.8) -- (3.5,3.8);
      }
    }
\end{tikzpicture}}
    \end{tabular} & $
\begin{pmatrix}
xy & y & 1+x & 1\\
0 & y & 1 & 0\\
\hline
1 & 1+y & 1 & 1\\
y & 0 & 1 & 1
\end{pmatrix}
$ & $\begin{pmatrix}
1+x \\
1+xy\\
\hline
x+xy \\
1+y
\end{pmatrix}$ &
    \begin{tabular}
    {@{}c@{}} 
    Qubits/mode: 2\\
    Max weight: 3\\
    Max graph cost: 2 \\
    Avg. graph cost: 1.75\\
    Error detecting?: Yes\\
    (New here*)
    \end{tabular}
    \\
    \hline
    \begin{tabular}{@{}c@{}} 
    Triang.(2)\\ 
    \scalebox{0.7}{\begin{tikzpicture}
    \foreach \x in {0,...,1}{
      \foreach \y in {0,...,1}{
        
        \draw (0,-0.8) -- (0,3.8);
        \draw (1,-0.8) -- (1,3.8);
        \draw (2,-0.8) -- (2,3.8);
        \draw (3,-0.8) -- (3,3.8);
        \draw (-0.8,0.5) -- (3.8,0.5);
        \draw (-0.8,2.5) -- (3.8,2.5);
        \draw (-0.65,-0.8) -- (1.65,3.8);
        \draw (-0.8,0.85) -- (.65,3.8);
        \draw (0.35,-0.8) -- (2.65,3.8);
        \draw (1.35,-0.8) -- (3.65,3.8);
        \draw (2.35,-0.8) -- (3.65,1.85);
        \node[draw,circle,inner sep=2pt,fill] at (2*\x,2*\y+0.5) {};
        \node[draw,circle,inner sep=2pt,fill] at (2*\x+1,2*\y+0.5) {};
        
        \node[align=left] at (0.3,0.5+0.3) {\Large\textcolor{orange}{1}};
        \node[align=left] at (1+0.3,0.5+0.3) {\Large\textcolor{orange}{2}};
        
        \draw [gray,dashed] (-0.5,-0.8) -- (-0.5,3.8);
        \draw [gray,dashed] (1.5,-0.8) -- (1.5,3.8);
        \draw [gray,dashed] (-0.8,1.5) -- (3.8,1.5);
        \draw [gray,dashed] (-0.8,-0.5) -- (3.8,-0.5);
        \draw [gray,dashed] (-0.8,3.5) -- (3.8,3.5);
        \draw [gray,dashed] (3.5,-0.8) -- (3.5,3.8);
      }
    }
\end{tikzpicture}}
    \end{tabular} & $
\begin{pmatrix}
xy & 1+y & x & 0\\
0 & x^{-1} & 1 & 1\\
\hline
1 & 1 & 1+x & 1\\
1 & 0 & 1 & 0
\end{pmatrix}
$ & $\begin{pmatrix}
1+y^{-1} \\
1+x^{-1}y^{-1} \\
\hline
1+x \\
1+y^{-1} \\
\end{pmatrix}$&
    \begin{tabular}
    {@{}c@{}} 
    Qubits/mode: 2\\
    Max weight: 3\\
    Max graph cost: 2 \\
    Avg. graph cost: 1.75\\
    Error detecting?: Yes\\
    (New here)
    \end{tabular}\\
    
\end{tabular}

\end{center}
\clearpage
\begin{center}
\begin{tabular}{c|c|c|c}
    Fermionic unit cell & Fermionic operators & \\
    \hline
    \begin{tabular}{@{}c@{}} 
    Spinful sq (2)\\ 
    \scalebox{0.7}{\begin{tikzpicture}
    \foreach \x in {0,...,1}{
      \foreach \y in {0,...,1}{
        
        \draw (-0.8,0) -- (3.8,0);
        \draw (-0.8,1) -- (3.8,1);
        \draw (-0.8,2) -- (3.8,2);
        \draw (-0.8,3) -- (3.8,3);
        \draw (0,-0.8) -- (0,3.8);
        \draw (1,-0.8) -- (1,3.8);
        \draw (2,-0.8) -- (2,3.8);
        \draw (3,-0.8) -- (3,3.8);
        \node[draw,circle,inner sep=2pt,fill] at (2*\x+1,2*\y) {};
        \node[draw,circle,inner sep=2pt,fill] at (2*\x,2*\y+1) {};
        
        \node[align=left] at (0.3,1+0.3) {\Large\textcolor{orange}{1}};
        \node[align=left] at (1+0.3,0.3) {\Large\textcolor{orange}{2}};
        
        \draw [gray,dashed] (-0.5,-0.8) -- (-0.5,3.8);
        \draw [gray,dashed] (1.5,-0.8) -- (1.5,3.8);
        \draw [gray,dashed] (-0.8,1.5) -- (3.8,1.5);
        \draw [gray,dashed] (-0.8,-0.5) -- (3.8,-0.5);
        \draw [gray,dashed] (-0.8,3.5) -- (3.8,3.5);
        \draw [gray,dashed] (3.5,-0.8) -- (3.5,3.8);
      }
    }
\end{tikzpicture}}
    \end{tabular} & $\begin{pmatrix}
    0 & 0 & 0 & 1+y & 1+x & 1 \\
    1+y & 1+x & 1 & 0 & 0 & 0 \\
    \hline
    0 & 0 & 0 & 0 & 0 & 1 \\
    0 & 0 & 1 & 0 & 0 & 0 
\end{pmatrix}$ &
    \\
    \hline
    Qubit unit cell & Encoded operators & Stabilizer & Encoding properties\\
    \hline

    \begin{tabular}{@{}c@{}} 
    Hex bilayer (4)\\ 
    \scalebox{0.7}{\begin{tikzpicture}
    \foreach \x in {0,...,1}{
      \foreach \y in {0,...,1}{
        
        \draw (-0.8,0.2+.3333*1) -- (2.8-.3333,3.8);
        \draw (0.2-.3333,-0.8) -- (3.8,2.8+.3333);
        \draw (-0.8,2.2+.3333) -- (0.8-.3333,3.8);
        \draw (2.2-.3333,-0.8) -- (3.8,0.8+.3333);
        
        \draw (0,1-.3333) -- (1,0-.3333);
        \draw (0,3-.3333) -- (1,2-.3333);
        \draw (2,1-.3333) -- (3,0-.3333);
        \draw (2,3-.3333) -- (3,2-.3333);
        
        \draw (-0.8,0.2-0.3333) -- (2.8+.3333,3.8);
        \draw (0.2+.3333,-0.8) -- (3.8,2.8-.3333);
        \draw (-0.8,2.2-.3333) -- (0.8+.3333,3.8);
        \draw (2.2+.3333,-0.8) -- (3.8,0.8-.3333);
        
        \draw (0,1+0.3333) -- (1,0+0.3333);
        \draw (0,3+0.3333) -- (1,2+0.3333);
        \draw (2,1+0.3333) -- (3,0+0.3333);
        \draw (2,3+0.3333) -- (3,2+0.3333);
        
        \draw (2*\x+1,2*\y+0.33) -- (2*\x+1,2*\y-0.33);
        \draw (2*\x,2*\y+1.33) -- (2*\x,2*\y+0.67);
        
        \node[draw,circle,inner sep=2pt,fill] at (2*\x+1,2*\y+0.33) {};
        \node[draw,circle,inner sep=2pt,fill] at (2*\x,2*\y+1.33) {};
        \node[draw,circle,inner sep=2pt,fill] at (2*\x+1,2*\y-0.33) {};
        \node[draw,circle,inner sep=2pt,fill] at (2*\x,2*\y+0.67) {};
        
        \node[align=left] at (-0.3,1-0.3) {\Large\textcolor{orange}{1}};
        \node[align=left] at (1-0.3,0-0.3) {\Large\textcolor{orange}{2}};
        \node[align=left] at (-0.3,1.33+0.3) {\Large\textcolor{orange}{3}};
        \node[align=left] at (1,.33+0.4) {\Large\textcolor{orange}{4}};
        
        \draw [gray,dashed] (-0.5,-0.8) -- (-0.5,3.8);
        \draw [gray,dashed] (1.5,-0.8) -- (1.5,3.8);
        \draw [gray,dashed] (-0.8,1.5) -- (3.8,1.5);
        \draw [gray,dashed] (-0.8,-0.5) -- (3.8,-0.5);
        \draw [gray,dashed] (-0.8,3.5) -- (3.8,3.5);
        \draw [gray,dashed] (3.5,-0.8) -- (3.5,3.8);
      }
    }
\end{tikzpicture}}
    \end{tabular} & $
\begin{pmatrix}
y & x & 0 & 0 & 0 & 0\\
0 & 1 & 1 & 0 & 0 & 0\\
0 & 0 & 0 & y & x & 0\\
0 & 0 & 0 & 0 & 1 & 1\\
\hline
1 & 1+x & 1 & 0 & 0 & 0\\
1 & 0 & 1 & 0 & 0 & 0\\
0 & 0 & 0 & 1 & 1+x & 1\\
0 & 0 & 0 & 1 & 0 & 1
\end{pmatrix}
$ & $\begin{pmatrix}
0  & x+y\\
0 & 1+y\\
x+y & 0\\
1+y & 0\\
\hline
0 & x+xy\\
0 & 1+x\\
x+xy & 0\\
1+x & 0
\end{pmatrix}$ & 
    \begin{tabular}
    {@{}c@{}} 
    Qubits/mode: 2\\
    Max weight: 3\\
    Max graph cost: 2 \\
    Avg. graph cost: 1.67\\
    Error detecting?: Yes\\
    (GSE family \cite{setia2019superfast})
    \end{tabular}
    \\
    \hline
    \begin{tabular}{@{}c@{}} 
    Sq bilayer (4)\\ 
    \scalebox{0.7}{\begin{tikzpicture}
    \foreach \x in {0,...,1}{
      \foreach \y in {0,...,1}{
        
        \draw (2*\x+0,-0.8) -- (2*\x+0,3.8);
        \draw (2*\x+.33,-0.8) -- (2*\x+.33,3.8);
        \draw (2*\x+.67,-0.8) -- (2*\x+.67,3.8);
        \draw (2*\x+1,-0.8) -- (2*\x+1,3.8);
        \draw (-0.8,2*\y+0) -- (3.8,2*\y+0);
        \draw (-0.8,2*\y+1) -- (3.8,2*\y+1);
        
        \draw (2*\x+.33,2*\y) -- (2*\x,2*\y+1);
        \draw (2*\x+1,2*\y) -- (2*\x+0.67,2*\y+1);
        
        \node[draw,circle,inner sep=2pt,fill] at (2*\x+.33,2*\y) {};
        \node[draw,circle,inner sep=2pt,fill] at (2*\x+1,2*\y) {};
        \node[draw,circle,inner sep=2pt,fill] at (2*\x,2*\y+1) {};
        \node[draw,circle,inner sep=2pt,fill] at (2*\x+0.67,2*\y+1) {};
        
        \node[align=left] at (0.33-0.25,0-0.3) {\Large\textcolor{orange}{1}};
        \node[align=left] at (1+0.3,0-0.3) {\Large\textcolor{orange}{2}};
        \node[align=left] at (0-0.3,1+0.3) {\Large\textcolor{orange}{3}};
        \node[align=left] at (1.2,1+0.3) {\Large\textcolor{orange}{4}};
        
        \draw [gray,dashed] (-0.5,-0.8) -- (-0.5,3.8);
        \draw [gray,dashed] (1.5,-0.8) -- (1.5,3.8);
        \draw [gray,dashed] (-0.8,1.5) -- (3.8,1.5);
        \draw [gray,dashed] (-0.8,-0.5) -- (3.8,-0.5);
        \draw [gray,dashed] (-0.8,3.5) -- (3.8,3.5);
        \draw [gray,dashed] (3.5,-0.8) -- (3.5,3.8);
      }
    }
\end{tikzpicture}}
    \end{tabular} & $
\begin{pmatrix}
y & x & 0 & 0 & 0 & 0\\
0 & 1 & 1 & 0 & 0 & 0\\
0 & 0 & 0 & y & x & 0\\
0 & 0 & 0 & 0 & 1 & 1\\
\hline
1 & 1+x & 1 & 0 & 0   & 0\\
1 & 0   & 1 & 0 & 0   & 0\\
0 & 0   & 0 & 1 & 1+x & 1\\
0 & 0   & 0 & 1 & 0   & 1
\end{pmatrix}
$ & $\begin{pmatrix}
0  & x+y\\
0 & 1+y\\
x+y & 0\\
1+y & 0\\
\hline
0 & x+xy\\
0 & 1+x\\
x+xy & 0\\
1+x & 0
\end{pmatrix}$ &
    \begin{tabular}
    {@{}c@{}} 
    Qubits/mode: 2\\
    Max weight: 3\\
    Max graph cost: 2 \\
    Avg. graph cost: 1.67\\
    Error detecting?: Yes\\
    (GSE family \cite{setia2019superfast})
    \end{tabular}\\
    
\end{tabular}
\end{center}
\clearpage
\begin{center}
\begin{tabular}{c|c|c|c}
    Fermionic unit cell & Fermionic operators & \\
    \hline
    \begin{tabular}{@{}c@{}} 
    Hex (2)\\ 
    \scalebox{0.7}{\begin{tikzpicture}
    \foreach \x in {0,...,1}{
      \foreach \y in {0,...,1}{
        
        \draw (2.2,3.8) -- (3.8,2.2);
        
        \draw (-0.8,2.8) -- (2.8,-0.8);
        
        \draw (0.2,3.8) -- (3.8,0.2);
        
        \draw (-0.8,0.8) -- (0.8,-0.8);
        
        \draw (2*\x,2*\y) -- (2*\x+1,2*\y+1);
        
        \node[draw,circle,inner sep=2pt,fill] at (2*\x,2*\y) {};
        \node[draw,circle,inner sep=2pt,fill] at (2*\x+1,2*\y+1) {};
        
        \node[align=left] at (-0.3,-0.3) {\Large\textcolor{orange}{1}};
        \node[align=left] at (1+0.3,1+0.3) {\Large\textcolor{orange}{2}};
        
        \draw [gray,dashed] (-0.5,-0.8) -- (-0.5,3.8);
        \draw [gray,dashed] (1.5,-0.8) -- (1.5,3.8);
        \draw [gray,dashed] (-0.8,1.5) -- (3.8,1.5);
        \draw [gray,dashed] (-0.8,-0.5) -- (3.8,-0.5);
        \draw [gray,dashed] (-0.8,3.5) -- (3.8,3.5);
        \draw [gray,dashed] (3.5,-0.8) -- (3.5,3.8);
      }
    }
\end{tikzpicture}}
    \end{tabular} & $\begin{pmatrix}
    y & x & 1 & 0 & 1  \\
    1 & 1 & 1 & 1 & 0  \\
    \hline
    0 & 0 & 0 & 0 & 1  \\
    0 & 0 & 0 & 1 & 0  
\end{pmatrix}$ &
    \\
    \hline
    Qubit unit cell & Encoded operators & Stabilizer & Encoding properties\\
    
    \hline
    \begin{tabular}{@{}c@{}} 
    Lieb lattice (3)\\ 
    \scalebox{0.7}{\begin{tikzpicture}
    \foreach \x in {0,...,1}{
      \foreach \y in {0,...,1}{
        
        \draw (0,-0.8) -- (0,3.8);
        \draw (2,-0.8) -- (2,3.8);
        \draw (-0.8,0) -- (3.8,0);
        \draw (-0.8,2) -- (3.8,2);
        \node[draw,circle,inner sep=2pt,fill] at (2*\x,2*\y) {};
        \node[draw,circle,inner sep=2pt,fill] at (2*\x+1,2*\y) {};
        \node[draw,circle,inner sep=2pt,fill] at (2*\x,2*\y+1) {};
        
        \node[align=left] at (0+0.3,1+0.3) {\Large\textcolor{orange}{1}};
        \node[align=left] at (0+0.3,0+0.3) {\Large\textcolor{orange}{2}};
        \node[align=left] at (1+0.3,0+0.3) {\Large\textcolor{orange}{3}};
        
        \draw [gray,dashed] (-0.5,-0.8) -- (-0.5,3.8);
        \draw [gray,dashed] (1.5,-0.8) -- (1.5,3.8);
        \draw [gray,dashed] (-0.8,1.5) -- (3.8,1.5);
        \draw [gray,dashed] (-0.8,-0.5) -- (3.8,-0.5);
        \draw [gray,dashed] (-0.8,3.5) -- (3.8,3.5);
        \draw [gray,dashed] (3.5,-0.8) -- (3.5,3.8);
      }
    }
\end{tikzpicture}}
    \end{tabular} & $
\begin{pmatrix}
0 & xy^{-1} & y^{-1} & 0 & y^{-1}\\
0 & 0 & 0 & 1 & 1\\
0 & 1 & 1 & 1 & 0\\
\hline
1 & 0 & y^{-1} & 0 & y^{-1}\\
0 & 0 & 1 & 0 & 0\\
1 & 0 & 1 & 1 & 0
\end{pmatrix}
$ & $\begin{pmatrix}
1+x\\
0\\
1+x\\
\hline
x+xy^{-1}\\
x+y\\
1+y
\end{pmatrix}$ &
    \begin{tabular}
    {@{}c@{}} 
    Qubits/mode: 1.5\\
    Max weight: 3\\
    Max graph cost: 2 \\
    Avg. graph cost: 1.6\\
    Error detecting?: No\\
    (New here)
    \end{tabular}\\
    
    \hline
    \begin{tabular}{@{}c@{}} 
    Kagome (3)\\ 
    \scalebox{0.7}{\begin{tikzpicture}
    \foreach \x in {0,...,1}{
      \foreach \y in {0,...,1}{
        
        \draw (0,-0.8) -- (0,3.8);
        \draw (2,-0.8) -- (2,3.8);
        \draw (-0.8,1) -- (3.8,1);
        \draw (-0.8,3) -- (3.8,3);
        \draw (-0.8,-0.8) -- (3.8,3.8);
        \draw (-0.8,1.2) -- (1.8,3.8);
        \draw (1.2,-0.8) -- (3.8,1.8);
        
        \node[draw,circle,inner sep=2pt,fill] at (2*\x,2*\y) {};
        \node[draw,circle,inner sep=2pt,fill] at (2*\x,2*\y+1) {};
        \node[draw,circle,inner sep=2pt,fill] at (2*\x+1,2*\y+1) {};
        
        \node[align=left] at (0+0.3,0) {\Large\textcolor{orange}{1}};
        \node[align=left] at (1,1-0.3) {\Large\textcolor{orange}{2}};
        \node[align=left] at (0+0.3,1+0.3) {\Large\textcolor{orange}{3}};
        
        \draw [gray,dashed] (-0.5,-0.8) -- (-0.5,3.8);
        \draw [gray,dashed] (1.5,-0.8) -- (1.5,3.8);
        \draw [gray,dashed] (-0.8,1.5) -- (3.8,1.5);
        \draw [gray,dashed] (-0.8,-0.5) -- (3.8,-0.5);
        \draw [gray,dashed] (-0.8,3.5) -- (3.8,3.5);
        \draw [gray,dashed] (3.5,-0.8) -- (3.5,3.8);
      }
    }
\end{tikzpicture}}
    \end{tabular} & $
\begin{pmatrix}
0 & 1 & 1 & 1 & 0\\
0 & y^{-1} & x^{-1}y^{-1} & 0 & x^{-1}y^{-1}\\
0 & 0 & 0 & 1 & 1\\
\hline
1 & 0 & 1 & 1 & 0\\
x^{-1} & 0 & x^{-1}y^{-1} & 0 & x^{-1}y^{-1}\\
0 & 0 & 1 & 0 & 0
\end{pmatrix}
$ & $\begin{pmatrix}
1+x\\
1+x^{-1}\\
0\\
\hline
1+y\\
1+y^{-1}\\
x+y
\end{pmatrix}$ &
    \begin{tabular}
    {@{}c@{}} 
    Qubits/mode: 1.5\\
    Max weight: 3\\
    Max graph cost: 2 \\
    Avg. graph cost: 1.8\\
    Error detecting?: No\\
    (New here)
    \end{tabular}\\
    
    \hline
    \begin{tabular}{@{}c@{}} 
    Rhombile (3)\\ 
    \scalebox{0.7}{\begin{tikzpicture}
    
    \foreach \x in {0,...,1}{
      \foreach \y in {0,...,1}{
        
        \draw (2*\x,2*\y+0.125*1.414213) -- (2*\x+0.5,2*\y+0.5*1.414213+0.125*1.414213);
        \draw (2*\x+1,2*\y+0.125*1.414213) -- (2*\x+0.5,2*\y+0.5*1.414213+0.125*1.414213);
        \draw (2*\x+0.5,0.5*1.414213+0.125*1.414213) -- (2*\x,2+0.125*1.414213);
        \draw (1,2*\y+0.125*1.414213) -- (2+0.5,2*\y+0.5*1.414213+0.125*1.414213);
        \draw (0,2+0.125*1.414213) -- (2+0.5,0.5*1.414213+0.125*1.414213);
        \draw (1,2+0.125*1.414213) -- (2+0.5,0.5*1.414213+0.125*1.414213);
        
        \draw (-0.8,2) -- (0.5,0.5*1.414213+0.125*1.414213);
        \draw (-0.8,1.54) -- (0.5,0.5*1.414213+0.125*1.414213);
        
        \draw (-0.6,3.8) -- (0.5,2+0.5*1.414213+0.125*1.414213);
        \draw (-0.8,3.54) -- (0.5,2+0.5*1.414213+0.125*1.414213);
        
        \draw (1.4,3.8) -- (2+0.5,2+0.5*1.414213+0.125*1.414213);
        \draw (0.8,3.8) -- (2+0.5,2+0.5*1.414213+0.125*1.414213);
        
        \draw (2,2+0.125*1.414213) -- (3.8,1.2);
        \draw (2+1,2+0.125*1.414213) -- (3.8,1.45);
        
        \draw (0,0.125*1.414213) -- (1.8,-0.8);
        \draw (1,0.125*1.414213) -- (2,-0.8);
        
        \draw (2,0.125*1.414213) -- (3.8,-0.8);
        \draw (2+1,0.125*1.414213) -- (3.8,-0.6);
        
        \draw (2*\x,0.125*1.414213) -- (2*\x+0.4,-0.8);
        
        \draw (2*\x+0.5,2+0.5*1.414213+0.125*1.414213) -- (2*\x+0.2,3.8);
        
        \draw (0.5,2*\y+0.5*1.414213+0.125*1.414213) -- (-0.8,2*\y +0.4);
        
        \draw (2+1,2*\y+0.125*1.414213) -- (3.8,2*\y +0.6);
        
        \draw (2.8,3.8) -- (3.8,3.3);
        
        \draw (-0.8,-0.1) -- (0.1,-0.8);
        
        \node[draw,circle,inner sep=2pt,fill] at (2*\x,2*\y+0.125*1.414213) {}; 
        \node[draw,circle,inner sep=2pt,fill] at (2*\x+1,2*\y+0.125*1.414213) {}; 
        \node[draw,circle,inner sep=2pt,fill] at (2*\x+0.5,2*\y+0.5*1.414213+0.125*1.414213) {}; 
        
        \node[align=left] at (-0.3,-0.3+0.125*1.414213) {\Large\textcolor{orange}{1}};
        \node[align=left] at (1.2,-0.4+0.125*1.414213) {\Large\textcolor{orange}{2}};
        \node[align=left] at (0.4+0.5,0.5*1.414213+0.125*1.414213) {\Large\textcolor{orange}{3}};

        \draw [gray,dashed] (-0.5,-0.8) -- (-0.5,3.8);
        \draw [gray,dashed] (1.5,-0.8) -- (1.5,3.8);
        \draw [gray,dashed] (-0.8,1.5) -- (3.8,1.5);
        \draw [gray,dashed] (-0.8,-0.5) -- (3.8,-0.5);
        \draw [gray,dashed] (-0.8,3.5) -- (3.8,3.5);
        \draw [gray,dashed] (3.5,-0.8) -- (3.5,3.8);
      }
    }
\end{tikzpicture}}
    \end{tabular} & $
\begin{pmatrix}
0 & 1 & 1 & 1 & 0\\
0 & 0 & 0 & y^{-1} & y^{-1}\\
0 & xy^{-1} & y^{-1} & 0 & y^{-1}\\
\hline
1 & 0 & 1 & 0 & 1\\
0 & y^{-1} & 0 & 0 & y^{-1}\\
1 & 0 & y^{-1} & 0 & y^{-1}
\end{pmatrix}
$ & $\begin{pmatrix}
1+x\\
0\\
1+x\\
\hline
1+y\\
1+y^{-1}\\
x+xy^{-1}
\end{pmatrix}$ &
    \begin{tabular}
    {@{}c@{}} 
    Qubits/mode: 1.5\\
    Max weight: 3\\
    Max graph cost: 2 \\
    Avg. graph cost: 1.6\\
    Error detecting?: No\\
    (New here)
    \end{tabular}\\
    
    \hline
    \begin{tabular}{@{}c@{}} 
    Trunc. Sq (4)\\ 
    \scalebox{0.7}{\begin{tikzpicture}
    \foreach \x in {0,...,1}{
      \foreach \y in {0,...,1}{
        
        \draw (0.5,-0.8) -- (0.5,0);
        \draw (2.5,-0.8) -- (2.5,0);
        \draw (0.5,3) -- (0.5,3.8);
        \draw (2.5,3) -- (2.5,3.8);
        \draw (-0.8,0.5) -- (0,0.5);
        \draw (-0.8,2.5) -- (0,2.5);
        \draw (3,0.5) -- (3.8,0.5);
        \draw (3,2.5) -- (3.8,2.5);
        \draw (0.5,1) -- (0.5,2);
        \draw (2.5,1) -- (2.5,2);
        \draw (1,0.5) -- (2,0.5);
        \draw (1,2.5) -- (2,2.5);
        
        \draw (0.5,0) -- (0,0.5);
        \draw (0.5,0) -- (1,0.5);
        \draw (0.5,1) -- (1,0.5);
        \draw (0.5,1) -- (0,0.5);
        
        \draw (0.5,2) -- (0,2.5);
        \draw (0.5,2) -- (1,2.5);
        \draw (0.5,3) -- (1,2.5);
        \draw (0.5,3) -- (0,2.5);
        
        \draw (2.5,0) -- (2,0.5);
        \draw (2.5,0) -- (3,0.5);
        \draw (2.5,1) -- (3,0.5);
        \draw (2.5,1) -- (2,0.5);
        
        \draw (2.5,2) -- (2,2.5);
        \draw (2.5,2) -- (3,2.5);
        \draw (2.5,3) -- (3,2.5);
        \draw (2.5,3) -- (2,2.5);
        
        \node[draw,circle,inner sep=2pt,fill] at (0.5+2*\x,0+2*\y) {};
        \node[draw,circle,inner sep=2pt,fill] at (0+2*\x,0.50+2*\y) {};
        \node[draw,circle,inner sep=2pt,fill] at (1+2*\x,0.50+2*\y) {};
        \node[draw,circle,inner sep=2pt,fill] at (0.5+2*\x,1+2*\y) {};
        
        \node[align=left] at (0.5+0.3,0-0.3) {\Large\textcolor{orange}{1}};
        \node[align=left] at (1+0.3,0.5-0.3) {\Large\textcolor{orange}{2}};
        \node[align=left] at (0-0.3,0.5+0.3) {\Large\textcolor{orange}{3}};
        \node[align=left] at (0.5-0.3,1+0.3) {\Large\textcolor{orange}{4}};
        
        \draw [gray,dashed] (-0.5,-0.8) -- (-0.5,3.8);
        \draw [gray,dashed] (1.5,-0.8) -- (1.5,3.8);
        \draw [gray,dashed] (-0.8,1.5) -- (3.8,1.5);
        \draw [gray,dashed] (-0.8,-0.5) -- (3.8,-0.5);
        \draw [gray,dashed] (-0.8,3.5) -- (3.8,3.5);
        \draw [gray,dashed] (3.5,-0.8) -- (3.5,3.8);
      }
    }
\end{tikzpicture}}
    \end{tabular} & $
\begin{pmatrix}
0 & x & 1 & 0 & 1\\
0 & 1 & 1 & 1 & 0\\
0 & 0 & 0 & 1 & 1\\
0 & 0 & 0 & 0 & 0\\
\hline
y & 0 & 1 & 0 & 1\\
1 & 0 & 1 & 1 & 0\\
0 & 0 & 1 & 0 & 0\\
0 & 0 & 0 & 0 & 0
\end{pmatrix}
$ & $\begin{pmatrix}
y+xy\\
1+x\\
0\\
0\\
\hline
x+xy\\
1+y\\
x+y\\
0
\end{pmatrix}$  &
    \begin{tabular}
    {@{}c@{}} 
    Qubits/mode: 1.5*\\
    Max weight: 3\\
    Max graph cost: 2 \\
    Avg. graph cost: 1.8\\
    Error detecting?: No\\
    (New here)
    \end{tabular}\\
    
    \hline
    \begin{tabular}{@{}c@{}} 
    Hex bilayer (4)\\ 
    \scalebox{0.7}{\begin{tikzpicture}
    \foreach \x in {0,...,1}{
      \foreach \y in {0,...,1}{
        
        \draw (-0.8,0.2+.3333*1) -- (2.8-.3333,3.8);
        \draw (0.2-.3333,-0.8) -- (3.8,2.8+.3333);
        \draw (-0.8,2.2+.3333) -- (0.8-.3333,3.8);
        \draw (2.2-.3333,-0.8) -- (3.8,0.8+.3333);
        
        \draw (0,1-.3333) -- (1,0-.3333);
        \draw (0,3-.3333) -- (1,2-.3333);
        \draw (2,1-.3333) -- (3,0-.3333);
        \draw (2,3-.3333) -- (3,2-.3333);
        
        \draw (-0.8,0.2-0.3333) -- (2.8+.3333,3.8);
        \draw (0.2+.3333,-0.8) -- (3.8,2.8-.3333);
        \draw (-0.8,2.2-.3333) -- (0.8+.3333,3.8);
        \draw (2.2+.3333,-0.8) -- (3.8,0.8-.3333);
        
        \draw (0,1+0.3333) -- (1,0+0.3333);
        \draw (0,3+0.3333) -- (1,2+0.3333);
        \draw (2,1+0.3333) -- (3,0+0.3333);
        \draw (2,3+0.3333) -- (3,2+0.3333);
        
        \draw (2*\x+1,2*\y+0.33) -- (2*\x+1,2*\y-0.33);
        \draw (2*\x,2*\y+1.33) -- (2*\x,2*\y+0.67);
        
        \node[draw,circle,inner sep=2pt,fill] at (2*\x+1,2*\y+0.33) {};
        \node[draw,circle,inner sep=2pt,fill] at (2*\x,2*\y+1.33) {};
        \node[draw,circle,inner sep=2pt,fill] at (2*\x+1,2*\y-0.33) {};
        \node[draw,circle,inner sep=2pt,fill] at (2*\x,2*\y+0.67) {};
        
        \node[align=left] at (-0.3,1-0.3) {\Large\textcolor{orange}{1}};
        \node[align=left] at (1-0.3,0-0.3) {\Large\textcolor{orange}{2}};
        \node[align=left] at (-0.3,1.33+0.3) {\Large\textcolor{orange}{3}};
        \node[align=left] at (1,.33+0.4) {\Large\textcolor{orange}{4}};
        
        \draw [gray,dashed] (-0.5,-0.8) -- (-0.5,3.8);
        \draw [gray,dashed] (1.5,-0.8) -- (1.5,3.8);
        \draw [gray,dashed] (-0.8,1.5) -- (3.8,1.5);
        \draw [gray,dashed] (-0.8,-0.5) -- (3.8,-0.5);
        \draw [gray,dashed] (-0.8,3.5) -- (3.8,3.5);
        \draw [gray,dashed] (3.5,-0.8) -- (3.5,3.8);
      }
    }
\end{tikzpicture}}
    \end{tabular} & $
\begin{pmatrix}
0 & x & 1 & 0 & 1\\
0 & 1 & 1 & 1 & 0\\
0 & 0 & 0 & 1 & 1\\
0 & 0 & 0 & 0 & 0\\
\hline
y & 0 & 1 & 0 & 1\\
1 & 0 & 1 & 1 & 0\\
0 & 0 & 1 & 0 & 0\\
0 & 0 & 0 & 0 & 0
\end{pmatrix}
$ & $\begin{pmatrix}
y+xy\\
1+x\\
0\\
0\\
\hline
x+xy\\
1+y\\
x+y\\
0
\end{pmatrix}$  &
    \begin{tabular}
    {@{}c@{}} 
    Qubits/mode: 1.5*\\
    Max weight: 3\\
    Max graph cost: 2 \\
    Avg. graph cost: 1.4\\
    Error detecting?: No\\
    (New here)
    \end{tabular}\\

\end{tabular}
\clearpage
\begin{tabular}{c|c|c|c}
    Fermionic unit cell & Fermionic operators & \\
    \hline
    \begin{tabular}{@{}c@{}} 
    Hex (2)\\ 
    \scalebox{0.7}{\begin{tikzpicture}
    \foreach \x in {0,...,1}{
      \foreach \y in {0,...,1}{
        
        \draw (2.2,3.8) -- (3.8,2.2);
        
        \draw (-0.8,2.8) -- (2.8,-0.8);
        
        \draw (0.2,3.8) -- (3.8,0.2);
        
        \draw (-0.8,0.8) -- (0.8,-0.8);
        
        \draw (2*\x,2*\y) -- (2*\x+1,2*\y+1);
        
        \node[draw,circle,inner sep=2pt,fill] at (2*\x,2*\y) {};
        \node[draw,circle,inner sep=2pt,fill] at (2*\x+1,2*\y+1) {};
        
        \node[align=left] at (-0.3,-0.3) {\Large\textcolor{orange}{1}};
        \node[align=left] at (1+0.3,1+0.3) {\Large\textcolor{orange}{2}};
        
        \draw [gray,dashed] (-0.5,-0.8) -- (-0.5,3.8);
        \draw [gray,dashed] (1.5,-0.8) -- (1.5,3.8);
        \draw [gray,dashed] (-0.8,1.5) -- (3.8,1.5);
        \draw [gray,dashed] (-0.8,-0.5) -- (3.8,-0.5);
        \draw [gray,dashed] (-0.8,3.5) -- (3.8,3.5);
        \draw [gray,dashed] (3.5,-0.8) -- (3.5,3.8);
      }
    }
\end{tikzpicture}}
    \end{tabular} & $\begin{pmatrix}
    y & x & 1 & 0 & 1  \\
    1 & 1 & 1 & 1 & 0  \\
    \hline
    0 & 0 & 0 & 0 & 1  \\
    0 & 0 & 0 & 1 & 0  
\end{pmatrix}$ & (cont.)
    \\
    \hline
    Qubit unit cell & Encoded operators & Stabilizer & Encoding properties\\
    
    \hline
    \begin{tabular}{@{}c@{}} 
    Sq bilayer (4)\\ 
    \scalebox{0.7}{\begin{tikzpicture}
    \foreach \x in {0,...,1}{
      \foreach \y in {0,...,1}{
        
        \draw (2*\x+0,-0.8) -- (2*\x+0,3.8);
        \draw (2*\x+.33,-0.8) -- (2*\x+.33,3.8);
        \draw (2*\x+.67,-0.8) -- (2*\x+.67,3.8);
        \draw (2*\x+1,-0.8) -- (2*\x+1,3.8);
        \draw (-0.8,2*\y+0) -- (3.8,2*\y+0);
        \draw (-0.8,2*\y+1) -- (3.8,2*\y+1);
        
        \draw (2*\x+.33,2*\y) -- (2*\x,2*\y+1);
        \draw (2*\x+1,2*\y) -- (2*\x+0.67,2*\y+1);
        
        \node[draw,circle,inner sep=2pt,fill] at (2*\x+.33,2*\y) {};
        \node[draw,circle,inner sep=2pt,fill] at (2*\x+1,2*\y) {};
        \node[draw,circle,inner sep=2pt,fill] at (2*\x,2*\y+1) {};
        \node[draw,circle,inner sep=2pt,fill] at (2*\x+0.67,2*\y+1) {};
        
        \node[align=left] at (0.33-0.25,0-0.3) {\Large\textcolor{orange}{1}};
        \node[align=left] at (1+0.3,0-0.3) {\Large\textcolor{orange}{2}};
        \node[align=left] at (0-0.3,1+0.3) {\Large\textcolor{orange}{3}};
        \node[align=left] at (1.2,1+0.3) {\Large\textcolor{orange}{4}};
        
        \draw [gray,dashed] (-0.5,-0.8) -- (-0.5,3.8);
        \draw [gray,dashed] (1.5,-0.8) -- (1.5,3.8);
        \draw [gray,dashed] (-0.8,1.5) -- (3.8,1.5);
        \draw [gray,dashed] (-0.8,-0.5) -- (3.8,-0.5);
        \draw [gray,dashed] (-0.8,3.5) -- (3.8,3.5);
        \draw [gray,dashed] (3.5,-0.8) -- (3.5,3.8);
      }
    }
\end{tikzpicture}}
    \end{tabular} & $
\begin{pmatrix}
0 & 1 & 1 & 1 & 0\\
0 & 1 & x^{-1} & 0 & x^{-1}\\
0 & 0 & 0 & 1 & 1\\
0 & 0 & 0 & 0 & 0\\
\hline
1 & 0 & 1 & 1 & 0\\
x^{-1}y & 0 & x^{-1} & 0 & x^{-1}\\
0 & 0 & 1 & 0 & 0\\
0 & 0 & 0 & 0 & 0 & 
\end{pmatrix}
$ & $\begin{pmatrix}
1+x\\
y+x^{-1}y\\
0\\
0\\
\hline
1+y\\
1+y\\
x+y\\
0
\end{pmatrix}$  &
    \begin{tabular}
    {@{}c@{}} 
    Qubits/mode: 1.5*\\
    Max weight: 3\\
    Max graph cost: 2 \\
    Avg. graph cost: 1.6\\
    Error detecting?: No\\
    (New here)
    \end{tabular}\\
    
    \hline
    \begin{tabular}{@{}c@{}} 
    Snub sq (4)\\ 
    \scalebox{0.7}{\begin{tikzpicture}
    \foreach \x in {0,...,1}{
      \foreach \y in {0,...,1}{
        
        \draw (1,0) -- (2,1);
        \draw (1,2) -- (2,3);
        \draw (-0.8,0.2) -- (0,1);
        \draw (-0.8,2.2) -- (0,3);
        \draw (3,0) -- (3.8,0.8);
        \draw (3,2) -- (3.8,2.8);
        
        \draw (1,1) -- (0,2);
        \draw (3,1) -- (2,2);
        \draw (0.8,-0.8) -- (0,0);
        \draw (2.8,-0.8) -- (2,0);
        \draw (1,3) -- (0.2,3.8);
        \draw (3,3) -- (2.2,3.8);
        
        \draw (-0.8,2*\y) -- (3.8,2*\y);
        \draw (-0.8,2*\y+1) -- (3.8,2*\y+1);
        \draw (2*\x,-0.8) -- (2*\x,3.8);
        \draw (2*\x+1,-0.8) -- (2*\x+1,3.8);
        
        \node[draw,circle,inner sep=2pt,fill] at (2*\x,2*\y) {};
        \node[draw,circle,inner sep=2pt,fill] at (2*\x,2*\y+1) {};
        \node[draw,circle,inner sep=2pt,fill] at (2*\x+1,2*\y) {};
        \node[draw,circle,inner sep=2pt,fill] at (2*\x+1,2*\y+1) {};
        
        \node[align=left] at (0-0.3,0.3) {\Large\textcolor{orange}{1}};
        \node[align=left] at (1-0.3,0+0.3) {\Large\textcolor{orange}{2}};
        \node[align=left] at (0-0.3,1+0.3) {\Large\textcolor{orange}{3}};
        \node[align=left] at (1+0.3,1+0.3) {\Large\textcolor{orange}{4}};
        
        \draw [gray,dashed] (-0.5,-0.8) -- (-0.5,3.8);
        \draw [gray,dashed] (1.5,-0.8) -- (1.5,3.8);
        \draw [gray,dashed] (-0.8,1.5) -- (3.8,1.5);
        \draw [gray,dashed] (-0.8,-0.5) -- (3.8,-0.5);
        \draw [gray,dashed] (-0.8,3.5) -- (3.8,3.5);
        \draw [gray,dashed] (3.5,-0.8) -- (3.5,3.8);
    
      }
    }
\end{tikzpicture}}
    \end{tabular} & $
\begin{pmatrix}
0 & x & 1 & 0 & 0\\
0 & 1 & 1 & 0 & 0\\
0 & 0 & 0 & 0 & 0\\
0 & 0 & 0 & 1 & y^{-1}\\
\hline
y & 0 & 1 & 0 & 1\\
1 & 0 & 1 & 1 & 0\\
0 & 0 & 0 & 0 & 0\\
1 & 0 & 0 & 0 & 0
\end{pmatrix}
$ & $\begin{pmatrix}
y+xy\\
1+x\\
0\\
0\\
\hline
x+xy\\
1+y\\
0\\
1+x
\end{pmatrix}$  &
    \begin{tabular}
    {@{}c@{}} 
    Qubits/mode: 1.5*\\
    Max weight: 3\\
    Max graph cost: 2 \\
    Avg. graph cost: 1.2\\
    Error detecting?: No\\
    (New here)
    \end{tabular}\\
    
    \hline
    \begin{tabular}{@{}c@{}} 
    Heavy hex (5)\\ 
    \scalebox{0.7}{\begin{tikzpicture}
    \foreach \x in {0,...,1}{
      \foreach \y in {0,...,1}{
        
        
        \draw (2*\x+0.3,2*\y-0.3) -- (2*\x+1.3,2*\y+0.7);
        
        \draw (0.8,-0.8) -- (-0.8,0.8);
        \draw (2.8,-0.8) -- (-0.8,2.8);
        \draw (3.8,0.2) -- (0.2,3.8);
        \draw (3.8,2.2) -- (2.2,3.8);
        
        \node[draw,circle,inner sep=2pt,fill] at (2*\x+0.8,2*\y+0.2) {};
        \node[draw,circle,inner sep=2pt,fill] at (2*\x+0.3,2*\y-0.3) {};
        \node[draw,circle,inner sep=2pt,fill] at (2*\x+1.3,2*\y+0.7) {};
        \node[draw,circle,inner sep=2pt,fill] at (2*\x+0.8,2*\y+1.2) {};
        \node[draw,circle,inner sep=2pt,fill] at (2*\x-0.2,2*\y+0.2) {};
        
        \node[align=left] at (-0.2,0.6) {\Large\textcolor{orange}{1}};
        \node[align=left] at (0,-0.4) {\Large\textcolor{orange}{2}};
        \node[align=left] at (1.2,-0.1) {\Large\textcolor{orange}{3}};
        \node[align=left] at (1.7,1) {\Large\textcolor{orange}{4}};
        \node[align=left] at (0.5,1) {\Large\textcolor{orange}{5}};

        \draw [gray,dashed] (-0.5,-0.8) -- (-0.5,3.8);
        \draw [gray,dashed] (1.5,-0.8) -- (1.5,3.8);
        \draw [gray,dashed] (-0.8,1.5) -- (3.8,1.5);
        \draw [gray,dashed] (-0.8,-0.5) -- (3.8,-0.5);
        \draw [gray,dashed] (-0.8,3.5) -- (3.8,3.5);
        \draw [gray,dashed] (3.5,-0.8) -- (3.5,3.8);
      }
    }
\end{tikzpicture}}
    \end{tabular} & $
\begin{pmatrix}
0 & 0 & 0 & 0 & 0\\
0 & x & 1 & 0 & 1\\
0 & 0 & 0 & 1 & 1\\
0 & 1 & 1 & 1 & 0\\
0 & 0 & 0 & 0 & 0\\
\hline
0 & 0 & 0 & 0 & 0\\
y & 0 & 1 & 0 & 1\\
0 & 0 & 1 & 0 & 0\\
1 & 0 & 1 & 1 & 0\\
0 & 0 & 0 & 0 & 0
\end{pmatrix}
$ & $\begin{pmatrix}
0\\
y+xy\\
0\\
1+x\\
0\\
\hline
0\\
x+xy\\
x+y\\
1+y\\
0
\end{pmatrix}$  &
    \begin{tabular}
    {@{}c@{}} 
    Qubits/mode: 1.5*\\
    Max weight: 3\\
    Max graph cost: 2 \\
    Avg. graph cost: 1.6\\
    Error detecting?: No\\
    (New here)
    \end{tabular}\\

\end{tabular}

\end{center}
\clearpage
\begin{center}
\begin{tabular}{c|c|c|c}
    Fermionic unit cell & Fermionic operators & \\
    \hline
    \begin{tabular}{@{}c@{}} 
    Tilted sq (2)\\ 
    \scalebox{0.7}{\begin{tikzpicture}
    \foreach \x in {0,...,1}{
      \foreach \y in {0,...,1}{
        
        \draw (-0.8,0.2) -- (2.8,3.8);
        \draw (0.2,-0.8) -- (3.8,2.8);
        \draw (1.8,-0.8) -- (-0.8,1.8);
        \draw (3.8,-0.8) -- (-0.8,3.8);
        \draw (3.8,-0.8) -- (-0.8,3.8);
        \draw (3.8,1.2) -- (1.2,3.8);
        \draw (-0.8,2.2) -- (0.8,3.8);
        \draw (2.2,-0.8) -- (3.8,0.8);
        \node[draw,circle,inner sep=2pt,fill] at (2*\x+1,2*\y) {};
        \node[draw,circle,inner sep=2pt,fill] at (2*\x,2*\y+1) {};
        
        \node[align=left] at (0.3,1+0.3) {\Large\textcolor{orange}{1}};
        \node[align=left] at (1+0.3,0.3) {\Large\textcolor{orange}{2}};
        
        \draw [gray,dashed] (-0.5,-0.8) -- (-0.5,3.8);
        \draw [gray,dashed] (1.5,-0.8) -- (1.5,3.8);
        \draw [gray,dashed] (-0.8,1.5) -- (3.8,1.5);
        \draw [gray,dashed] (-0.8,-0.5) -- (3.8,-0.5);
        \draw [gray,dashed] (-0.8,3.5) -- (3.8,3.5);
        \draw [gray,dashed] (3.5,-0.8) -- (3.5,3.8);
      }
    }
\end{tikzpicture}}
    \end{tabular} & $\begin{pmatrix}
    1 & x & 1 & 1 & 0 & 1 \\
    x^{-1}y & 1 & y & 1 & 1 & 0 \\
    \hline
    0 & 0 & 0 & 0 & 0 & 1 \\
    0 & 0 & 0 & 0 & 1 & 0 
\end{pmatrix}$ &
    \\
    \hline
    Qubit unit cell & Encoded operators & Stabilizer & Encoding properties\\
    
    \hline
    \begin{tabular}{@{}c@{}} 
    Lieb lattice (4)\\ 
    \scalebox{0.7}{\begin{tikzpicture}
    \foreach \x in {0,...,1}{
      \foreach \y in {0,...,1}{
        
        \draw (0,-0.8) -- (0,3.8);
        \draw (2,-0.8) -- (2,3.8);
        \draw (-0.8,0) -- (3.8,0);
        \draw (-0.8,2) -- (3.8,2);
        \node[draw,circle,inner sep=2pt,fill] at (2*\x,2*\y) {};
        \node[draw,circle,inner sep=2pt,fill] at (2*\x+1,2*\y) {};
        \node[draw,circle,inner sep=2pt,fill] at (2*\x,2*\y+1) {};
        
        \node[align=left] at (0+0.3,1+0.3) {\Large\textcolor{orange}{1}};
        \node[align=left] at (0+0.3,0+0.3) {\Large\textcolor{orange}{2}};
        \node[align=left] at (1+0.3,0+0.3) {\Large\textcolor{orange}{3}};
        
        \draw [gray,dashed] (-0.5,-0.8) -- (-0.5,3.8);
        \draw [gray,dashed] (1.5,-0.8) -- (1.5,3.8);
        \draw [gray,dashed] (-0.8,1.5) -- (3.8,1.5);
        \draw [gray,dashed] (-0.8,-0.5) -- (3.8,-0.5);
        \draw [gray,dashed] (-0.8,3.5) -- (3.8,3.5);
        \draw [gray,dashed] (3.5,-0.8) -- (3.5,3.8);
      }
    }
\end{tikzpicture}}
    \end{tabular} & $
\begin{pmatrix}
0 & x & 0 & 1 & 0 & 1\\
0 & x & y & 0 & 0 & 0\\
0 & 0 & y & 1 & 1 & 0\\
\hline
1 & 0 & 1 & 0 & 0 & 1\\
y & 0 & 0 & 1 & 0 & 0\\
x^{-1}y & 1 & 0 & 0 & 1 & 0 
\end{pmatrix}
$& $\begin{pmatrix}
1+x\\
x+y\\
1+y\\
\hline
1+x\\
1+xy\\
1+y
\end{pmatrix}$ &
    \begin{tabular}
    {@{}c@{}} 
    Qubits/mode: 1.5\\
    Max weight: 3\\
    Max graph cost: 2 \\
    Avg. graph cost: 1.33\\
    Error detecting?: No\\
    (Compact encoding \cite{derby2021compact})
    \end{tabular}\\

\end{tabular}

\end{center}
\clearpage
\begin{center}
\begin{tabular}{c|c|c|c}
    Fermionic unit cell & Fermionic operators & \\
    \hline
    \begin{tabular}{@{}c@{}} 
    Kagome (3)\\ 
    \scalebox{0.7}{\begin{tikzpicture}
    \foreach \x in {0,...,1}{
      \foreach \y in {0,...,1}{
        
        \draw (0,-0.8) -- (0,3.8);
        \draw (2,-0.8) -- (2,3.8);
        \draw (-0.8,1) -- (3.8,1);
        \draw (-0.8,3) -- (3.8,3);
        \draw (-0.8,-0.8) -- (3.8,3.8);
        \draw (-0.8,1.2) -- (1.8,3.8);
        \draw (1.2,-0.8) -- (3.8,1.8);
        
        \node[draw,circle,inner sep=2pt,fill] at (2*\x,2*\y) {};
        \node[draw,circle,inner sep=2pt,fill] at (2*\x,2*\y+1) {};
        \node[draw,circle,inner sep=2pt,fill] at (2*\x+1,2*\y+1) {};
        
        \node[align=left] at (0+0.3,0) {\Large\textcolor{orange}{1}};
        \node[align=left] at (1,1-0.3) {\Large\textcolor{orange}{2}};
        \node[align=left] at (0+0.3,1+0.3) {\Large\textcolor{orange}{3}};
        
        \draw [gray,dashed] (-0.5,-0.8) -- (-0.5,3.8);
        \draw [gray,dashed] (1.5,-0.8) -- (1.5,3.8);
        \draw [gray,dashed] (-0.8,1.5) -- (3.8,1.5);
        \draw [gray,dashed] (-0.8,-0.5) -- (3.8,-0.5);
        \draw [gray,dashed] (-0.8,3.5) -- (3.8,3.5);
        \draw [gray,dashed] (3.5,-0.8) -- (3.5,3.8);
      }
    }
\end{tikzpicture}}
    \end{tabular} & $\begin{pmatrix}
    0 & 0 & xy & 1 & y & 1 & 0 & 0 & 1 \\
    x & 1 & 0 & 0 & 1 & 1 & 0 & 1 & 0 \\
    1 & 1 & 1 & 1 & 0 & 0 & 1 & 0 & 0 \\
    \hline
    0 & 0 & 0 & 0 & 0 & 0 & 0 & 0 & 1 \\
    0 & 0 & 0 & 0 & 0 & 0 & 0 & 1 & 0 \\
    0 & 0 & 0 & 0 & 0 & 0 & 1 & 0 & 0 
\end{pmatrix}$ &
    \\
    \hline
    Qubit unit cell & Encoded operators & Stabilizer & Encoding properties\\
    
    \hline
    \begin{tabular}{@{}c@{}} 
    Trunc. Sq. (4)\\ 
    \scalebox{0.7}{\begin{tikzpicture}
    \foreach \x in {0,...,1}{
      \foreach \y in {0,...,1}{
        
        \draw (0.5,-0.8) -- (0.5,0);
        \draw (2.5,-0.8) -- (2.5,0);
        \draw (0.5,3) -- (0.5,3.8);
        \draw (2.5,3) -- (2.5,3.8);
        \draw (-0.8,0.5) -- (0,0.5);
        \draw (-0.8,2.5) -- (0,2.5);
        \draw (3,0.5) -- (3.8,0.5);
        \draw (3,2.5) -- (3.8,2.5);
        \draw (0.5,1) -- (0.5,2);
        \draw (2.5,1) -- (2.5,2);
        \draw (1,0.5) -- (2,0.5);
        \draw (1,2.5) -- (2,2.5);
        
        \draw (0.5,0) -- (0,0.5);
        \draw (0.5,0) -- (1,0.5);
        \draw (0.5,1) -- (1,0.5);
        \draw (0.5,1) -- (0,0.5);
        
        \draw (0.5,2) -- (0,2.5);
        \draw (0.5,2) -- (1,2.5);
        \draw (0.5,3) -- (1,2.5);
        \draw (0.5,3) -- (0,2.5);
        
        \draw (2.5,0) -- (2,0.5);
        \draw (2.5,0) -- (3,0.5);
        \draw (2.5,1) -- (3,0.5);
        \draw (2.5,1) -- (2,0.5);
        
        \draw (2.5,2) -- (2,2.5);
        \draw (2.5,2) -- (3,2.5);
        \draw (2.5,3) -- (3,2.5);
        \draw (2.5,3) -- (2,2.5);
        
        \node[draw,circle,inner sep=2pt,fill] at (0.5+2*\x,0+2*\y) {};
        \node[draw,circle,inner sep=2pt,fill] at (0+2*\x,0.50+2*\y) {};
        \node[draw,circle,inner sep=2pt,fill] at (1+2*\x,0.50+2*\y) {};
        \node[draw,circle,inner sep=2pt,fill] at (0.5+2*\x,1+2*\y) {};
        
        \node[align=left] at (0.5+0.3,0-0.3) {\Large\textcolor{orange}{1}};
        \node[align=left] at (1+0.3,0.5-0.3) {\Large\textcolor{orange}{2}};
        \node[align=left] at (0-0.3,0.5+0.3) {\Large\textcolor{orange}{3}};
        \node[align=left] at (0.5-0.3,1+0.3) {\Large\textcolor{orange}{4}};
        
        \draw [gray,dashed] (-0.5,-0.8) -- (-0.5,3.8);
        \draw [gray,dashed] (1.5,-0.8) -- (1.5,3.8);
        \draw [gray,dashed] (-0.8,1.5) -- (3.8,1.5);
        \draw [gray,dashed] (-0.8,-0.5) -- (3.8,-0.5);
        \draw [gray,dashed] (-0.8,3.5) -- (3.8,3.5);
        \draw [gray,dashed] (3.5,-0.8) -- (3.5,3.8);
      }
    }
\end{tikzpicture}}
    \end{tabular} & $
\begin{pmatrix}
0 & 0 & 0 & 1 & 0 & 1 & 1 & 1 & 1\\
0 & 1 & 1 & 1 & x^{-1} & 0 & 1 & 0 & 0\\
0 & 1 & 0 & 0 & 0 & 1 & 0 & 1 & 0\\
0 & 0 & 0 & y^{-1} & 0 & y^{-1} & 0 & 0 & y^{-1}\\
\hline
0 & 1 & 0 & 0 & 0 & 1 & 1 & 1 & 1\\
1 & 0 & 1 & 0 & 0 & 0 & 0 & 0 & 0\\
x & 0 & 0 & 0 & 1 & 0 & 0 & 0 & 1\\
0 & 0 & x & 0 & 1 & 0 & 0 & y^{-1} & 0
\end{pmatrix}
$& $\begin{pmatrix}
1+x\\
1+x^{-1}y^{-1}\\
x+y^{-1}\\
y^{-1}+xy^{-1}\\
\hline
x+y^{-1}\\
1+y^{-1}\\
x+y^{-1}\\
y^{-1}+xy^{-1}
\end{pmatrix}$  &
    \begin{tabular}
    {@{}c@{}} 
    Qubits/mode: 1.33\\
    Max weight: 3\\
    Max graph cost: 2 \\
    Avg. graph cost: 1.78\\
    Error detecting?: No\\
    (New here)
    \end{tabular}\\

\end{tabular}

\end{center}
\clearpage
\begin{center}
\begin{tabular}{c|c|c|c}
    Fermionic unit cell & Fermionic operators & \\
    \hline
    \begin{tabular}{@{}c@{}} 
    Kagome alt. (3)\\ 
    \scalebox{0.7}{\begin{tikzpicture}
    \foreach \x in {0,...,1}{
      \foreach \y in {0,...,1}{
        
        \draw (1,-0.8) -- (1,3.8);
        \draw (3,-0.8) -- (3,3.8);
        \draw (-0.8,0) -- (3.8,0);
        \draw (-0.8,2) -- (3.8,2);
        \draw (-0.8,2.8) -- (2.8,-0.8);
        \draw (3.8,0.2) -- (0.2,3.8);
        \draw (0.8,-0.8) -- (-0.8,0.8);
        \draw (3.8,2.2) -- (2.2,3.8);
        
        \node[draw,circle,inner sep=2pt,fill] at (2*\x,2*\y) {};
        \node[draw,circle,inner sep=2pt,fill] at (2*\x+1,2*\y) {};
        \node[draw,circle,inner sep=2pt,fill] at (2*\x+1,2*\y+1) {};
        
        \node[align=left] at (-0.1,0.4) {\Large\textcolor{orange}{1}};
        \node[align=left] at (1-0.3,0.3) {\Large\textcolor{orange}{2}};
        \node[align=left] at (1-0.4,1) {\Large\textcolor{orange}{3}};
        
        \draw [gray,dashed] (-0.5,-0.8) -- (-0.5,3.8);
        \draw [gray,dashed] (1.5,-0.8) -- (1.5,3.8);
        \draw [gray,dashed] (-0.8,1.5) -- (3.8,1.5);
        \draw [gray,dashed] (-0.8,-0.5) -- (3.8,-0.5);
        \draw [gray,dashed] (-0.8,3.5) -- (3.8,3.5);
        \draw [gray,dashed] (3.5,-0.8) -- (3.5,3.8);
      }
    }
\end{tikzpicture}}
    \end{tabular} & $\begin{pmatrix}
    1 & x & 0 & 0 & x & 1 & 0 & 0 & 1 \\
    0 & 0 & 1 & 1 & 1 & 1 & 0 & 1 & 0 \\
    y^{-1} & 1 & y^{-1} & 1 & 0 & 0 & 1 & 0 & 0 \\
    \hline
    0 & 0 & 0 & 0 & 0 & 0 & 0 & 0 & 1 \\
    0 & 0 & 0 & 0 & 0 & 0 & 0 & 1 & 0 \\
    0 & 0 & 0 & 0 & 0 & 0 & 1 & 0 & 0 
\end{pmatrix}$ &
    \\
    \hline
    Qubit unit cell & Encoded operators & Stabilizer & Encoding properties\\
    
    \hline
    \begin{tabular}{@{}c@{}} 
    Trunc. Sq. (4)\\ 
    \scalebox{0.7}{\begin{tikzpicture}
    \foreach \x in {0,...,1}{
      \foreach \y in {0,...,1}{
        
        \draw (0.5,-0.8) -- (0.5,0);
        \draw (2.5,-0.8) -- (2.5,0);
        \draw (0.5,3) -- (0.5,3.8);
        \draw (2.5,3) -- (2.5,3.8);
        \draw (-0.8,0.5) -- (0,0.5);
        \draw (-0.8,2.5) -- (0,2.5);
        \draw (3,0.5) -- (3.8,0.5);
        \draw (3,2.5) -- (3.8,2.5);
        \draw (0.5,1) -- (0.5,2);
        \draw (2.5,1) -- (2.5,2);
        \draw (1,0.5) -- (2,0.5);
        \draw (1,2.5) -- (2,2.5);
        
        \draw (0.5,0) -- (0,0.5);
        \draw (0.5,0) -- (1,0.5);
        \draw (0.5,1) -- (1,0.5);
        \draw (0.5,1) -- (0,0.5);
        
        \draw (0.5,2) -- (0,2.5);
        \draw (0.5,2) -- (1,2.5);
        \draw (0.5,3) -- (1,2.5);
        \draw (0.5,3) -- (0,2.5);
        
        \draw (2.5,0) -- (2,0.5);
        \draw (2.5,0) -- (3,0.5);
        \draw (2.5,1) -- (3,0.5);
        \draw (2.5,1) -- (2,0.5);
        
        \draw (2.5,2) -- (2,2.5);
        \draw (2.5,2) -- (3,2.5);
        \draw (2.5,3) -- (3,2.5);
        \draw (2.5,3) -- (2,2.5);
        
        \node[draw,circle,inner sep=2pt,fill] at (0.5+2*\x,0+2*\y) {};
        \node[draw,circle,inner sep=2pt,fill] at (0+2*\x,0.50+2*\y) {};
        \node[draw,circle,inner sep=2pt,fill] at (1+2*\x,0.50+2*\y) {};
        \node[draw,circle,inner sep=2pt,fill] at (0.5+2*\x,1+2*\y) {};
        
        \node[align=left] at (0.5+0.3,0-0.3) {\Large\textcolor{orange}{1}};
        \node[align=left] at (1+0.3,0.5-0.3) {\Large\textcolor{orange}{2}};
        \node[align=left] at (0-0.3,0.5+0.3) {\Large\textcolor{orange}{3}};
        \node[align=left] at (0.5-0.3,1+0.3) {\Large\textcolor{orange}{4}};
        
        \draw [gray,dashed] (-0.5,-0.8) -- (-0.5,3.8);
        \draw [gray,dashed] (1.5,-0.8) -- (1.5,3.8);
        \draw [gray,dashed] (-0.8,1.5) -- (3.8,1.5);
        \draw [gray,dashed] (-0.8,-0.5) -- (3.8,-0.5);
        \draw [gray,dashed] (-0.8,3.5) -- (3.8,3.5);
        \draw [gray,dashed] (3.5,-0.8) -- (3.5,3.8);
      }
    }
\end{tikzpicture}}
    \end{tabular} & $
\begin{pmatrix}
0 & y & 1 & y & 0 & 1 & y & 1 & 1\\
0 & 1 & 0 & 0 & 1 & 0 & 0 & x^{-1} & 0\\
0 & 0 & 0 & 1 & 1 & 0 & 0 & 0 & 1\\
0 & 0 & 0 & 1 & 1 & 0 & 1 & 0 & 0\\
\hline
1 & 0 & 1 & 0 & 0 & 0 & 0 & 0 & 0\\
x^{-1} & 0 & 0 & 0 & 0 & x^{-1} & 0 & 0 & x^{-1}\\
0 & 0 & 1 & 0 & 0 & 1 & 0 & 1 & 0\\
0 & 1 & 0 & 0 & 1 & 0 & 1 & 0 & 0
\end{pmatrix}
$  & $\begin{pmatrix}
x+xy\\
1+y\\
x+y\\
x+y\\
\hline
x+xy\\
1+y\\
x+y\\
1+y
\end{pmatrix}$ &
    \begin{tabular}
    {@{}c@{}} 
    Qubits/mode: 1.33\\
    Max weight: 3\\
    Max graph cost: 2 \\
    Avg. graph cost: 1.78\\
    Error detecting?: No\\
    (New here)
    \end{tabular}\\
    
    \hline
    \begin{tabular}{@{}c@{}} 
    Heavy hex (5)\\ 
    \scalebox{0.7}{\begin{tikzpicture}
    \foreach \x in {0,...,1}{
      \foreach \y in {0,...,1}{
        
        
        \draw (2*\x+0.3,2*\y-0.3) -- (2*\x+1.3,2*\y+0.7);
        
        \draw (0.8,-0.8) -- (-0.8,0.8);
        \draw (2.8,-0.8) -- (-0.8,2.8);
        \draw (3.8,0.2) -- (0.2,3.8);
        \draw (3.8,2.2) -- (2.2,3.8);
        
        \node[draw,circle,inner sep=2pt,fill] at (2*\x+0.8,2*\y+0.2) {};
        \node[draw,circle,inner sep=2pt,fill] at (2*\x+0.3,2*\y-0.3) {};
        \node[draw,circle,inner sep=2pt,fill] at (2*\x+1.3,2*\y+0.7) {};
        \node[draw,circle,inner sep=2pt,fill] at (2*\x+0.8,2*\y+1.2) {};
        \node[draw,circle,inner sep=2pt,fill] at (2*\x-0.2,2*\y+0.2) {};
        
        \node[align=left] at (-0.2,0.6) {\Large\textcolor{orange}{1}};
        \node[align=left] at (0,-0.4) {\Large\textcolor{orange}{2}};
        \node[align=left] at (1.2,-0.1) {\Large\textcolor{orange}{3}};
        \node[align=left] at (1.7,1) {\Large\textcolor{orange}{4}};
        \node[align=left] at (0.5,1) {\Large\textcolor{orange}{5}};

        \draw [gray,dashed] (-0.5,-0.8) -- (-0.5,3.8);
        \draw [gray,dashed] (1.5,-0.8) -- (1.5,3.8);
        \draw [gray,dashed] (-0.8,1.5) -- (3.8,1.5);
        \draw [gray,dashed] (-0.8,-0.5) -- (3.8,-0.5);
        \draw [gray,dashed] (-0.8,3.5) -- (3.8,3.5);
        \draw [gray,dashed] (3.5,-0.8) -- (3.5,3.8);
      }
    }
\end{tikzpicture}}
    \end{tabular} & $
\begin{pmatrix}
0 & x & 0 & 0 & x & 1 & 0 & 0 & 1\\
0 & 0 & 0 & 0 & 0 & 0 & 0 & 1 & 1\\
0 & 0 & 0 & 1 & 1 & 1 & 0 & 1 & 0\\
0 & 0 & 0 & 1 & 1 & 0 & 1 & 0 & 0\\
0 & 1 & y^{-1} & 1 & 0 & 0 & 1 & 0 & 0\\
\hline
1 & 0 & 0 & 0 & 0 & 1 & 0 & 0 & 1\\
0 & 0 & 0 & 0 & 0 & 1 & 0 & 0 & 0\\
0 & 0 & 1 & 0 & 0 & 1 & 0 & 1 & 0\\
0 & 1 & 0 & 0 & 1 & 0 & 1 & 0 & 0\\
y^{-1} & 0 & y^{-1} & 0 & 0 & 0 & 0 & 0 & 0
\end{pmatrix}
$ &$\begin{pmatrix}
1 & xy\\
0 & 0\\
1 & y\\
0 & x+y\\
y^{-1} & x\\
\hline
0 & x + xy\\
1 & x\\
0 & x+y\\
0 & 1+y\\
0 & 1+x
\end{pmatrix}$ & 
    \begin{tabular}
    {@{}c@{}} 
    Qubits/mode: 1.67\\
    Max weight: 3\\
    Max graph cost: 2 \\
    Avg. graph cost: 1.67\\
    Error detecting?: No\\
    (New here)
    \end{tabular}\\
    
\end{tabular}

\end{center}
\restoregeometry

\onecolumn

\section{Proof that $\sigma_V$ is a faithful representation of $F/G$}\label{sigmav_faithful}

\begin{theorem}
The map $\sigma_V$, as defined in equation (\ref{projected_sigma}), is a faithful representation of $F/G$.
\end{theorem}

\begin{proof}
Note that $F/G \simeq \Gamma/(\ker(\tau)\ker(\sigma))$, and that $\ker(\tau)\ker(\sigma) \subseteq \ker(\sigma_V)$, so $\sigma_V$ is a rep. of $F/G$. In order to show that $\sigma_V$ is a faithful representation of $F/G$, we need only show that $\ker(\tau)\ker(\sigma) = \ker(\sigma_V)$.
Note that $$\ker(\sigma_V) = \ker(\Pi_V \circ \sigma)= \sigma^{-1}(\ker(\Pi_V))$$ so we need to show that
$$\sigma(\ker(\tau)\ker(\sigma))=S=\ker(\Pi_V),$$ i.e. there does not exist a Pauli operator $p \in \ker(\Pi_V)$ such that $p \not \in S$. 

Suppose such a Pauli $p$ exists, then $\Pi_V p = \Pi_V$. Expanding the projector $\Pi_V$:
$$ \frac{1}{|S|} \sum_{s \in S } s p =\frac{1}{|S|} \sum_{s' \in S } s'$$
$$\forall q \in S \;:\; \textrm{Tr}[\sum_{s \in S } q s p] =\textrm{Tr}[ \sum_{s' \in S } qs']$$
given that $-1 \not \in S$ and $S$ is hermitian, we can use the orthogonality of the trace inner product on Pauli operators
$$\forall q \in S \;:\; \sum_{s \in S } \textrm{Tr}[ q s p]=1 $$
which implies $\exists q,s \in S$ s.t. $qs=p$
and so $p\in S$, which is a contradiction.

\end{proof}

\begin{corollary}
 $\sigma_V$ is a faithful representation of $F$ iff $G=\{1\}$, i.e. $\ker(\sigma)\subseteq \ker(\tau)$
\end{corollary}

\section{Proof of equations (\ref{ker_sigma}) and (\ref{ker_tau})} \label{proof_ker_sigma_tau}

\begin{theorem}
$\ker(\sigma) = \{\sigma(\Gamma(b))^* \Gamma(b) \vert b \in \ker(\hat{\sigma}) \}$
\end{theorem}

\begin{proof}
The cyclic group $C:=\{\pm 1, \pm i\}$ is a normal subgroup of both $\Gamma$ and $P$. The quotient groups $\Gamma/C$ and $P/C$ are isomorphic to the group of binary vectors under addition mod 2: $\mathbb{F}^{|\mathcal{F}|}$ and $\mathbb{F}^{2n}$ respectively. Thus the matrix $\hat{\sigma}$ is a group homorphism from $\Gamma/C$ to $P/C$. Define the quotient maps $\phi: \Gamma \rightarrow \Gamma/C$ and $\psi: P \rightarrow P/C$. 

First we note that these maps commute in the way one expects, ie:
$$ \hat{\sigma} \circ \phi = \psi \circ \sigma .$$
This can be seen by noting that by construction
$$\forall f \in \mathcal{F}\;:\; \hat{\sigma} \circ \phi(f)= \psi \circ \sigma(f)$$
and, since $\sigma$, $\hat{\sigma}$, $\phi$ and $\psi$ are all group homomorphisms, for any other element $\delta \prod_{i} f^{b_i} \in \Gamma$, with $\delta \in C$:
$$ \hat{\sigma} \circ \phi( \delta \prod_{i} f^{b_i} )=\prod_{i} \hat{\sigma} \circ \phi( f^{b_i}) = \prod_{i} \psi \circ \sigma( f^{b_i})= \psi \circ \sigma(\delta \prod_{i} f^{b_i})$$

Next we argue that $\ker{\sigma} \simeq \ker{\hat{\sigma}}$. This can be seen by noting that
$$ \ker(\psi \circ \sigma) = \sigma^{-1}(C)= C \ker(\sigma)$$
and that
$$ \ker(\hat{\sigma}) \circ \phi)= \phi^{-1}(\ker(\hat{\sigma}))$$
so
$$\ker(\hat{\sigma})= \phi( C \ker(\sigma))$$ 
however $\ker(\sigma) \cap C=1$ and so $\ker{\sigma} \simeq \ker{\hat{\sigma}}.$

Finally since
$$ \forall x \in \ker(\hat{\sigma}) \;:\; \sigma(\Gamma(x))^* \Gamma(x) \in \ker(\sigma)$$
and 
$$ \forall x \neq y \in \ker(\hat{\sigma}) \;:\; \sigma(\Gamma(x))^* \Gamma(x)  \neq \sigma(\Gamma(y))^* \Gamma(y)$$
it follows by a counting argument that
$$\ker(\sigma) = \{\sigma(\Gamma(b))^* \Gamma(b) \vert b \in \ker(\hat{\sigma}) \}$$
\end{proof}

\begin{theorem}
$\ker(\tau) = \{\tau(\Gamma(b))^* \Gamma(b) \vert b \in \ker(\hat{\tau}) \}$
\end{theorem}

\begin{proof}
The proof of this theorem proceeds in precisely the same way as the previous theorem, replacing $\sigma$ with $\tau$ and $P$ with $F$.
\end{proof}

\section{Removing $-1$ from the Stabilizer group $S$}\label{sec:stab_group_minus}

Given a valid mapping $\sigma(f_i)$, we are always free to define a new valid mapping $\sigma(f_i)\rightarrow \delta_i \sigma(f_i)\;,\; \delta_i =\pm 1$. Here we explain how if the stabilizer group $S$ contains $-1$, one can always find a choice of signs using Gaussian elimination such that $S$ does not contain $-1$.

\begin{theorem}
Given a mapping $\sigma$, there always exists a choice of sign $\{\delta_i\}$ such that $-1 \not \in S$ 
\end{theorem}

\begin{proof}
Consider the subgroup of $S$ of all elements $\{e\} \subset \ker(\tau)$ such that $\sigma(e)=\pm1$. Consider a generating set of this subgroup $\{a_k\}$. Let $\vec{a}_k$ be a binary vector indicating which elements in $\mathcal{F}$ are included in $a_k$. Let $A$ be the binary matrix whose rows are the vectors $\vec{a}_k$. Let $\vec{b}$ be a binary vector indicating the signature of $\sigma$ on $a_k$, i.e. it has a $0$ at index $k$ if $\sigma(a_k)=1$ and a $1$ if $\sigma(a_k)=-1$. 

There exists a choice of signs $\{\delta_i\}$ such that $S$ does not contain $-1$ iff there exists a binary vector $\vec{\delta}$ such that:
$$ A \vec{\delta} = \vec{b} $$
It is not hard to see this by noting that if we can find such a $\vec{\delta}$, then applying a corresponding change of signs $\delta_i = (-1)^{\vec{\delta}_i}$ will map $\sigma(a_i) \rightarrow  (-1)^{\vec{\delta}\cdot \vec{a_i}} \sigma(a_i) = (-1)^{b_i} (-1)^{b_i}=1$

Finally we note that the rows of $A$ are linearly independent, since the elements $a_i$ are generators, so $A$ can always be put into reduced row echelon form:
$$\tilde{A} = \left[\mathbb{I} | * \right]$$
and so $ A \vec{\delta} = \vec{b} $ always admits a solution.
\end{proof}

\section{Computing the Kernel of Polynomial Matrices}\label{sec:kernel}

The set of Laurent polynomials over $\mathbb{F}_2$ forms a {\it ring} $R$, the set of vectors over this ring is formally referred to as an {\it $R$-module} -- denoted $R^m$ for $m$ dimensional vectors -- and the ``matrices'' acting on this module are called {\it $R$-module homomorphisms}. The ``kernel'' of a module homomorphism is called a {\it syzygy module}. 
\begin{definition}[Syzygy Module]
Let $\{ f_i \}_1^t$ be elements of an $R$-module. The set of all tuples $(a_1,...,a_t)$, where $a_i \in R$,  such that $\sum_i a_i f_i =0$ is a syzygy module. 
\end{definition}
Here the elements $f_i$ correspond to the columns of our polynomial matrix. The syzygy module is a {\it submodule } of $R^m$. Importantly, any syzygy module is generated by a finite set of elements if $R$ is finitely generated (which in our case it is). 

We've deliberately avoided using the language of modules throughout the body of the text to keep things accessible. However when discussing computation of the kernel of the polynomial matrices, it is important to engage with the formal distinction between matrices acting on vector spaces and module homomorphisms acting on modules. In particular, unlike the binary matrix case where the kernel can be computed by simple Gaussian elimination, the computation of the generating set of a syzygy module requires the computation of a Gr\"obner basis using Buchberger's algorithm, followed by some additional steps. The reason for this is that although one may apply Gaussian elimination to a polynomial matrix and compute elements of the syzygy module, these elements may not constitute a complete generating set. Here we briefly outline the procedure for computing a generating set of the syzygy module associated with a polynomial matrix. For more details we refer the reader to Chapter 5 of \cite{cox2005using} and Chapter 2 of \cite{cox2008ideals}.

To make things easier we begin by considering the subset of polynomials whose monomials have non-negative degree (which we call $\mathbb{F}^+_2[x_1,..,x_D]^{n}$). This allows us to more easily place an ordering convention on the monomials of the polynomial. Any matrix over Laurent polynomials may be translated so that it only contains elements of this type. It is not so difficult to see that if we find a set of generators for the syzygy module of the translated polynomial matrix, then under an inverse transformation this must be contained in the syzygy module of the original matrix. However it is not obvious that this captures a minimal generating set when we take the full ring of Laurent polynomials, since we have access to more operations. We will return to this.

The trick to computing a syzygy module of a module homomorphism (kernel of a polynomial matrix), is to treat the basis elements of the vector $e_i$ as variables in a polynomial. So for example the vector $v = (1+x, 1+y)$ may be treated as a polynomial: $$v(x,y,e_1,e_2) = (1+x)e_1 + (1+y)e_2= e_1 + x e_1 +e_2 + ye_2 .$$ So we transform $n$ dimensional vectors over polynomials in $\mathbb{F}^+_2[x_1,..,x_D]$ into polynomials in $\mathbb{F}^+_2[x_1,..,x_D, e_1,...,e_n]$. To compute the syzygy module of the original module homomorphism -- specified by a collection of polynomials $\{f_i \}$ in $\mathbb{F}^+_2[x_1,..,x_D, e_1,...,e_n]$ corresponding to the column vectors of our polynomial matrix -- we are interested in finding solutions $p_i$ to the equation:
\begin{equation}\label{eq:syzygy_solution}
 \sum_i p_i f_i = 0 \;,\; p_i \in \mathbb{F}^+_2(x_1,..,x_D).
 \end{equation}
 The set of such solutions constitutes a syzygy module of the original module homomorphism. We may leverage existing methods in algebraic geometry to find solutions to this equation. Standard methods using Gr\"obner bases typically have both $p_i$ and $f_i$ belonging to the same polynomial rings. However in our case there is an additional restriction on $p_i$. Thus we need slightly modified definitions for Gr\"obner bases, and for Buchberger's algorithm, which is used to compute Gr\"obner bases.
 
Let $R$ denote a polynomial ring, for example $\mathbb{F}^+_2[x_1,..,x_D]$. Let $R^m$ denote the $m$ dimensional module over $R$ expressed in terms of basis elements $\{e_i\}_1^m$, i.e. $f \in R^m$ has the form $f= \sum_j^m f_j e_j$ where $f_j \in R$. If $a$ and $b$ are monomials in $R$, then $a' = a e_i$ and $b' = b e_j$ are monomials in $R^m$, and:
\begin{itemize}
\item  $a'$ divides $b'$ iff $i=j$ and $a$ divides $b$. 
\item  The greatest common divisor of $a'$ and $b'$: $\textrm{GCD}(a',b')=\textrm{GCD}(a,b)e_i$.
\item  The least common multiple of $a'$ and $b'$: $\textrm{LCM}(a',b')= \textrm{LCM}(a,b)e_i$ if $i=j$ otherwise $\textrm{LCM}(a',b')=0$.
\end{itemize}
 Let $\textrm{LT}(g)$ denote the leading term of a polynomial $g \in R^m$ according to some predetermined ordering convention on the monomials, and $\textrm{LM}(g)$ the leading monomial -- for our purposes they are identical, but they are distinct for generic fields $R$.

\begin{definition}[Gr\"obner Basis of a submodule]
A submodule $M \subseteq R^m$ is a monomial submodule if it can be generated by monomials in $R^m$.
Given a submodule $M \subseteq R^m$, denote by $\langle \textrm{LT}(M) \rangle$ the monomial submodule generated by the leading terms of all $f \in M$. 
A finite collection $G =\{g_i\} \subset M$ is called a Gr\"obner basis of $M$ if $\langle LT(M) \rangle = \langle \{\textrm{LT}(g_i) \} \rangle$. 
\end{definition}
The Gr\"obner basis generates $M$: $M = \langle G \rangle$. 
There exists an algorithm for computing a Gr\"obner basis of ideals, called Buchberger's algorithm. Buchberger's algorithm can be used in exactly the same way to compute Gr\"obner Bases of submodules, however one needs to use the revised notions of divisibility, GCD and LCM given above. 

Equipped with the Gr\"obner Basis $G=\{g_i\}_1^s$ of a submodule $M=\langle \{f_i\}_1^t \rangle $, computing the syzygy module requires a few final steps. 
\begin{definition}[S-polynomial]
Given $f, g \in R^m$, let $L=\textrm{LCM}(\textrm{LM}(f),\textrm{LM}(g))$ the S-polynomial is defined as:
$$ S(f,g) :=  \frac{L}{\textrm{LT}(f)} f - \frac{L}{\textrm{LT}(g)}g$$
\end{definition}
\begin{theorem}[Buchberger's Criterion]
A set $G=\{g_i\} \subset R^m$ is a Gr\"obner basis if and only if there exists polynomials $a_{ijk} \in R$ such that $S(g_i,g_j)=\sum_k a_{ijk} g_k$ and $\textrm{LT}(a_{ijk}g_k)\leq \textrm{LT}(S(g_i,g_j))$ for all $i, j, k$.
\end{theorem}
Define
$$s_{ij}:= \frac{L_{ij}}{\textrm{LT}(g_i)} e_i - \frac{L_{ij}}{\textrm{LT}(g_j)}e_j - \sum_k a_{ijk} e_k$$
where $L_{ij}=\textrm{LCM}(\textrm{LM}(g_i),\textrm{LM}(g_j))$.
Let $\hat{F} = (f_1, ... , f_t)$ be an $m \times t$  polynomial matrix and $\hat{G}=(g_1,...,g_s)$ an $m \times s $ polynomial matrix. There exists a $t \times s$ polynomial matrix $\hat{B}$ and a $s \times t$ polynomial matrix $\hat{C}$ such that $\hat{G}=\hat{F}\hat{B}$ and $\hat{F}=\hat{G}\hat{C}$. Let $S_k$ be the columns of the matrix $I - BC$. 
\begin{proposition}[Syzygy Module]
$$\textrm{Syz}(f_1,...,f_t) = \langle B s_{ij}, S_k \rangle $$
\end{proposition}
Computing $B$, $C$ and $a_{ijk}$ can be done by application of a polynomial division algorithm which is a subroutine of Buchberger's algorithm.

The arguments outlined above apply to polynomials with non-negative degree. Here we argue that given a polynomial matrix $A$ with entries in a Laurent polynomial, the above methods should suffice to compute the kernel. First multiply the matrix $A$ by a translation $T_{\vec{\alpha}}=\prod_i x_i^{\alpha_i}$, where $\alpha_i$ is the smallest degree of $x_i$ in any monomial in $A$. The new matrix $A'=T_{\vec{\alpha}} A$ will only contain polynomials with non-negative degree. Suppose there exists a vector $f$ such that $Af = 0$. Let $T_{\vec{\beta}}$ be a translation where $\beta_i$ is the smallest degree of $x_i$ in any monomial in $f$ so that $f'=T_{\vec{\beta}}f$ also has non-negative degree. It is not hard to see that $f'$ is in the kernel of $A'$: $A'f'=T_{\vec{\alpha}+ \vec{\beta}} Af = 0$, and furthermore since $f'$ has non-negative degree it must be generated by the generators computed using the methods outlined above. Thus if we apply the above methods to compute the kernel of $A'$, $f'$ will be contained in that kernel. Thus everything in the kernel of $A$ can be found by computing the kernel of $A'$ using the methods outlined above, and then applying an appropriate translation.

\end{document}